\documentclass[acmsmall,screen]{acmart}

\InputIfFileExists{no-editing-marks}{
  \def\noeditingmarks{}
}

\def\longversion{}

\usepackage[mathletters]{ucs}
\usepackage[utf8x]{inputenc}
\DeclareUnicodeCharacter{8797}{\ensuremath{\stackrel{\scriptscriptstyle {\mathrm{def}}}{=}}}
\DeclareUnicodeCharacter{183}{\ensuremath{\cdot}} 
\makeatletter
\@ifundefined{lhs2tex.lhs2tex.sty.read}%
  {\@namedef{lhs2tex.lhs2tex.sty.read}{}%
   \newcommand\SkipToFmtEnd{}%
   \newcommand\EndFmtInput{}%
   \long\def\SkipToFmtEnd#1\EndFmtInput{}%
  }\SkipToFmtEnd

\newcommand\ReadOnlyOnce[1]{\@ifundefined{#1}{\@namedef{#1}{}}\SkipToFmtEnd}
\usepackage{amstext}
\usepackage{amssymb}
\usepackage{stmaryrd}
\DeclareFontFamily{OT1}{cmtex}{}
\DeclareFontShape{OT1}{cmtex}{m}{n}
  {<5><6><7><8>cmtex8
   <9>cmtex9
   <10><10.95><12><14.4><17.28><20.74><24.88>cmtex10}{}
\DeclareFontShape{OT1}{cmtex}{m}{it}
  {<-> ssub * cmtt/m/it}{}

\DeclareFontShape{OT1}{cmtt}{bx}{n}
  {<5><6><7><8>cmtt8
   <9>cmbtt9
   <10><10.95><12><14.4><17.28><20.74><24.88>cmbtt10}{}
\DeclareFontShape{OT1}{cmtex}{bx}{n}
  {<-> ssub * cmtt/bx/n}{}

\newcommand{\Conid}[1]{\mathit{#1}}
\newcommand{\Varid}[1]{\mathit{#1}}
\newcommand{\anonymous}{\kern0.06em \vbox{\hrule\@width.5em}}
\newcommand{\plus}{\mathbin{+\!\!\!+}}

\renewcommand{\geq}{\geqslant}
\usepackage{polytable}

\@ifundefined{mathindent}%
  {\newdimen\mathindent\mathindent\leftmargini}%
  {}%

\def\resethooks{%
  \global\let\SaveRestoreHook\empty
  \global\let\ColumnHook\empty}
\newcommand*{\savecolumns}[1][default]%
  {\g@addto@macro\SaveRestoreHook{\savecolumns[#1]}}
\newcommand*{\restorecolumns}[1][default]%
  {\g@addto@macro\SaveRestoreHook{\restorecolumns[#1]}}
\newcommand*{\aligncolumn}[2]%
  {\g@addto@macro\ColumnHook{\column{#1}{#2}}}

\resethooks

\newcommand{\onelinecommentchars}{\quad-{}- }
\newcommand{\commentbeginchars}{\enskip\{-}
\newcommand{\commentendchars}{-\}\enskip}

\newcommand{\visiblecomments}{%
  \let\onelinecomment=\onelinecommentchars
  \let\commentbegin=\commentbeginchars
  \let\commentend=\commentendchars}

\newcommand{\invisiblecomments}{%
  \let\onelinecomment=\empty
  \let\commentbegin=\empty
  \let\commentend=\empty}

\visiblecomments

\newlength{\blanklineskip}
\newcommand{\hsindent}[1]{\quad}
\let\hspre\empty
\let\hspost\empty

\EndFmtInput
\makeatother
\ReadOnlyOnce{polycode.fmt}%
\makeatletter

\newcommand{\hsnewpar}[1]%
  {{\parskip=0pt\parindent=0pt\par\vskip #1\noindent}}

\newcommand{\hscodestyle}{}

\newcommand{\sethscode}[1]%
  {\expandafter\let\expandafter\hscode\csname #1\endcsname
   \expandafter\let\expandafter\endhscode\csname end#1\endcsname}

  {\par\noindent
   \advance\leftskip\mathindent
   \hscodestyle
   \let\\=\@normalcr
   \let\hspre\(\let\hspost\)%
   \pboxed}%
  {\endpboxed\)%
   \par\noindent
   \ignorespacesafterend}

  {\hsnewpar\abovedisplayskip
   \advance\leftskip\mathindent
   \hscodestyle
   \let\hspre\(\let\hspost\)%
   \pboxed}%
  {\endpboxed%
   \hsnewpar\belowdisplayskip
   \ignorespacesafterend}

  {\hsnewpar\abovedisplayskip
   \advance\leftskip\mathindent
   \hscodestyle
   \let\\=\@normalcr
   \(\pboxed}%
  {\endpboxed\)%
   \hsnewpar\belowdisplayskip
   \ignorespacesafterend}

\newcommand{\plainhs}{\sethscode{plainhscode}}

\plainhs

  {\hsnewpar\abovedisplayskip
   \advance\leftskip\mathindent
   \hscodestyle
   \let\\=\@normalcr
   \(\parray}%
  {\endparray\)%
   \hsnewpar\belowdisplayskip
   \ignorespacesafterend}

  {\parray}{\endparray}

  {\(\parray}{\endparray\)}

\def\codeframewidth{\arrayrulewidth}
\RequirePackage{calc}

  {\parskip=\abovedisplayskip\par\noindent
   \hscodestyle
   \arrayrulewidth=\codeframewidth
   \tabular{@{}|p{\linewidth-2\arraycolsep-2\arrayrulewidth-2pt}|@{}}%
   \hline\framedhslinecorrect\\{-1.5ex}%
   \let\endoflinesave=\\
   \let\\=\@normalcr
   \(\pboxed}%
  {\endpboxed\)%
   \framedhslinecorrect\endoflinesave{.5ex}\hline
   \endtabular
   \parskip=\belowdisplayskip\par\noindent
   \ignorespacesafterend}

\newcommand{\framedhslinecorrect}[2]%
  {#1[#2]}

  {\(\def\column##1##2{}%
   \let\>\undefined\let\<\undefined\let\\\undefined
   \newcommand\>[1][]{}\newcommand\<[1][]{}\newcommand\\[1][]{}%
   \def\fromto##1##2##3{##3}%
   }{\) }%

  {\let\orighscode=\hscode
   \let\origendhscode=\endhscode
   \def\endhscode{\def\hscode{\endgroup\def\@currenvir{hscode}\\}\begingroup}
   \orighscode\def\hscode{\endgroup\def\@currenvir{hscode}}}%
  {\origendhscode
   \global\let\hscode=\orighscode
   \global\let\endhscode=\origendhscode}%

\makeatother
\EndFmtInput
\def\mathindent{1em} 
\renewcommand\Varid[1]{\mathord{\textsf{#1}}}
\renewcommand\Conid[1]{\mathord{\textsf{#1}}}

\usepackage{natbib}
\usepackage{graphicx}
\usepackage{grffile}
\usepackage{longtable}
\usepackage{wrapfig}
\usepackage{rotating}
\usepackage[normalem]{ulem}
\usepackage{amsmath}
\usepackage{amsthm}
\usepackage{textcomp}
\usepackage{amssymb}
\usepackage{cmll}
\usepackage{mathpartir}
\usepackage[plain]{fancyref}

\def\frefdefname{Def.}
\def\Frefdefname{Def.}
\def\freflemname{Lemma}
\def\Freflemname{Lemma}
\def\frefthmname{Theorem}
\def\Frefthmname{Theorem}
\def\frefappendixname{Appendix}
\def\Frefappendixname{Appendix}
\def\fancyrefdeflabelprefix{def}
\frefformat{plain}{\fancyrefdeflabelprefix}{{\frefdefname}\fancyrefdefaultspacing#1}
\Frefformat{plain}{\fancyrefdeflabelprefix}{{\Frefdefname}\fancyrefdefaultspacing#1}
\def\fancyreflemlabelprefix{lem}
\frefformat{plain}{\fancyreflemlabelprefix}{{\freflemname}\fancyrefdefaultspacing#1}
\Frefformat{plain}{\fancyreflemlabelprefix}{{\Freflemname}\fancyrefdefaultspacing#1}
\def\fancyrefthmlabelprefix{thm}
\frefformat{plain}{\fancyrefthmlabelprefix}{{\frefthmname}\fancyrefdefaultspacing#1}
\Frefformat{plain}{\fancyrefthmlabelprefix}{{\Frefthmname}\fancyrefdefaultspacing#1}
\def\fancyrefappendixlabelprefix{appendix}
\frefformat{plain}{\fancyrefappendixlabelprefix}{{\frefappendixname}\fancyrefdefaultspacing#1}
\Frefformat{plain}{\fancyrefappendixlabelprefix}{{\Frefappendixname}\fancyrefdefaultspacing#1}

\usepackage{xspace}
\newcommand{\ghc}{\textsc{ghc}\xspace}
\newcommand{\eg}{\textit{e.g.}\xspace}
\newcommand{\ie}{\textit{i.e.}\xspace}

\newcommand{\case}[3][]{\mathsf{case}_{#1} #2 \mathsf{of} \{#3\}^m_{k=1}}
\newcommand{\data}{\mathsf{data} }

\newcommand{\flet}[1][]{\mathsf{let}_{#1} }
\newcommand{\fin}{ \mathsf{in} }
\newcommand{\varid}[1]{\ensuremath{\Varid{#1}}}

\newcommand{\termsOf}[1]{\mathnormal{terms}(#1)}
\newcommand{\multiplicatedTypes}[1]{\bigotimes(#1)}
\newcommand{\ta}[2]{γ(#1)(#2)}

\newcommand{\figuresection}[1]{\par \addvspace{1em} \textbf{\sf #1}}

\usepackage{xargs}
\usepackage[colorinlistoftodos,prependcaption,textsize=tiny]{todonotes}
\ifx\noeditingmarks\undefined    
    \newcommand{\note}[1]{{\color{blue}{\begin{itemize} \item {#1} \end{itemize}}}}
    \newenvironment{alt}{\color{red}}{}

    \newcommandx{\unsure}[2][1=]{\todo[linecolor=red,backgroundcolor=red!25,bordercolor=red,#1]{#2}}
    \newcommandx{\info}[2][1=]{\todo[linecolor=green,backgroundcolor=green!25,bordercolor=green,#1]{#2}}
    \newcommandx{\change}[2][1=]{\todo[linecolor=blue,backgroundcolor=blue!25,bordercolor=blue,#1]{#2}}
    \newcommandx{\inconsistent}[2][1=]{\todo[linecolor=blue,backgroundcolor=blue!25,bordercolor=red,#1]{#2}}
    \newcommandx{\critical}[2][1=]{\todo[linecolor=blue,backgroundcolor=blue!25,bordercolor=red,#1]{#2}}
    \newcommand{\improvement}[1]{\todo[linecolor=pink,backgroundcolor=pink!25,bordercolor=pink]{#1}}
    \newcommandx{\resolved}[2][1=]{\todo[linecolor=OliveGreen,backgroundcolor=OliveGreen!25,bordercolor=OliveGreen,#1]{#2}} 

    \newcommandx{\rn}[1]{\todo[]{RRN: #1}} 
    \newcommandx{\simon}[1]{\todo[]{SPJ: #1}}
    \newcommandx{\jp}[1]{\todo[linecolor=blue,bordercolor=blue,backgroundcolor=cyan!10]{JP: #1}{}}
    \newcommand{\manuel}[1]{\todo[linecolor=purple,bordercolor=purple,backgroundcolor=blue!10]{Manuel: #1}{}}

\else

    \newcommand{\note}[1]{}

    \newcommand{\unsure}[2]{}
    \newcommand{\info}[2]{}
    \newcommand{\change}[2]{}
    \newcommand{\inconsistent}[2]{}
    \newcommand{\critical}[2]{}
    \newcommand{\improvement}[1]{}
    \newcommand{\resolved}[2]{}
    \newcommand{\rn}[1]{}
    \newcommand{\simon}[1]{}
    \newcommand{\jp}[1]{}
    \newcommand{\manuel}[1]{}
\fi

\newcommand\HaskeLL{Linear Haskell\xspace{}}
\newcommand\calc{{\ensuremath{λ^q_\to}}}

\usepackage{booktabs}   
\usepackage{subcaption} 

\setcopyright{acmlicensed}
\acmPrice{}
\acmDOI{10.1145/3158093}
\acmYear{2018}
\copyrightyear{2018}
\acmJournal{PACMPL}
\acmVolume{2}
\acmNumber{POPL}
\acmArticle{5}
\acmMonth{1}

\begin{document}

\title{Linear Haskell}       
\subtitle{Practical Linearity in a Higher-Order Polymorphic Language}                     


\author{Jean-Philippe Bernardy}
\affiliation{%
  \institution{University of Gothenburg}
  \department{Department of Philosophy, Linguistics and Theory of Science}
  \streetaddress{Olof Wijksgatan 6}
  \city{Gothenburg}
  \postcode{41255}
  \country{Sweden}}
\email{jean-philippe.bernardy@gu.se}
\author{Mathieu Boespflug}
\affiliation{%
  \institution{Tweag I/O}
  \city{Paris}
  \country{France}
}
\email{m@tweag.io}
\author{Ryan R. Newton}
\affiliation{%
  \institution{Indiana University}
  \city{Bloomington}
  \state{IN}
  \country{USA}
}
\email{rrnewton@indiana.edu}
\author{Simon Peyton Jones}
\affiliation{
  \institution{Microsoft Research}
  \city{Cambridge}
  \country{UK}
}
\email{simonpj@microsoft.com}
\author{Arnaud Spiwack}
\affiliation{
  \institution{Tweag I/O}
  \city{Paris}
  \country{France}
}
\email{arnaud.spiwack@tweag.io}

\renewcommand{\shortauthors}{J.-P. Bernardy, M. Boespflug, R. Newton,
  S. Peyton Jones, and A. Spiwack}


\begin{abstract}
  Linear type systems have a long and storied history, but not a clear
  path forward to integrate with existing languages such as OCaml or
  Haskell. In this paper, we study a linear type system
  designed with two crucial properties in mind:
  backwards-compatibility and code reuse across linear and non-linear
  users of a library. Only then can the benefits of linear types
  permeate conventional functional programming.  Rather than bifurcate
  types into linear and non-linear counterparts, we instead
  attach linearity to {\em function arrows}.  Linear functions can
  receive inputs from linearly-bound values, but can {\em also} operate over
  unrestricted, regular values.

  To demonstrate the efficacy of our linear type system~---~both how
  easy it can be integrated in an existing language implementation and
  how streamlined it makes it to write programs with linear
  types~---~we implemented our type system in 
  \textsc{ghc}, the leading Haskell compiler, and demonstrate
  two kinds of applications of linear types:
  mutable data with pure interfaces;
  and enforcing protocols in I/O-performing functions.
\end{abstract}

\ifx\draftmode\undefined    


\fi

\maketitle

\ifx\longversion\undefined{
}
\else{
  \it This paper appears in the Proceeding of the ACM Conference on Principles of Programming Languages (POPL) 2018.  This version includes an Appendix that gives an operational semantics for the core language, and proofs of the metatheoretical results stated in the paper.
}
\fi

\section{Introduction}
\label{sec:introduction}

Despite their obvious promise, and a huge research literature, linear
type systems have not made it into mainstream programming languages,
even though linearity has inspired uniqueness typing in Clean, and
ownership typing in Rust.  We take up this challenge by extending
Haskell with linear types.
Our design supports many applications for linear types, but we focus on
two particular use-cases.  First, safe
update-in-place for mutable structures, such as arrays; and second,
enforcing access protocols for external \textsc{api}s, such as files,
sockets, channels and other resources.  Our particular contributions
are these:
\begin{itemize}
\item We describe an extension to Haskell for linear types, dubbed \HaskeLL, using
      two extended examples (\fref{sec:consumed}-\fref{sec:io-protocols}).
      The extension is \emph{non-invasive}:
      existing programs continue to typecheck,
      and existing datatypes can be used as-is even in linear parts
      of the program.
      The key to this non-invasiveness is that, in contrast to most other
      approaches, we focus on \emph{linearity on the function arrow}
      rather than \emph{linearity in the kinds} (\fref{sec:lin-arrow}).
\item Every function arrow can be declared linear, including those of
      constructor types. This results in datatypes which can store
      both linear values, in addition to unrestricted ones
      (\fref{sec:datatypes}-\ref{sec:non-linear-constructors}).
\item A benefit of linearity-on-the-arrow is that it naturally supports
      \emph{linearity polymorphism} (\fref{sec:lin-poly}).  This contributes
      to a smooth extension of Haskell by allowing many existing functions
      (map, compose, etc) to be given more general types, so they can
      work uniformly in both linear and non-linear code.
\item We formalise our system in a small, statically-typed core
      calculus that exhibits all these features (\fref{sec:calculus}).
      It enjoys the usual properties of progress and preservation.
\item We have implemented a prototype of the system as a modest extension to \textsc{ghc}
      (\fref{sec:impl}), which substantiates our claim of non-invasiveness.
      We use this prototype to implement case-study applications (\fref{sec:applications}).
      Our prototype performs linearity \emph{inference}, but a systematic
      treatment of type inference for linearity in our system remains open.
\end{itemize}
Retrofits often involve compromise and ad-hoc choices, but in fact we
have found that, as well as fitting into Haskell, our design holds
together in its own right.  We hope that it may perhaps serve as a
template for similar work in other languages.  There is a rich
literature on linear type systems, as we discuss in a long related
work section (\fref{sec:related}).

\section{Motivation and intuitions}
\label{sec:programming-intro}

Informally, \emph{a function is ``linear'' if it consumes its argument exactly once}.
(It is ``affine'' if it consumes it at most once.)  A linear type system
gives a static guarantee that a claimed linear function really is linear.
There are many motivations for linear type systems, but here we focus on two of them:
\begin{itemize}
\item \emph{Is it safe to update this value in-place} (\fref{sec:freezing-arrays})?
That depends on whether there
are aliases to the value; update-in-place is \textsc{ok} if there are no other pointers to it.
Linearity supports a more efficient implementation, by $O(1)$ update rather than $O(n)$ copying.
\item \emph{Am I obeying the usage protocol of this external resource}
(\fref{sec:io-protocols})?
For example, an open file should be closed, and should not be used after it it has been closed;
a socket should be opened, then bound, and only then used for reading; a malloc'd memory
block should be freed, and should not be used after that.
Here, linearity does not affect efficiency, but rather eliminates many bugs.
\end{itemize}
We introduce our extension to Haskell, which we call \HaskeLL{}, by
focusing on these two use-cases.  In doing so, we introduce a number
of ideas that we flesh out in subsequent subsections.

\subsection{Operational intuitions}
\label{sec:consumed}

We have said informally that \emph{``a linear function consumes its argument
exactly once''}. But what exactly does that mean?

\begin{quote}
\emph{Meaning of the linear arrow}:
\ensuremath{\Varid{f}\mathbin{::}\Varid{s}\mathbin{⊸}\Varid{t}} guarantees that if \ensuremath{(\Varid{f}\;\Varid{u})} is consumed exactly once,
then the argument \ensuremath{\Varid{u}} is consumed exactly once.
\end{quote}
To make sense of this statement we need to know what ``consumed exactly once'' means.
Our definition is based on the type of the value concerned:
\begin{definition}[Consume exactly once]~ \label{def:consume}
\begin{itemize}
\item To consume a value of atomic base type (like \ensuremath{\Conid{Int}} or \ensuremath{\Conid{Ptr}}) exactly once, just evaluate it.
\item To consume a function value exactly once, apply it to one argument, and consume its result exactly once.
\item To consume a pair exactly once, pattern-match on it, and consume each component exactly once.
\item In general, to consume a value of an algebraic datatype exactly once, pattern-match on it,
  and consume all its linear components exactly once
  (\fref{sec:non-linear-constructors})\footnote{You may deduce that pairs have linear components,
    and indeed they do, as we discuss in \fref{sec:non-linear-constructors}.}.
\end{itemize}
\end{definition}
\noindent
This definition is enough to allow programmers to reason about the
typing of their functions, and it drives the formal typing judgements in
\fref{sec:statics}.

Note that a linear arrow specifies \emph{how the function uses its argument}. It does not
restrict \emph{the arguments to which the function can be applied}.
In particular, a linear function cannot assume that it is given the
unique pointer to its argument.  For example, if \ensuremath{\Varid{f}\mathbin{::}\Varid{s}\mathbin{⊸}\Varid{t}}, then\
this is fine:
\begin{hscode}\SaveRestoreHook
\column{B}{@{}>{\hspre}l<{\hspost}@{}}%
\column{E}{@{}>{\hspre}l<{\hspost}@{}}%
\>[B]{}\Varid{g}\mathbin{::}\Varid{s}\to \Varid{t}{}\<[E]%
\\
\>[B]{}\Varid{g}\;\Varid{x}\mathrel{=}\Varid{f}\;\Varid{x}{}\<[E]%
\ColumnHook
\end{hscode}\resethooks
The type of \ensuremath{\Varid{g}} makes no particular guarantees about the way in which it uses \ensuremath{\Varid{x}};
in particular, \ensuremath{\Varid{g}} can pass that argument to \ensuremath{\Varid{f}}.

\subsection{Safe mutable arrays}
\label{sec:freezing-arrays}

The Haskell language provides immutable arrays, built with the function \ensuremath{\Varid{array}}\footnote{
Haskell actually generalises over the type of array indices, but for this
paper we will assume that the arrays are indexed, from 0, by \ensuremath{\Conid{Int}} indices.}:
\begin{hscode}\SaveRestoreHook
\column{B}{@{}>{\hspre}l<{\hspost}@{}}%
\column{E}{@{}>{\hspre}l<{\hspost}@{}}%
\>[B]{}\Varid{array}\mathbin{::}\Conid{Int}\to [\mskip1.5mu (\Conid{Int},\Varid{a})\mskip1.5mu]\to \Conid{Array}\;\Varid{a}{}\<[E]%
\ColumnHook
\end{hscode}\resethooks

\begin{wrapfigure}[7]{r}[0pt]{7.0cm} 
\vspace{-8mm}
\begin{hscode}\SaveRestoreHook
\column{B}{@{}>{\hspre}l<{\hspost}@{}}%
\column{3}{@{}>{\hspre}l<{\hspost}@{}}%
\column{9}{@{}>{\hspre}l<{\hspost}@{}}%
\column{13}{@{}>{\hspre}l<{\hspost}@{}}%
\column{E}{@{}>{\hspre}l<{\hspost}@{}}%
\>[3]{}\mathbf{type}\;\Conid{MArray}\;\Varid{s}\;\Varid{a}{}\<[E]%
\\
\>[3]{}\mathbf{type}\;\Conid{Array}\;\Varid{a}{}\<[E]%
\\[\blanklineskip]%
\>[3]{}\Varid{newMArray}\mathbin{::}\Conid{Int}\to \Conid{ST}\;\Varid{s}\;(\Conid{MArray}\;\Varid{s}\;\Varid{a}){}\<[E]%
\\
\>[3]{}\Varid{read}{}\<[9]%
\>[9]{}\mathbin{::}\Conid{MArray}\;\Varid{s}\;\Varid{a}\to \Conid{Int}\to \Conid{ST}\;\Varid{s}\;\Varid{a}{}\<[E]%
\\
\>[3]{}\Varid{write}\mathbin{::}\Conid{MArray}\;\Varid{s}\;\Varid{a}\to (\Conid{Int},\Varid{a})\to \Conid{ST}\;\Varid{s}\;(){}\<[E]%
\\
\>[3]{}\Varid{unsafeFreeze}\mathbin{::}\Conid{MArray}\;\Varid{s}\;\Varid{a}\to \Conid{ST}\;\Varid{s}\;(\Conid{Array}\;\Varid{a}){}\<[E]%
\\[\blanklineskip]%
\>[3]{}\varid{forM}\_\mathbin{::}{}\<[13]%
\>[13]{}\Conid{Monad}\;\Varid{m}\Rightarrow {}\<[E]%
\\
\>[13]{}[\mskip1.5mu \Varid{a}\mskip1.5mu]\to (\Varid{a}\to \Varid{m}\;())\to \Varid{m}\;(){}\<[E]%
\\
\>[3]{}\Varid{runST}\mathbin{::}(∀\Varid{s}. \Conid{ST}\;\Varid{s}\;\Varid{a})\to \Varid{a}{}\<[E]%
\ColumnHook
\end{hscode}\resethooks
\vspace{-4mm}
\caption{Signatures for array primitives (current \textsc{ghc})}
{\vspace{-2.5mm}\hfill}
\label{fig:array-sigs}
\end{wrapfigure}
But how is \ensuremath{\Varid{array}} implemented? A possible answer is ``it is built-in; don't ask''.
But in reality \textsc{ghc} implements \ensuremath{\Varid{array}} using more primitive pieces, so that library authors
can readily implement more complex variations --- and they certainly do: see for example \fref{sec:cursors}.  Here is the
definition of \ensuremath{\Varid{array}}, using library functions whose types are given
in \fref{fig:array-sigs}.

\begin{hscode}\SaveRestoreHook
\column{B}{@{}>{\hspre}l<{\hspost}@{}}%
\column{3}{@{}>{\hspre}l<{\hspost}@{}}%
\column{8}{@{}>{\hspre}l<{\hspost}@{}}%
\column{E}{@{}>{\hspre}l<{\hspost}@{}}%
\>[B]{}\Varid{array}\mathbin{::}\Conid{Int}\to [\mskip1.5mu (\Conid{Int},\Varid{a})\mskip1.5mu]\to \Conid{Array}\;\Varid{a}{}\<[E]%
\\
\>[B]{}\Varid{array}\;\Varid{size}\;\Varid{pairs}\mathrel{=}\Varid{runST}{}\<[E]%
\\
\>[B]{}\hsindent{3}{}\<[3]%
\>[3]{}(\mathbf{do}\;{}\<[8]%
\>[8]{}\{\mskip1.5mu \Varid{ma}\leftarrow \Varid{newMArray}\;\Varid{size}{}\<[E]%
\\
\>[8]{};\varid{forM}\_\;\Varid{pairs}\;(\Varid{write}\;\Varid{ma}){}\<[E]%
\\
\>[8]{};\Varid{unsafeFreeze}\;\Varid{ma}\mskip1.5mu\}){}\<[E]%
\ColumnHook
\end{hscode}\resethooks
In the first line we allocate a mutable array, of type \ensuremath{\Conid{MArray}\;\Varid{s}\;\Varid{a}}.
Then we iterate over the \ensuremath{\Varid{pairs}}, with \ensuremath{\varid{forM}\_}, updating the array in place
for each pair.  Finally, we freeze the mutable array, returning an immutable
array as required.  All this is done in the \ensuremath{\Conid{ST}} monad, using \ensuremath{\Varid{runST}} to
securely encapsulate an imperative algorithm in a purely-functional context,
as described in \cite{launchbury_st_1995}.

Why is \ensuremath{\Varid{unsafeFreeze}} unsafe?  The result of \ensuremath{(\Varid{unsafeFreeze}\;\Varid{ma})} is a new
immutable array, but to avoid an unnecessary copy,
it is actually \ensuremath{\Varid{ma}}.  The intention is, of course, that
that \ensuremath{\Varid{unsafeFreeze}} should be the last use of the mutable array; but
nothing stops us continuing to mutate it further, with quite undefined semantics.
The ``unsafe'' in the function name is a \textsc{ghc} convention meaning ``the programmer
has a proof obligation here that the compiler cannot check''.

The other unsatisfactory thing about the monadic approach to array
construction is that it is overly sequential. Suppose you had a pair of
mutable arrays, with some updates to perform to each; these updates could
be done in {\em parallel}, but the \ensuremath{\Conid{ST}} monad would serialise them.

Linear types allow a more secure and less sequential interface.  \HaskeLL{}
introduces a new kind of function type: the \emph{linear arrow} \ensuremath{\Varid{a}\mathbin{⊸}\Varid{b}}. A linear
function \ensuremath{\Varid{f}\mathbin{::}\Varid{a}\mathbin{⊸}\Varid{b}} must consume its argument \emph{exactly once}.  This new
arrow is used in a new array \textsc{api}, given in
\fref{fig:linear-array-sigs}.
\begin{figure}[h]
\hfill
\vspace{-2mm}
\begin{hscode}\SaveRestoreHook
\column{B}{@{}>{\hspre}l<{\hspost}@{}}%
\column{3}{@{}>{\hspre}l<{\hspost}@{}}%
\column{E}{@{}>{\hspre}l<{\hspost}@{}}%
\>[3]{}\mathbf{type}\;\Conid{MArray}\;\Varid{a}{}\<[E]%
\\
\>[3]{}\mathbf{type}\;\Conid{Array}\;\Varid{a}{}\<[E]%
\\[\blanklineskip]%
\>[3]{}\Varid{newMArray}\mathbin{::}\Conid{Int}\to (\Conid{MArray}\;\Varid{a}\mathbin{⊸}\Conid{Unrestricted}\;\Varid{b})\mathbin{⊸}\Varid{b}{}\<[E]%
\\
\>[3]{}\Varid{write}\mathbin{::}\Conid{MArray}\;\Varid{a}\mathbin{⊸}(\Conid{Int},\Varid{a})\to \Conid{MArray}\;\Varid{a}{}\<[E]%
\\
\>[3]{}\Varid{read}\mathbin{::}\Conid{MArray}\;\Varid{a}\mathbin{⊸}\Conid{Int}\to (\Conid{MArray}\;\Varid{a},\Conid{Unrestricted}\;\Varid{a}){}\<[E]%
\\
\>[3]{}\Varid{freeze}\mathbin{::}\Conid{MArray}\;\Varid{a}\mathbin{⊸}\Conid{Unrestricted}\;(\Conid{Array}\;\Varid{a}){}\<[E]%
\ColumnHook
\end{hscode}\resethooks
\vspace{-5mm}
\caption{Type signatures for array primitives (linear version), allowing
  in-place update.}
{\vspace{-2.5mm}\hfill\vspace{-2.5mm}}
\label{fig:linear-array-sigs}
\end{figure}

\noindent
Using this \textsc{api} we can define \ensuremath{\Varid{array}} thus:
\begin{hscode}\SaveRestoreHook
\column{B}{@{}>{\hspre}l<{\hspost}@{}}%
\column{E}{@{}>{\hspre}l<{\hspost}@{}}%
\>[B]{}\Varid{array}\mathbin{::}\Conid{Int}\to [\mskip1.5mu (\Conid{Int},\Varid{a})\mskip1.5mu]\to \Conid{Array}\;\Varid{a}{}\<[E]%
\\
\>[B]{}\Varid{array}\;\Varid{size}\;\Varid{pairs}\mathrel{=}\Varid{newMArray}\;\Varid{size}\;(\lambda \Varid{ma}\to \Varid{freeze}\;(\Varid{foldl}\;\Varid{write}\;\Varid{ma}\;\Varid{pairs})){}\<[E]%
\ColumnHook
\end{hscode}\resethooks
There are several things to note here:
\begin{itemize}

\item The function \ensuremath{\Varid{newMArray}} allocates a fresh, mutable array of the
  specified size, and passes it to the function supplied as the second
  argument to \ensuremath{\Varid{newMArray}}, as a {\em linear} value \ensuremath{\Varid{ma}}.

\item Even though linearity is a property of function arrows, not of
  types (\fref{sec:lin-arrow}), we still disinguish the type of
  mutable arrays \ensuremath{\Conid{MArray}} from that of immutable arrays \ensuremath{\Conid{Array}}, because in this {\sc api}
  only immutable arrays are {\em allowed} to be non-linear (unrestricted).  The
  way to say that results can be freely shared is to use
  \ensuremath{\Conid{Unrestricted}} (our version of linear logic's \ensuremath{\mathbin{!}} modality, see \fref{sec:datatypes}), as in the type of \ensuremath{\Varid{freeze}}.

\item Because \ensuremath{\Varid{freeze}} consumes its input, there is no danger of the same
  mutable array being subsequently written to, eliminating the problem with
  \ensuremath{\Varid{unsafeFreeze}}.
  
\item Since \ensuremath{\Varid{ma}} is linear, we can only use it once. Thus each call to
  \ensuremath{\Varid{write}} returns a (logically) new array, so that the array is single-threaded,
  by \ensuremath{\Varid{foldl}}, through the sequence of writes.
  
\item Above, \ensuremath{\Varid{foldl}} has the type \ensuremath{(\Varid{a}\mathbin{⊸}\Varid{b}\mathbin{⊸}\Varid{a})\to \Varid{a}\mathbin{⊸}[\mskip1.5mu \Varid{b}\mskip1.5mu]\mathbin{⊸}\Varid{a}},
which expresses that it consumes its second argument linearly
(the mutable array), while the function it is given as
its first argument (\ensuremath{\Varid{write}}) must be linear.
As we shall see in \fref{sec:lin-poly} this is not a new \ensuremath{\Varid{foldl}}, but
an instance of a more general, multiplicity-polymorphic version of
a single \ensuremath{\Varid{foldl}} (where ``multiplicity'' refers to how many times a function
consumes its input).
\end{itemize}

Three factors ensure that a unique \ensuremath{\Conid{MArray}} is needed
in any given application \ensuremath{\Varid{x}\mathrel{=}\Varid{newMArray}\;\Varid{k}}, and in turn that
update-in-place is safe. First, \ensuremath{\Varid{newMArray}} introduces {\em only}
a linear \ensuremath{\Varid{ma}\mathbin{::}\Conid{MArray}\;\Varid{a}}.  Second, no function that
consumes an \ensuremath{\Conid{MArray}\;\Varid{a}} returns more than a {\em single} pointer to it;
so \ensuremath{\Varid{k}} can never obtain two pointers to \ensuremath{\Varid{ma}}.  Third, \ensuremath{\Varid{k}} must wrap its
result in \ensuremath{\Conid{Unrestricted}}. This third point means that even if \ensuremath{\Varid{x}} is used in an
unrestricted way, it suffices to call \ensuremath{\Varid{k}} a single time to obtain the
result, and in turn no mutable pointer to \ensuremath{\Varid{ma}} can {\em escape} when
\ensuremath{\Varid{newArray}} returns (\ie{} when the \ensuremath{\Varid{b}} result of \ensuremath{\Varid{newArray}} is
evaluated).

With this mutable array \textsc{api},
the \ensuremath{\Conid{ST}} monad has disappeared altogether; it is the array \emph{itself}
that must be single threaded, not the operations of a monad. That removes
the unnecessary sequentialisation we mentioned earlier and opens the
possibility of exploiting more parallelism at runtime.

Compared to the \textit{status quo} (using \ensuremath{\Conid{ST}} and \ensuremath{\Varid{unsafeFreeze}}), the
other major benefit
is in shrinking the trusted code base, because more library code (and it
can be particularly gnarly code) is statically typechecked.  Clients who
use only {\em immutable} arrays do not see the inner workings of the library, and will
be unaffected. Of course, the functions of
\fref{fig:linear-array-sigs} are still part of the trusted code base,
in particular, they must not lie about linearity. Our second use-case has a much more direct impact on library clients.

\subsection{I/O protocols} \label{sec:io-protocols}

\begin{figure}
\begin{minipage}{0.40 \textwidth} 
\begin{hscode}\SaveRestoreHook
\column{B}{@{}>{\hspre}l<{\hspost}@{}}%
\column{3}{@{}>{\hspre}l<{\hspost}@{}}%
\column{E}{@{}>{\hspre}l<{\hspost}@{}}%
\>[3]{}\mathbf{type}\;\Conid{File}{}\<[E]%
\\
\>[3]{}\Varid{openFile}\mathbin{::}\Conid{FilePath}\to \Conid{IO}\;\Conid{File}{}\<[E]%
\\
\>[3]{}\Varid{readLine}\mathbin{::}\Conid{File}\to \Conid{IO}\;\Conid{ByteString}{}\<[E]%
\\
\>[3]{}\Varid{closeFile}\mathbin{::}\Conid{File}\to \Conid{IO}\;(){}\<[E]%
\ColumnHook
\end{hscode}\resethooks
\vspace{-7mm}
\caption{Types for traditional file IO} \label{fig:io-traditional}
\end{minipage}%
\begin{minipage}{0.60 \textwidth}
\begin{hscode}\SaveRestoreHook
\column{B}{@{}>{\hspre}l<{\hspost}@{}}%
\column{3}{@{}>{\hspre}l<{\hspost}@{}}%
\column{E}{@{}>{\hspre}l<{\hspost}@{}}%
\>[3]{}\mathbf{type}\;\Conid{File}{}\<[E]%
\\
\>[3]{}\Varid{openFile}\mathbin{::}\Conid{FilePath}\to \varid{IO}_{\varid{L}}\;\mathrm{1}\;\Conid{File}{}\<[E]%
\\
\>[3]{}\Varid{readLine}\mathbin{::}\Conid{File}\mathbin{⊸}\varid{IO}_{\varid{L}}\;\mathrm{1}\;(\Conid{File},\Conid{Unrestricted}\;\Conid{ByteString}){}\<[E]%
\\
\>[3]{}\Varid{closeFile}\mathbin{::}\Conid{File}\mathbin{⊸}\varid{IO}_{\varid{L}}\;\omega\;(){}\<[E]%
\ColumnHook
\end{hscode}\resethooks
\vspace{-7mm}
\caption{Types for linear file IO} \label{fig:io-linear}
\end{minipage}
\vspace{-2mm}
\hfill
\end{figure}

Consider the \textsc{api} for files in \fref{fig:io-traditional}, where a
\ensuremath{\Conid{File}} is a cursor in a physical file.
Each call
to \ensuremath{\Varid{readLine}} returns a \ensuremath{\Conid{ByteString}} (the line) and moves the cursor one line
forward.  But nothing stops us reading a file after we have closed it,
or forgetting to close it.
An alternative \textsc{api} using linear types is given in \fref{fig:io-linear}.
Using it we can write a simple file handling program, \ensuremath{\Varid{firstLine}}, shown here.
%
\begin{hscode}\SaveRestoreHook
\column{B}{@{}>{\hspre}l<{\hspost}@{}}%
\column{3}{@{}>{\hspre}l<{\hspost}@{}}%
\column{7}{@{}>{\hspre}l<{\hspost}@{}}%
\column{E}{@{}>{\hspre}l<{\hspost}@{}}%
\>[B]{}\Varid{firstLine}\mathbin{::}\Conid{FilePath}\to \varid{IO}_{\varid{L}}\;\omega\;\Conid{Bytestring}{}\<[E]%
\\
\>[B]{}\Varid{firstLine}\;\Varid{fp}\mathrel{=}{}\<[E]%
\\
\>[B]{}\hsindent{3}{}\<[3]%
\>[3]{}\mathbf{do}\;{}\<[7]%
\>[7]{}\{\mskip1.5mu \Varid{f}\leftarrow \Varid{openFile}\;\Varid{fp}{}\<[E]%
\\
\>[7]{};(\Varid{f},\Conid{Unrestricted}\;\Varid{bs})\leftarrow \Varid{readLine}\;\Varid{f}{}\<[E]%
\\
\>[7]{};\Varid{closeFile}\;\Varid{f}{}\<[E]%
\\
\>[7]{};\Varid{return}\;\Varid{bs}\mskip1.5mu\}{}\<[E]%
\ColumnHook
\end{hscode}\resethooks
%
Notice several things:
\begin{itemize}

\item Operations on files remain monadic, unlike the case with mutable arrays.
I/O operations affect the world, and hence must be sequenced.  It is not enough
to sequence operations on files individually, as it was for arrays.

\item We generalise the \ensuremath{\Conid{IO}} monad so that it expresses whether or not the
returned value is linear.  We add an extra {\em multiplicity} type parameter \ensuremath{\Varid{p}} to the monad \ensuremath{\varid{IO}_{\varid{L}}},
where \ensuremath{\Varid{p}} can be \ensuremath{\mathrm{1}} or \ensuremath{\omega}, indicating a linear or unrestricted result, respectively.
Now \ensuremath{\Varid{openFile}} returns \ensuremath{\varid{IO}_{\varid{L}}\;\mathrm{1}\;\Conid{File}},
the ``\ensuremath{\mathrm{1}}'' indicating that the returned \ensuremath{\Conid{File}} must be used linearly.
We will return to how \ensuremath{\varid{IO}_{\varid{L}}} is defined in \fref{sec:linear-io}.

\item As before, operations on linear values must consume their input
and return a new one; here \ensuremath{\Varid{readLine}} consumes the \ensuremath{\Conid{File}} and produces a new one.

\item Unlike the \ensuremath{\Conid{File}}, the \ensuremath{\Conid{ByteString}} returned by \ensuremath{\Varid{readLine}} is unrestricted,
and the type of \ensuremath{\Varid{readLine}} indicates this.

\end{itemize}
It may seem tiresome to have to thread the \ensuremath{\Conid{File}} as well as sequence
operations with the \ensuremath{\Conid{IO}} monad. But in fact it is often useful do
to do so, because we can use types
to witness the state of the resource, \eg, with separate
types for an open or closed \ensuremath{\Conid{File}}. We show applications in \fref{sec:cursors} and \fref{sec:sockets}.

\subsection{Linear datatypes}
\label{sec:linear-constructors}
\label{sec:datatypes}

With the above intutions in mind, what type should we assign to a data
constructor such as the pairing constructor \ensuremath{(,)}?  Here are two possibilities:
\begin{center}
\vspace{-1mm}
\begin{hscode}\SaveRestoreHook
\column{B}{@{}>{\hspre}l<{\hspost}@{}}%
\column{10}{@{}>{\hspre}l<{\hspost}@{}}%
\column{26}{@{}>{\hspre}l<{\hspost}@{}}%
\column{41}{@{}>{\hspre}l<{\hspost}@{}}%
\column{49}{@{}>{\hspre}l<{\hspost}@{}}%
\column{E}{@{}>{\hspre}l<{\hspost}@{}}%
\>[B]{}(,)\mathbin{::}{}\<[10]%
\>[10]{}\Varid{a}\mathbin{⊸}\Varid{b}\mathbin{⊸}(\Varid{a},\Varid{b})\;{}\<[26]%
\>[26]{}\hspace{2cm}\;{}\<[41]%
\>[41]{}(,)\mathbin{::}{}\<[49]%
\>[49]{}\Varid{a}\to \Varid{b}\to (\Varid{a},\Varid{b}){}\<[E]%
\ColumnHook
\end{hscode}\resethooks
\vspace{-4mm}
\end{center}
Using the definition in \fref{sec:consumed}, the former is clearly the correct
choice: if the result of \ensuremath{(,)\;\Varid{e1}\;\Varid{e2}} is consumed exactly once,
then (by \fref{def:consume}),
\ensuremath{\Varid{e1}} and \ensuremath{\Varid{e2}} are each consumed exactly once; and hence \ensuremath{(,)} is linear it its
arguments.

\begin{wrapfigure}[5]{r}[0pt]{5cm} 
\vspace{-6mm}
\begin{hscode}\SaveRestoreHook
\column{B}{@{}>{\hspre}l<{\hspost}@{}}%
\column{E}{@{}>{\hspre}l<{\hspost}@{}}%
\>[B]{}\Varid{f1}\mathbin{::}(\Conid{Int},\Conid{Int})\to (\Conid{Int},\Conid{Int}){}\<[E]%
\\
\>[B]{}\Varid{f1}\;\Varid{x}\mathrel{=}\mathbf{case}\;\Varid{x}\;\mathbf{of}\;(\Varid{a},\Varid{b})\to (\Varid{a},\Varid{a}){}\<[E]%
\\[\blanklineskip]%
\>[B]{}\Varid{f2}\mathbin{::}(\Conid{Int},\Conid{Int})\mathbin{⊸}(\Conid{Int},\Conid{Int}){}\<[E]%
\\
\>[B]{}\Varid{f2}\;\Varid{x}\mathrel{=}\mathbf{case}\;\Varid{x}\;\mathbf{of}\;(\Varid{a},\Varid{b})\to (\Varid{b},\Varid{a}){}\<[E]%
\ColumnHook
\end{hscode}\resethooks
\end{wrapfigure}
So much for construction; what about pattern matching?  Consider
\ensuremath{\Varid{f1}} and \ensuremath{\Varid{f2}} defined here;
\ensuremath{\Varid{f1}} is an ordinary Haskell function. Even though the data constructor \ensuremath{(,)} has
a linear type, that does \emph{not} imply that the pattern-bound variables \ensuremath{\Varid{a}} and \ensuremath{\Varid{b}}
must be consumed exactly once; and indeed they are not.
Therefore, \ensuremath{\Varid{f1}} does not have the linear type \ensuremath{(\Conid{Int},\Conid{Int})\mathbin{⊸}(\Conid{Int},\Conid{Int})}.
Why not?  If the result of \ensuremath{(\Varid{f1}\;\Varid{t})} is consumed once, is \ensuremath{\Varid{t}} guaranteed to be consumed
once?  No: \ensuremath{\Varid{t}} is guaranteed to be evaluated once, but its first component is then
consumed twice and its second component not at all, contradicting \Fref{def:consume}.
In contrast, \ensuremath{\Varid{f2}} \emph{does} have a linear type: if \ensuremath{(\Varid{f2}\;\Varid{t})} is consumed exactly once,
then indeed \ensuremath{\Varid{t}} is consumed exactly once.
The key point here is that \emph{the same pair constructor works in both functions;
we do not need a special non-linear pair}.

The same idea applies to all existing Haskell datatypes: in
\HaskeLL{} we treat all datatypes defined using legacy Haskell-98
(non-\textsc{gadt}) syntax as defining constructors with linear arrows.
For example here is a declaration of \HaskeLL{}'s list type,
whose constructor \ensuremath{(\mathbin{:})} uses linear arrows:

\begin{wrapfigure}[8]{r}[0pt]{4.5cm}\vspace{-2mm}
\begin{hscode}\SaveRestoreHook
\column{B}{@{}>{\hspre}l<{\hspost}@{}}%
\column{E}{@{}>{\hspre}l<{\hspost}@{}}%
\>[B]{}\mathbf{data}\;[\mskip1.5mu \Varid{a}\mskip1.5mu]\mathrel{=}[\mskip1.5mu \mskip1.5mu]\mid \Varid{a}\mathbin{:}[\mskip1.5mu \Varid{a}\mskip1.5mu]{}\<[E]%
\ColumnHook
\end{hscode}\resethooks
\begin{hscode}\SaveRestoreHook
\column{B}{@{}>{\hspre}l<{\hspost}@{}}%
\column{9}{@{}>{\hspre}l<{\hspost}@{}}%
\column{E}{@{}>{\hspre}l<{\hspost}@{}}%
\>[B]{}(\plus )\mathbin{::}[\mskip1.5mu \Varid{a}\mskip1.5mu]\mathbin{⊸}[\mskip1.5mu \Varid{a}\mskip1.5mu]\mathbin{⊸}[\mskip1.5mu \Varid{a}\mskip1.5mu]{}\<[E]%
\\
\>[B]{}[\mskip1.5mu \mskip1.5mu]{}\<[9]%
\>[9]{}\plus \Varid{ys}\mathrel{=}\Varid{ys}{}\<[E]%
\\
\>[B]{}(\Varid{x}\mathbin{:}\Varid{xs}){}\<[9]%
\>[9]{}\plus \Varid{ys}\mathrel{=}\Varid{x}\mathbin{:}(\Varid{xs}\plus \Varid{ys}){}\<[E]%
\ColumnHook
\end{hscode}\resethooks
\end{wrapfigure}
Just as with pairs, this is not a new, linear list type: this \emph{is}
\HaskeLL{}'s list type, and all existing Haskell functions will work
over it perfectly well.
Even better, many list-based functions are in fact linear, and
can be given a more precise type. For example we can write \ensuremath{(\plus )} as
follows:

This type says that if \ensuremath{(\Varid{xs}\plus \Varid{ys})} is consumed exactly once, then
\ensuremath{\Varid{xs}} is consumed exactly once, and so is \ensuremath{\Varid{ys}}, and indeed our type
system will accept this definition.

As before, giving a more precise type to \ensuremath{(\plus )} only {\em strengthens} the
contract that \ensuremath{(\plus )} offers to its callers; \emph{it does not restrict
  its usage}. For example:
\begin{hscode}\SaveRestoreHook
\column{B}{@{}>{\hspre}l<{\hspost}@{}}%
\column{3}{@{}>{\hspre}l<{\hspost}@{}}%
\column{E}{@{}>{\hspre}l<{\hspost}@{}}%
\>[3]{}\Varid{sum}\mathbin{::}[\mskip1.5mu \Conid{Int}\mskip1.5mu]\mathbin{⊸}\Conid{Int}{}\<[E]%
\\
\>[3]{}\Varid{f}\mathbin{::}[\mskip1.5mu \Conid{Int}\mskip1.5mu]\mathbin{⊸}[\mskip1.5mu \Conid{Int}\mskip1.5mu]\to \Conid{Int}{}\<[E]%
\\
\>[3]{}\Varid{f}\;\Varid{xs}\;\Varid{ys}\mathrel{=}\Varid{sum}\;(\Varid{xs}\plus \Varid{ys})\mathbin{+}\Varid{sum}\;\Varid{ys}{}\<[E]%
\ColumnHook
\end{hscode}\resethooks
Here the two arguments to \ensuremath{(\plus )} have different multiplicities, but
the function \ensuremath{\Varid{f}} guarantees that it will consume \ensuremath{\Varid{xs}} exactly once if
\ensuremath{(\Varid{f}\;\Varid{xs}\;\Varid{ys})} is consumed exactly once.

For an existing language, being able to strengthen \ensuremath{(\plus )}, and similar
functions, in a {\em backwards-compatible} way is a huge boon.  Of
course, not all functions are linear: a function may legitimately
demand unrestricted input.  For example, the function \ensuremath{\Varid{f}} above
consumes \ensuremath{\Varid{ys}} twice, and so
\ensuremath{\Varid{f}} needs an unrestricted arrow for that argument.
\label{sec:compatibility}

Finally, we can use the very same pairs and lists
types to contain linear values (such as mutable arrays) without
compromising safety.  For example:
\begin{hscode}\SaveRestoreHook
\column{B}{@{}>{\hspre}l<{\hspost}@{}}%
\column{17}{@{}>{\hspre}l<{\hspost}@{}}%
\column{30}{@{}>{\hspre}l<{\hspost}@{}}%
\column{E}{@{}>{\hspre}l<{\hspost}@{}}%
\>[B]{}\Varid{upd}\mathbin{::}(\Conid{MArray}\;\Conid{Char},\Conid{MArray}\;\Conid{Char})\mathbin{⊸}\Conid{Int}\to (\Conid{MArray}\;\Conid{Char},\Conid{MArray}\;\Conid{Char}){}\<[E]%
\\
\>[B]{}\Varid{upd}\;(\Varid{a1},\Varid{a2})\;\Varid{n}{}\<[17]%
\>[17]{}\mid \Varid{n}\geq \mathrm{10}{}\<[30]%
\>[30]{}\mathrel{=}(\Varid{write}\;\Varid{a1}\;\Varid{n}\;\text{\tt 'x'},\Varid{a2}){}\<[E]%
\\
\>[17]{}\mid \Varid{otherwise}{}\<[30]%
\>[30]{}\mathrel{=}(\Varid{write}\;\Varid{a2}\;\Varid{n}\;\text{\tt 'o'},\Varid{a1}){}\<[E]%
\ColumnHook
\end{hscode}\resethooks

\subsection{Unrestricted data constructors}
\label{sec:non-linear-constructors}

Suppose we want to pass a linear \ensuremath{\Conid{MArray}} and an unrestricted \ensuremath{\Conid{Int}} to a function \ensuremath{\Varid{f}}.
We could give \ensuremath{\Varid{f}} the signature \ensuremath{\Varid{f}\mathbin{::}\Conid{MArray}\;\Conid{Int}\mathbin{⊸}\Conid{Int}\to \Conid{MArray}\;\Conid{Int}}.  But suppose
we wanted to uncurry the function; we could then give it the type
\begin{hscode}\SaveRestoreHook
\column{B}{@{}>{\hspre}l<{\hspost}@{}}%
\column{3}{@{}>{\hspre}l<{\hspost}@{}}%
\column{29}{@{}>{\hspre}l<{\hspost}@{}}%
\column{E}{@{}>{\hspre}l<{\hspost}@{}}%
\>[3]{}\Varid{f}\mathbin{::}(\Conid{MArray}\;\Conid{Int},\Conid{Int})\mathbin{⊸}{}\<[29]%
\>[29]{}\Conid{MArray}\;\Conid{Int}{}\<[E]%
\ColumnHook
\end{hscode}\resethooks
But this is no good: now \ensuremath{\Varid{f}} is only allowed to use the \ensuremath{\Conid{Int}} linearly, but it
might actually use it many times.  For this reason it is extremely useful to be
able to declare data constructors with non-linear types, like this:
\begin{hscode}\SaveRestoreHook
\column{B}{@{}>{\hspre}l<{\hspost}@{}}%
\column{3}{@{}>{\hspre}l<{\hspost}@{}}%
\column{32}{@{}>{\hspre}l<{\hspost}@{}}%
\column{E}{@{}>{\hspre}l<{\hspost}@{}}%
\>[3]{}\mathbf{data}\;\Conid{PLU}\;\Varid{a}\;\Varid{b}\;\mathbf{where}\;\{\mskip1.5mu \Conid{PLU}\mathbin{::}\Varid{a}\mathbin{⊸}\Varid{b}\to \Conid{PLU}\;\Varid{a}\;\Varid{b}\mskip1.5mu\}{}\<[E]%
\\[\blanklineskip]%
\>[3]{}\Varid{f}\mathbin{::}\Conid{PLU}\;(\Conid{MArray}\;\Conid{Int})\;\Conid{Int}\mathbin{⊸}{}\<[32]%
\>[32]{}\Conid{MArray}\;\Conid{Int}{}\<[E]%
\ColumnHook
\end{hscode}\resethooks
Here we use \textsc{gadt}-style syntax to give an explicit type signature to the data
constructor \ensuremath{\Conid{PLU}}, with mixed linearity.
Now, when \emph{constructing} a \ensuremath{\Conid{PLU}} pair the type of the constructor means
that we must always supply an unrestricted second argument; and dually
when \emph{pattern-matching} on \ensuremath{\Conid{PLU}} we are therefore free to use the second argument
in an unrestricted way, even if the \ensuremath{\Conid{PLU}} value itself is linear.

Instead of defining a pair with mixed linearity, we can also write
\begin{hscode}\SaveRestoreHook
\column{B}{@{}>{\hspre}l<{\hspost}@{}}%
\column{3}{@{}>{\hspre}l<{\hspost}@{}}%
\column{42}{@{}>{\hspre}l<{\hspost}@{}}%
\column{E}{@{}>{\hspre}l<{\hspost}@{}}%
\>[3]{}\mathbf{data}\;\Conid{Unrestricted}\;\Varid{a}\;\mathbf{where}\;\{\mskip1.5mu \Conid{Unrestricted}\mathbin{::}\Varid{a}\mathbin{→}\Conid{Unrestricted}\;\Varid{a}\mskip1.5mu\}{}\<[E]%
\\[\blanklineskip]%
\>[3]{}\Varid{f}\mathbin{::}(\Conid{MArray}\;\Conid{Int},\Conid{Unrestricted}\;\Conid{Int})\mathbin{⊸}{}\<[42]%
\>[42]{}\Conid{MArray}\;\Conid{Int}{}\<[E]%
\ColumnHook
\end{hscode}\resethooks
The type \ensuremath{(\Conid{Unrestricted}\;\Varid{t})} is very much like ``\ensuremath{\mathbin{!}\Varid{t}}'' in linear
logic, but in our setting it is just an ordinary user-defined datatype.
We saw it used in \fref{fig:linear-array-sigs}, where the result of \ensuremath{\Varid{read}} was
a pair of a linear \ensuremath{\Conid{MArray}} and an unrestricted array element:
\begin{hscode}\SaveRestoreHook
\column{B}{@{}>{\hspre}l<{\hspost}@{}}%
\column{3}{@{}>{\hspre}l<{\hspost}@{}}%
\column{E}{@{}>{\hspre}l<{\hspost}@{}}%
\>[3]{}\Varid{read}\mathbin{::}\Conid{MArray}\;\Varid{a}\mathbin{⊸}\Conid{Int}\to (\Conid{MArray}\;\Varid{a},\Conid{Unrestricted}\;\Varid{a}){}\<[E]%
\ColumnHook
\end{hscode}\resethooks
Note that, according to the definition in \fref{sec:consumed},
if a value of type \ensuremath{(\Conid{Unrestricted}\;\Varid{t})} is consumed exactly once,
that tells us nothing about how the argument of the data constructor is consumed:
it may be consumed many times or not at all.

\subsection{Multiplicity polymorphism}
\label{sec:lin-poly}
A linear function provides more guarantees to its caller than
a non-linear one~---~it is more general.  But the higher-order
case thickens the plot. Consider the standard \ensuremath{\Varid{map}} function over
(linear) lists:
\begin{hscode}\SaveRestoreHook
\column{B}{@{}>{\hspre}l<{\hspost}@{}}%
\column{15}{@{}>{\hspre}l<{\hspost}@{}}%
\column{E}{@{}>{\hspre}l<{\hspost}@{}}%
\>[B]{}\Varid{map}\;\Varid{f}\;[\mskip1.5mu \mskip1.5mu]{}\<[15]%
\>[15]{}\mathrel{=}[\mskip1.5mu \mskip1.5mu]{}\<[E]%
\\
\>[B]{}\Varid{map}\;\Varid{f}\;(\Varid{x}\mathbin{:}\Varid{xs}){}\<[15]%
\>[15]{}\mathrel{=}\Varid{f}\;\Varid{x}\mathbin{:}\Varid{map}\;\Varid{f}\;\Varid{xs}{}\<[E]%
\ColumnHook
\end{hscode}\resethooks
It can be given the two following incomparable types:
  \ensuremath{(\Varid{a}\mathbin{⊸}\Varid{b})\to [\mskip1.5mu \Varid{a}\mskip1.5mu]\mathbin{⊸}[\mskip1.5mu \Varid{b}\mskip1.5mu]}  and
  \ensuremath{(\Varid{a}\to \Varid{b})\to [\mskip1.5mu \Varid{a}\mskip1.5mu]\to [\mskip1.5mu \Varid{b}\mskip1.5mu]}.
  Thus, \HaskeLL{} features quantification over multiplicities and
  parameterised arrows (\ensuremath{\Conid{A}\mathbin{→}_{\Varid{q}}\;\Conid{B}}).  Using these, \ensuremath{\Varid{map}} can be given
  the following more general type: \ensuremath{\mathbin{∀}\Varid{p}. (\Varid{a}\to _{\Varid{p}}\;\Varid{b})\to [\mskip1.5mu \Varid{a}\mskip1.5mu]\to _{\Varid{p}}\;[\mskip1.5mu \Varid{b}\mskip1.5mu]}.
Likewise, function composition and \ensuremath{\Varid{foldl}} (\textit{cf.} \Fref{sec:freezing-arrays})
can be given the following general types:
\begin{hscode}\SaveRestoreHook
\column{B}{@{}>{\hspre}l<{\hspost}@{}}%
\column{40}{@{}>{\hspre}l<{\hspost}@{}}%
\column{E}{@{}>{\hspre}l<{\hspost}@{}}%
\>[B]{}\Varid{foldl}\mathbin{::}∀\Varid{p}\;\Varid{q}. (\Varid{a}\mathbin{→}_{\Varid{p}}\;\Varid{b}\mathbin{→}_{\Varid{q}}\;{}\<[40]%
\>[40]{}\Varid{a})\to \Varid{a}\mathbin{→}_{\Varid{p}}\;[\mskip1.5mu \Varid{b}\mskip1.5mu]\mathbin{→}_{\Varid{q}}\;\Varid{a}{}\<[E]%
\\[\blanklineskip]%
\>[B]{}(\mathbin{∘})\mathbin{::}∀\Varid{p}\;\Varid{q}. (\Varid{b}\mathbin{→}_{\Varid{p}}\;\Varid{c})\mathbin{⊸}(\Varid{a}\mathbin{→}_{\Varid{q}}\;\Varid{b})\mathbin{→}_{\Varid{p}}\;\Varid{a}\mathbin{→}_{\Varid{p}\mathbin{·}\Varid{q}}\;\Varid{c}{}\<[E]%
\\
\>[B]{}(\Varid{f}\mathbin{∘}\Varid{g})\;\Varid{x}\mathrel{=}\Varid{f}\;(\Varid{g}\;\Varid{x}){}\<[E]%
\ColumnHook
\end{hscode}\resethooks
The type of \ensuremath{(\mathbin{∘})} says that two functions that accept arguments of arbitrary
multiplicities (\ensuremath{\Varid{p}} and \ensuremath{\Varid{q}} respectively) can be composed to form a
function accepting arguments of multiplicity \ensuremath{\Varid{p}\mathbin{·}\Varid{q}} (\ie the
product of \ensuremath{\Varid{p}} and \ensuremath{\Varid{q}} --- see \fref{def:equiv-multiplicity}).
Finally, from a backwards-compatibility perspective, all of these
subscripts and binders for multiplicity polymorphism can be
ignored. Indeed, in a context where client code does not use
linearity, all inputs will have unlimited multiplicity, $ω$, and transitively all
expressions can be promoted to $ω$. Thus in such a context the
compiler, or indeed documentation tools, can even altogether hide
linearity annotations from the programmer when this language
extension is not turned on.

\subsection{Linear input/output} \label{sec:linear-io}

In \fref{sec:io-protocols} we introduced the \ensuremath{\varid{IO}_{\varid{L}}}
monad.\footnote{\ensuremath{\varid{IO}_{\varid{L}}\;\Varid{p}} is not a monad in the strict sense, because \ensuremath{\Varid{p}} and
  \ensuremath{\Varid{q}} can be different in \ensuremath{\varid{bind}_{\varid{IO}_{\varid{L}}}}. However it is a relative
  monad~\cite{altenkirch_monads_2010}. The details, involving the
  functor \ensuremath{\mathbf{data}\;\Conid{Mult}\;\Varid{p}\;\Varid{a}\;\mathbf{where}\;\{\mskip1.5mu \Conid{Mult}\mathbin{::}\Varid{a}\to _{\Varid{p}}\;\Conid{Mult}\;\Varid{p}\;\Varid{a}\mskip1.5mu\}} and linear
  arrows, are left as an exercise to the reader.}
But how does it work?  \ensuremath{\varid{IO}_{\varid{L}}}
is just a generalisation of the \ensuremath{\Conid{IO}} monad, thus:
\begin{hscode}\SaveRestoreHook
\column{B}{@{}>{\hspre}l<{\hspost}@{}}%
\column{3}{@{}>{\hspre}l<{\hspost}@{}}%
\column{13}{@{}>{\hspre}l<{\hspost}@{}}%
\column{E}{@{}>{\hspre}l<{\hspost}@{}}%
\>[3]{}\mathbf{type}\;\varid{IO}_{\varid{L}}\;\Varid{p}\;\Varid{a}{}\<[E]%
\\
\>[3]{}\varid{return}_{\varid{IO}_{\varid{L}}}\mathbin{::}\Varid{a}\to _{\Varid{p}}\;\varid{IO}_{\varid{L}}\;\Varid{p}\;\Varid{a}{}\<[E]%
\\
\>[3]{}\varid{bind}_{\varid{IO}_{\varid{L}}}{}\<[13]%
\>[13]{}\mathbin{::}\varid{IO}_{\varid{L}}\;\Varid{p}\;\Varid{a}\mathbin{⊸}(\Varid{a}\to _{\Varid{p}}\;\varid{IO}_{\varid{L}}\;\Varid{q}\;\Varid{b})\mathbin{⊸}\varid{IO}_{\varid{L}}\;\Varid{q}\;\Varid{b}{}\<[E]%
\ColumnHook
\end{hscode}\resethooks
The idea is that if \ensuremath{\Varid{m}\mathbin{::}\varid{IO}_{\varid{L}}\;\mathrm{1}\;\Varid{t}}, then \ensuremath{\Varid{m}} is an input/output
computation that returns a linear value of type \ensuremath{\Varid{t}}.  But what does it mean to
``return a linear value'' in a world where linearity applies only to
function arrows?  Fortunately, in the world of monads each computation
has an explicit continuation, so we just need to control the linearity of
the continuation arrow.  More precisely, in an application \ensuremath{\Varid{m}~`\varid{bind}_{\varid{IO}_{\varid{L}}}\!`\,{}~\Varid{k}}
where \ensuremath{\Varid{m}\mathbin{::}\varid{IO}_{\varid{L}}\;\mathrm{1}\;\Varid{t}}, we need the continuation \ensuremath{\Varid{k}} to be linear, \ensuremath{\Varid{k}\mathbin{::}\Varid{t}\mathbin{⊸}\varid{IO}_{\varid{L}}\;\Varid{q}\;\Varid{t'}}.
And that is captured by the multiplicity-polymorphic type of \ensuremath{\varid{bind}_{\varid{IO}_{\varid{L}}}}.

Even though they have a different type than usual, the bind and return
combinators of \ensuremath{\varid{IO}_{\varid{L}}} can be used in the familiar way. The difference
with the usual monad is that multiplicities may be mixed, but this
poses no problem in practice.  Consider
\begin{hscode}\SaveRestoreHook
\column{B}{@{}>{\hspre}l<{\hspost}@{}}%
\column{3}{@{}>{\hspre}l<{\hspost}@{}}%
\column{5}{@{}>{\hspre}l<{\hspost}@{}}%
\column{8}{@{}>{\hspre}l<{\hspost}@{}}%
\column{17}{@{}>{\hspre}l<{\hspost}@{}}%
\column{49}{@{}>{\hspre}l<{\hspost}@{}}%
\column{E}{@{}>{\hspre}l<{\hspost}@{}}%
\>[3]{}\Varid{printHandle}\mathbin{::}\Conid{File}\mathbin{⊸}\Conid{IO}\;\omega\;(){}\<[E]%
\\
\>[3]{}\Varid{printHandle}\;\Varid{f}\mathrel{=}\mathbf{do}{}\<[E]%
\\
\>[3]{}\hsindent{2}{}\<[5]%
\>[5]{}\{\mskip1.5mu (\Varid{f},\Conid{Unrestricted}\;\Varid{b})\leftarrow \Varid{atEOF}\;\Varid{f}{}\<[49]%
\>[49]{}\mbox{\onelinecomment  \ensuremath{\Varid{atEOF}\mathbin{::}\Conid{File}\mathbin{⊸}\varid{IO}_{\varid{L}}\;\mathrm{1}\;(\Conid{File},\Conid{Unrestricted}\;\Conid{Bool})}}{}\<[E]%
\\
\>[3]{}\hsindent{2}{}\<[5]%
\>[5]{};{}\<[8]%
\>[8]{}\mathbf{if}\;\Varid{b}\;\mathbf{then}\;\Varid{closeFile}\;\Varid{f}{}\<[49]%
\>[49]{}\mbox{\onelinecomment  \ensuremath{\Varid{closeFile}\mathbin{::}\Conid{File}\mathbin{⊸}\varid{IO}_{\varid{L}}\;\omega\;()}}{}\<[E]%
\\
\>[8]{}\mathbf{else}\;\mathbf{do}\;{}\<[17]%
\>[17]{}\{\mskip1.5mu (\Varid{f},\Conid{Unrestricted}\;\Varid{c})\leftarrow \Varid{read}\;\Varid{f}{}\<[49]%
\>[49]{}\mbox{\onelinecomment  \ensuremath{\Varid{read}\mathbin{::}\Conid{File}\mathbin{⊸}\varid{IO}_{\varid{L}}\;\mathrm{1}\;(\Conid{File},\Conid{Unrestricted}\;\Conid{Char})}}{}\<[E]%
\\
\>[17]{};\Varid{putChar}\;\Varid{c}{}\<[49]%
\>[49]{}\mbox{\onelinecomment  \ensuremath{\Varid{putChar}\mathbin{::}\Conid{Char}\to \varid{IO}_{\varid{L}}\;\omega\;()}}{}\<[E]%
\\
\>[17]{};\Varid{printHandle}\;\Varid{f}\mskip1.5mu\}\mskip1.5mu\}{}\<[E]%
\ColumnHook
\end{hscode}\resethooks
Here \ensuremath{\Varid{atEOF}} and \ensuremath{\Varid{read}} return a linear \ensuremath{\Conid{File}} that should be closed,
but \ensuremath{\Varid{close}} and \ensuremath{\Varid{putChar}} return an ordinary non-linear \ensuremath{()}.  So this
sequence of operations has mixed linearity.  Nevertheless, we can
interpret the \ensuremath{\mathbf{do}}-notation with \ensuremath{\varid{bind}_{\varid{IO}_{\varid{L}}}} in the usual way:
\begin{hscode}\SaveRestoreHook
\column{B}{@{}>{\hspre}l<{\hspost}@{}}%
\column{3}{@{}>{\hspre}l<{\hspost}@{}}%
\column{14}{@{}>{\hspre}l<{\hspost}@{}}%
\column{E}{@{}>{\hspre}l<{\hspost}@{}}%
\>[3]{}\Varid{read}\;\Varid{f}{}\<[14]%
\>[14]{}~`\varid{bind}_{\varid{IO}_{\varid{L}}}\!`\,{}~\lambda (\Varid{f},\Conid{Unrestricted}\;\Varid{c})\to {}\<[E]%
\\
\>[3]{}\Varid{putChar}\;\Varid{c}{}\<[14]%
\>[14]{}~`\varid{bind}_{\varid{IO}_{\varid{L}}}\!`\,{}~\lambda \_\to \mathbin{…}{}\<[E]%
\ColumnHook
\end{hscode}\resethooks
Such an interpretation of the \ensuremath{\mathbf{do}}-notation requires Haskell's
\texttt{-XRebindableSyntax} extension, but if linear I/O becomes
commonplace it would be worth considering a more robust solution.

Internally, hidden from clients, \textsc{ghc} actually implements \ensuremath{\Conid{IO}} as a function,
and that implementation too is illuminated by linearity, like so:
\begin{hscode}\SaveRestoreHook
\column{B}{@{}>{\hspre}l<{\hspost}@{}}%
\column{3}{@{}>{\hspre}l<{\hspost}@{}}%
\column{11}{@{}>{\hspre}l<{\hspost}@{}}%
\column{34}{@{}>{\hspre}l<{\hspost}@{}}%
\column{36}{@{}>{\hspre}l<{\hspost}@{}}%
\column{E}{@{}>{\hspre}l<{\hspost}@{}}%
\>[B]{}\mathbf{data}\;\Conid{World}{}\<[E]%
\\
\>[B]{}\mathbf{newtype}\;\varid{IO}_{\varid{L}}\;\Varid{p}\;\Varid{a}\mathrel{=}\varid{IO}_{\varid{L}}\;\{\mskip1.5mu \varid{unIO}_{\varid{L}}\mathbin{::}\Conid{World}\mathbin{⊸}\Conid{IORes}\;\Varid{p}\;\Varid{a}\mskip1.5mu\}{}\<[E]%
\\
\>[B]{}\mathbf{data}\;\Conid{IORes}\;\Varid{p}\;\Varid{a}\;\mathbf{where}{}\<[E]%
\\
\>[B]{}\hsindent{3}{}\<[3]%
\>[3]{}\Conid{IOR}\mathbin{::}\Conid{World}\mathbin{⊸}\Varid{a}\to _{\Varid{p}}\;\Conid{IORes}\;\Varid{p}\;\Varid{a}{}\<[E]%
\\[\blanklineskip]%
\>[B]{}\varid{bind}_{\varid{IO}_{\varid{L}}}{}\<[11]%
\>[11]{}\mathbin{::}\varid{IO}_{\varid{L}}\;\Varid{p}\;\Varid{a}\mathbin{⊸}(\Varid{a}\to _{\Varid{p}}\;\varid{IO}_{\varid{L}}\;\Varid{q}\;\Varid{b})\mathbin{⊸}\varid{IO}_{\varid{L}}\;\Varid{q}\;\Varid{b}{}\<[E]%
\\
\>[B]{}\varid{bind}_{\varid{IO}_{\varid{L}}}\;(\varid{IO}_{\varid{L}}\;\Varid{m})\;\Varid{k}\mathrel{=}\varid{IO}_{\varid{L}}\;(\lambda \Varid{w}\to {}\<[34]%
\>[34]{}\mathbf{case}\;\Varid{m}\;\Varid{w}\;\mathbf{of}{}\<[E]%
\\
\>[34]{}\hsindent{2}{}\<[36]%
\>[36]{}\Conid{IOR}\;\Varid{w'}\;\Varid{r}\to \varid{unIO}_{\varid{L}}\;(\Varid{k}\;\Varid{r})\;\Varid{w'}){}\<[E]%
\ColumnHook
\end{hscode}\resethooks
A value of type \ensuremath{\Conid{World}} represents the state of the world, and is
threaded linearly through I/O computations.  The linearity of the
result of the computation is captured by the \ensuremath{\Varid{p}} parameter of \ensuremath{\varid{IO}_{\varid{L}}},
which is inherited by the specialised form of pair, \ensuremath{\Conid{IORes}} that an
\ensuremath{\varid{IO}_{\varid{L}}} computation returns.  All linearity checks are verified by the
compiler, further reducing the size of the trusted code base.
\subsection{Linearity and strictness}

It is tempting to assume that, since a linear function consumes its
argument exactly once, then it must also be strict.  But not so!
For example
\begin{hscode}\SaveRestoreHook
\column{B}{@{}>{\hspre}l<{\hspost}@{}}%
\column{E}{@{}>{\hspre}l<{\hspost}@{}}%
\>[B]{}\Varid{f}\mathbin{::}\Varid{a}\mathbin{⊸}(\Varid{a},\Conid{Bool}){}\<[E]%
\\
\>[B]{}\Varid{f}\;\Varid{x}\mathrel{=}(\Varid{x},\Conid{True}){}\<[E]%
\ColumnHook
\end{hscode}\resethooks
Here \ensuremath{\Varid{f}} is certainly linear according to \fref{sec:consumed}, and
given the type of \ensuremath{(,)} in \fref{sec:linear-constructors}. That is, if \ensuremath{(\Varid{f}\;\Varid{x})}
is consumed exactly once, then each component of its result pair is
consumed exactly once, and hence \ensuremath{\Varid{x}} is consumed exactly once.
But \ensuremath{\Varid{f}} is certainly not strict: \ensuremath{\Varid{f}\;\bot } is not \ensuremath{\bot }.

\section{\calc{}: a core calculus for \HaskeLL}
\label{sec:statics}
\label{sec:calculus}

We do not formalise all of \HaskeLL{}, but rather a core calculus,
\calc{} which exhibits all key features, including datatypes and
multiplicity polymorphism.  This way we make precise much of the
informal discussion above.

\subsection{Syntax}
\newcommand{\pip}{\kern 1.18em | }
\label{sec:syntax}

\begin{figure}
  \begin{minipage}{0.4 \textwidth} \centering
  \figuresection{Multiplicities}
  \begin{align*}
    π,μ &::= 1 ~|~ ω ~|~ p ~|~ π+μ ~|~ π·μ
  \end{align*}
  \end{minipage}
  \begin{minipage}{0.4 \textwidth} \centering
  \figuresection{Types}
  \begin{align*}
  A,B ::= A →_π B ~|~  ∀p. A ~|~ D~p_1~…~p_n
  \end{align*}
  \end{minipage}
\\[3mm]
  \begin{minipage}{0.3 \textwidth} \centering
  \figuresection{Contexts}
  \begin{align*}
    Γ,Δ & ::=  (x :_{μ} A), Γ ~|~ –
  \end{align*}
  \end{minipage}
  \begin{minipage}{0.6\linewidth} \centering
    \figuresection{Datatype declaration}
    \begin{align*}
      \data D~p_1~…~p_n~\mathsf{where} \left(c_k : A₁ →_{π₁} …    A_{n_k} →_{π_{n_k}} D\right)^m_{k=1}
    \end{align*}
  \end{minipage}

  \figuresection{Terms}
  \begin{align*}
    e,s,t,u & ::= x & \text{variable} \\
            & \pip λ_π (x{:}A). t & \text{abstraction} \\
            & \pip t s & \text{application} \\
            & \pip λp. t & \text{multiplicity abstraction} \\
            & \pip t π & \text{multiplicity application} \\
            & \pip c t₁ … t_n & \text{data construction} \\
            & \pip \case[π] t {c_k  x₁ … x_{n_k} → u_k}  & \text{case} \\
            & \pip \flet[π] x_1 : A₁ = t₁ … x_n : A_n = t_n \fin u & \text{let}
  \end{align*}

  \caption{Syntax of \calc{}}
  \label{fig:syntax}
  \label{fig:contexts}
\end{figure}

The term syntax of \calc{} is that of a type-annotated (\textit{à la}
Church) simply-typed $λ$-calculus with let-definitions
(\fref{fig:syntax}).  It includes multiplicity polymorphism, but to avoid clutter
we omit ordinary type polymorphism.

\calc{} is an explicitly-typed language: each binder is annotated with
its type and multiplicity; and multiplicity abstraction and application
are explicit.  \HaskeLL{} will use type inference to fill in
much of this information, but we do not address the challenges of type
inference here.

The types of \calc{} (see \fref{fig:syntax}) are simple types with
arrows (albeit multiplicity-annotated ones), datatypes, and
multiplicity polymorphism.
We use the following abbreviations:
\(A → B ≝  A →_ω B\) and
\(A ⊸ B ≝ A →_1 B\).
Datatype declarations (see \fref{fig:syntax}) are of the following form:
\begin{align*}
  \data D~p_1~…~p_n~\mathsf{where} \left(c_k : A₁ →_{π₁} ⋯    A_{n_k} →_{π_{n_k}} D\right)^m_{k=1}
\end{align*}
The above declaration means that \(D\) is parameterized over \(n\) multiplicities $p_i$ and has \(m\) constructors \(c_k\),
each with \(n_k\) arguments. Arguments of
constructors have a multiplicity, just like arguments of functions: an
argument of multiplicity $ω$ means that consuming the data constructor once
makes no claim on how often that argument is consumed (\fref{def:consume}).
All the variables in the multiplicities $π_i$ must be among
$p_1…p_n$; we write $π[π_1…π_n]$ for the substitution of $p_i$ by
$π_i$ in $π$.

\subsection{Static semantics}
\label{sec:typing-contexts}

\newcommand{\apprule}{\inferrule{Γ ⊢ t :  A →_π B  \\   Δ ⊢ u : A}{Γ+πΔ ⊢ t u  :  B}\text{app}}
\newcommand{\varrule}{\inferrule{ }{ωΓ + x :_1 A ⊢ x : A}\text{var}}
\newcommand{\caserule}{\inferrule{Γ   ⊢  t  : D~π_1~…~π_n \\ Δ, x₁:_{πμ_i[π_1…π_n]} A_i, …,
      x_{n_k}:_{πμ_{n_k}[π_1…π_n]} A_{n_k} ⊢ u_k : C \\
      \text{for each $c_k : A_1 →_{μ_1} … →_{μ_{n_k-1}} A_{n_k} →_{μ_{n_k}} D~p_1~…~p_n$}}
    {πΓ+Δ ⊢ \case[π] t {c_k  x₁ … x_{n_k} → u_k} : C}\text{case}}

\begin{figure}
  \begin{mathpar}
    \varrule

    \inferrule{Γ, x :_{π} A  ⊢   t : B}
    {Γ ⊢ λ_π (x{:}A). t  :  A  →_q  B}\text{abs}

    \apprule

    \inferrule{Δ_i ⊢ t_i : A_i \\ \text {$c_k : A_1 →_{μ_1} … →_{μ_{n-1}}
        A_n →_{μ_n} D~p_1~…~p_n$ constructor}}
    {ωΓ+\sum_i μ_i[π₁…π_n]Δ_i ⊢ c_k  t₁ … t_n : D~π₁~…~π_n}\text{con}

    \caserule

    \inferrule{Γ_i   ⊢  t_i  : A_i  \\ Δ, x₁:_{π} A₁ …  x_n:_{π} A_n ⊢ u : C }
    { Δ+π\sum_i Γ_i ⊢ \flet[π] x_1 : A_1 = t₁  …  x_n : A_n = t_n  \fin u : C}\text{let}

    \inferrule{Γ ⊢  t : A \\ \text {$p$ fresh for $Γ$}}
    {Γ ⊢ λp. t : ∀p. A}\text{m.abs}

    \inferrule{Γ ⊢ t :  ∀p. A}
    {Γ ⊢ t π  :  A[π/p]}\text{m.app}
  \end{mathpar}

  \caption{Typing rules}
  \label{fig:typing}
\end{figure}

The static semantics of \calc{} is given in \fref{fig:typing}.  Each
binding in $Γ$, of form \(x :_π A\), includes a multiplicity $π$
(\fref{fig:syntax}).  The familiar judgement \(Γ ⊢ t : A\) should
be read as follows
\begin{quote}
 \(Γ ⊢ t : A\) asserts that consuming the term $t : A$ exactly once will
  consume each binding $(x :_{π} A)$ in $Γ$ with its multiplicity $π$.
\end{quote}
One may want to think of the \emph{types} in $Γ$ as
inputs of the judgement, and the \emph{multiplicities} as outputs.

The rule (abs) for lambda abstraction adds $(x :_{π} A)$ to the
environment $Γ$ before checking the body \ensuremath{\Varid{t}} of the abstraction.
Notice that in \calc{}, the lambda abstraction  $λ_π(x{:}A). t$
is explicitly annotated with multiplicity $π$.  Remember, this
is an explicitly-typed intermediate language; in \HaskeLL{}
this multiplicity is inferred.

The dual application rule (app) is more interesting:
$$\apprule$$
To consume \ensuremath{(\Varid{t}\;\Varid{u})} once, we consume \ensuremath{\Varid{t}} once, yielding the
multiplicities in $Γ$, and \ensuremath{\Varid{u}} once, yielding the multiplicies in
$\Delta$.  But if the multiplicity $π$ on \ensuremath{\Varid{u}}'s function arrow is $ω$,
then the function consumes its argument not once but $ω$ times, so all
\ensuremath{\Varid{u}}'s free variables must also be used with multiplicity $ω$. We
express this by {\em scaling} the multiplicities in $\Delta$ by $π$.
Finally we need to add together all the
multiplicities in $Γ$ and $π\Delta$; hence the context $Γ+πΔ$ in the
conclusion of the rule.

In writing this rule we needed to ``scale'' a context by
a multiplicity, and ``add'' two contexts.  We pause to define these operations.
\begin{definition}[Context addition]~
  \begin{align*}
    (x :_π A,Γ) + (x :_{μ} A,Δ) &= x :_{π+μ} A, (Γ+Δ)\\
    (x :_π A,Γ) + Δ &= x :_π A, Γ+Δ & (x ∉ Δ)\\
    () + Δ &= Δ
  \end{align*}
\end{definition}
\noindent
Context addition is total: if a variable occurs in both operands the
first rule applies (with possible re-ordering of bindings in $Δ$), if
not the second or third rule applies.

\begin{definition}[Context scaling]
  \begin{displaymath}
    π(x :_{μ} A, Γ) =  x :_{πμ} A, πΓ
  \end{displaymath}
\end{definition}

\begin{lemma}[Contexts form a module]
  The following laws hold:
  \begin{align*}
    Γ + Δ &= Δ + Γ &
    π (Γ+Δ) &= π Γ + π Δ\\
    (π+μ) Γ &= π Γ+ μ Γ \\
    (πμ) Γ &= π (μ Γ) &
    1 Γ &= Γ
  \end{align*}
\end{lemma}

These operations depend, in turn, on addition and multiplication of multiplicities.
The syntax of multiplicities is given in \fref{fig:syntax}.
We need the concrete multiplicities $1$ and $ω$ and, to support polymorphism,
multiplicity variables (ranged over by the metasyntactic
variables \(p\) and \(q\)) as well as formal sums and products of multiplicities.
Multiplicity expressions are quotiented by the following equivalence
relation:
\begin{definition}[equivalence of multiplicities]
  \label{def:equiv-multiplicity}
  The equivalence of multiplicities is the smallest transitive and
  reflexive relation, which obeys the following laws:
\begin{itemize}
\item $+$ and $·$ are associative and commutative
\item $1$ is the unit of $·$
\item $·$ distributes over $+$
\item $ω · ω = ω$
\item $1 + 1 = 1 + ω = ω + ω = ω$
\end{itemize}
\end{definition}
Thus, multiplicities form a semi-ring (without a zero), which extends
to a module structure on typing contexts. We may want to have a
stronger notion of equivalence for multiplicities, as we discuss in
\fref{sec:design-choices}.

Returning to the typing rules in \fref{fig:typing}, the rule (let) is like
a combination of (abs) and (app).  Again, each $\flet$ binding is
explicitly annotated with its multiplicity.
The variable rule (var) uses a standard idiom:
\vspace{-3mm}
$$\varrule$$
This rule allows us to ignore variables with
multiplicity $ω$ (usually called weakening),
so that, for example $x :_1 A, y :_ω B ⊢ x : A$ holds
\footnote{Pushing weakening to the variable rule is
  classic in many $λ$-calculi, and in the case of linear logic,
  dates back at least to Andreoli's work on
  focusing~\cite{andreoli_logic_1992}.}. Note that the judgement
$x :_ω A ⊢ x : A$ is an instance of the variable rule, because
$(x :_ω A)+(x :_1 A) = x:_ω A$.

Finally, abstraction and application for multiplicity polymorphism
are handled straightforwardly by (m.abs) and (m.app).

\subsection{Data constructors and case expressions}
\label{sec:typing-rules}

The handling of data constructors and case expressions is a
distinctive aspect of our design.  For constructor applications, the rule
(con), everything is straightforward: we treat the data constructor in
precisely the same way as an application of a function with that data constructor's type.
This includes weakening via the $ωΓ$ context in the conclusion.
The (case) rule is more interesting:
$$\caserule$$
First, notice that the $\mathsf{case}$ keyword is annotated with a
multiplicity $π$; this is analogous to the explicit
multiplicity on a $\mathsf{let}$ binding.  It says how often the scrutinee (or,
for a $\mathsf{let}$, the right hand side) will be consumed.  Just as
for $\mathsf{let}$, we expect $π$ to be inferred from an un-annotated $\mathsf{case}$ in
\HaskeLL{}.

The scrutinee $t$ is consumed $π$ times, which accounts for the $πΓ$ in
the conclusion.  Now consider the bindings $(x_i :_{πμ_i[π_1…π_n]} A_i)$ in the
environment for typechecking $u_k$.  That binding will be linear only if
\emph{both} $π$ \emph{and} $π_i$ are linear; that is, only if we specify
that the scrutinee is consumed once, and the $i$'th field of the data constructor $c_k$
specifies that is it consumed once if the constructor is (\fref{def:consume}).
To put it another way, suppose one of the linear
fields\footnote{
Recall \fref{sec:non-linear-constructors}, which described
how each constructor can have a mixture of linear and non-linear fields.}
of $c_k$ is used non-linearly in $u_k$.  Then, $μ_i=1$ (it is a linear field),
so $π$ must be $ω$, so that $πμ_i=ω$.  In short, using a linear field non-linearly
forces the scrutinee to be used non-linearly, which is just what we want.
Here are some concrete examples:
\begin{center}\vspace{-3mm}  
\begin{hscode}\SaveRestoreHook
\column{B}{@{}>{\hspre}l<{\hspost}@{}}%
\column{3}{@{}>{\hspre}l<{\hspost}@{}}%
\column{8}{@{}>{\hspre}c<{\hspost}@{}}%
\column{8E}{@{}l@{}}%
\column{12}{@{}>{\hspre}l<{\hspost}@{}}%
\column{19}{@{}>{\hspre}c<{\hspost}@{}}%
\column{19E}{@{}l@{}}%
\column{21}{@{}>{\hspre}l<{\hspost}@{}}%
\column{22}{@{}>{\hspre}l<{\hspost}@{}}%
\column{27}{@{}>{\hspre}l<{\hspost}@{}}%
\column{39}{@{}>{\hspre}l<{\hspost}@{}}%
\column{48}{@{}>{\hspre}l<{\hspost}@{}}%
\column{57}{@{}>{\hspre}l<{\hspost}@{}}%
\column{E}{@{}>{\hspre}l<{\hspost}@{}}%
\>[3]{}\Varid{fst}{}\<[8]%
\>[8]{}\mathbin{::}{}\<[8E]%
\>[12]{}(\Varid{a},\Varid{b})\mathbin{→}{}\<[21]%
\>[21]{}\Varid{a}\;{}\<[27]%
\>[27]{}\hspace{2cm}\;{}\<[39]%
\>[39]{}\Varid{swap}\mathbin{::}{}\<[48]%
\>[48]{}(\Varid{a},\Varid{b})\mathbin{⊸}(\Varid{b},\Varid{a}){}\<[E]%
\\
\>[3]{}\Varid{fst}\;{}\<[12]%
\>[12]{}(\Varid{a},\Varid{b}){}\<[19]%
\>[19]{}\mathrel{=}{}\<[19E]%
\>[22]{}\Varid{a}\;{}\<[39]%
\>[39]{}\Varid{swap}\;{}\<[48]%
\>[48]{}(\Varid{a},\Varid{b})\mathrel{=}{}\<[57]%
\>[57]{}(\Varid{b},\Varid{a}){}\<[E]%
\ColumnHook
\end{hscode}\resethooks
\end{center}
Recall that both fields of a pair are linear (\fref{sec:linear-constructors}).
In \ensuremath{\Varid{fst}}, the second component of the pair is used non-linearly (by being
discarded) which forces the use of $\mathsf{case}_ω$, and hence a non-linear type
for \ensuremath{\Varid{fst}}.  But \ensuremath{\Varid{swap}} uses the components linearly, so we can use $\mathsf{case}_1$, giving
\ensuremath{\Varid{swap}} a linear type.

\subsection{Metatheory}
\label{sec:metatheory}

In order to prove that our type system meets its stated goals, we
introduce an operational semantics.
\ifx\longversion\undefined{
  The details can be found in the extended version of this
  article~\cite{extended_version}.
}
\else{
  The details are deferred to \fref{appendix:dynamics}.
}
\fi

\paragraph{Of consuming exactly once}
The operational semantics is a big-step operational semantics for
lazy evaluation in the style of \citet{launchbury_natural_1993}.
Following
\citet{gunter_partial-big-step_1993}, starting from a big-step evaluation
relation $a⇓b$, we define \emph{partial derivations} and from there a
\emph{partial evaluation} relation
\ifx\longversion\undefined{$a⇓^*b$. }
\else{$a⇓^*b$ (see \fref{sec:partial-derivations}). }
\fi
Progress is expressed as the
fact that a derivation of $a⇓^*b$ can always be extended.

The operational semantics differs from
\citeauthor{launchbury_natural_1993}'s in two major respects:
\begin{itemize}
\item The reduction states are heavily annotated with type
  information. These type annotations are used for the proofs.
\item Reduction is indexed by whether we intend to consume the term
  under consideration exactly once or an arbitrary number of times
\item Variables in the environments are annotated by a multiplicity
  ($1$ or $ω$), $ω$-variables are ordinary variables. When forced, an
  $ω$-variable is replaced by its value (to model lazy sharing), but
  $1$-variables \emph{must be consumed exactly once}: when forced,
  they are removed from the environment. Reduction gets stuck if a
  1-variable is used more than once.
\end{itemize}
Because the operational semantics gets stuck if a 1-variable is used
more than once, the progress theorem
(\fref{thm:progress-denotational}) shows that linear functions do
indeed consume their argument at most once if their result is consumed
exactly once. The 1-variables are in fact used exactly once: it is a
consequence of type preservation that evaluation of a closed term of a
basic type (say \ensuremath{\Conid{Bool}}) returns an environment with no 1-variables.

Our preservation and progress theorems
\ifx\longversion\undefined{(proved in the extended
  version~\cite{extended_version})}
\else {(proved in \fref{sec:denotational})}
\fi
read as follows:

\begin{theorem}[Type preservation]\label{thm:type-safety}
  If $a$ is well typed, and $a⇓b$, or $a⇓^*b$ then $b$ is well-typed.
\end{theorem}
\begin{theorem}[Progress]\label{thm:progress-denotational}
  Evaluation does not block. That is, for any partial evaluation
  $a⇓^*b$, where $a$ is well-typed, the derivation can be extended.
\end{theorem}

\paragraph{In-place update \& typestate}
Furthermore, linear types can be
used to implement some operations as in-place updates, and
typestates (like whether an array is mutable or frozen) are actually
enforced by the type system.

To prove this, we introduce a second, distinct, semantics. It is also
a Launchbury-style semantics. It differs from \citet{launchbury_natural_1993} in the following ways:
\begin{itemize}
\item Environments are enriched with mutable references (for the sake
  of concreteness, they are all references to arrays but they could be
  anything).
\item Typestates are implemented by mutating the type of such
  references, functions can block if the type of the references is not
  correct: that is, we track typestates dynamically.
\end{itemize}
The idea behind the latter is that progress will show that we are
never blocked by typestates. In other words, they are enforced
statically and can be erased at runtime.

It is hard to reason on a lazy language with mutation. But what we
show is that we are using mutation carefully enough so that
they behave as pure data. To formalise this, we relate this
semantics with mutation to our pure semantics above. Specifically, we
show that they are \emph{bisimilar}.
\citet{amani_cogent_2016} use a similar technique for a language with linear types
and both a pure and imperative semantics.

Bisimilarity allows us to lift the type-preservation and
progress from the pure semantics. That is, writing $σ,τ$ for states of
this evaluation with mutation:

\begin{theorem}[Type preservation]\label{thm:type-preservation}
  For any well-typed $σ$, if $σ⇓τ$ or $σ⇓^*τ$, then $τ$ is
  well-typed.
\end{theorem}
\begin{theorem}[Progress]\label{thm:progress}
  Evaluation does not block. That is, for any partial evaluation
  $σ⇓^*τ$, for $σ$ well-typed, the evaluation can be
  extended.
  In particular, typestates need not be checked dynamically.
\end{theorem}

Just as importantly, we can prove that, indeed, we cannot observe
mutations. More precisely, we prove that the pure semantics and the
semantics with mutation are observationally equivalent: any observation, which we
reduce to a boolean test, is identical in either semantics.
\begin{theorem}[Observational equivalence]\label{thm:obs-equiv}
  The semantics with in-place mutation is observationally equivalent
  to the pure semantics.

  That is, for any closed term of type $\varid{Bool}$, if $e$ evaluates
  to the value $z$ with the pure semantics, and to the value $z'$ with
  the semantics with mutation, then $z=z'$.
\end{theorem}

\subsection{Design choices \& trade-offs}
\label{sec:design-choices}

We could as well have picked different points in the design space for
\calc{}. We review some of the choices we made in this section.

\paragraph{Case rule}
Thanks to $\varid{case}_ω$, we can use linear arrows on all data
types. Indeed we can write the following and have \ensuremath{\Varid{fst}} typecheck:
\begin{hscode}\SaveRestoreHook
\column{B}{@{}>{\hspre}l<{\hspost}@{}}%
\column{3}{@{}>{\hspre}l<{\hspost}@{}}%
\column{E}{@{}>{\hspre}l<{\hspost}@{}}%
\>[B]{}\mathbf{data}\;\Conid{Pair}\;\Varid{a}\;\Varid{b}\;\mathbf{where}{}\<[E]%
\\
\>[B]{}\hsindent{3}{}\<[3]%
\>[3]{}\Conid{Pair}\mathbin{::}\Varid{a}\mathbin{⊸}\Varid{b}\mathbin{⊸}\Conid{Pair}\;\Varid{a}\;\Varid{b}{}\<[E]%
\\[\blanklineskip]%
\>[B]{}\Varid{fst}\mathbin{::}\Conid{Pair}\;\Varid{a}\;\Varid{b}\to \Varid{a}{}\<[E]%
\\
\>[B]{}\Varid{fst}\;\Varid{x}\mathrel{=}\mathbf{case}\;\!\!_{\omega}\;\Varid{x}\;\mathbf{of}\;\Conid{Pair}\;\Varid{a}\;\Varid{b}\to \Varid{a}{}\<[E]%
\ColumnHook
\end{hscode}\resethooks
It is possible to do without $\varid{case}_ω$, and have only $\varid{case}_1$.
Consider \ensuremath{\Varid{fst}} again.  We could instead have
\begin{hscode}\SaveRestoreHook
\column{B}{@{}>{\hspre}l<{\hspost}@{}}%
\column{3}{@{}>{\hspre}l<{\hspost}@{}}%
\column{E}{@{}>{\hspre}l<{\hspost}@{}}%
\>[B]{}\mathbf{data}\;\Conid{Pair}\;\Varid{p}\;\Varid{q}\;\Varid{a}\;\Varid{b}\;\mathbf{where}{}\<[E]%
\\
\>[B]{}\hsindent{3}{}\<[3]%
\>[3]{}\Conid{Pair}\mathbin{::}\Varid{a}\mathbin{→}_{\Varid{p}}\;\Varid{b}\mathbin{→}_{\Varid{q}}\;\Conid{Pair}\;\Varid{p}\;\Varid{q}\;\Varid{a}\;\Varid{b}{}\<[E]%
\\[\blanklineskip]%
\>[B]{}\Varid{fst}\mathbin{::}\Conid{Pair}\;\mathrm{1}\;\omega\;\Varid{a}\;\Varid{b}\mathbin{⊸}\Varid{a}{}\<[E]%
\\
\>[B]{}\Varid{fst}\;\Varid{x}\mathrel{=}\mathbf{case}\;\!\!_{\mathrm{1}}\;\Varid{x}\;\mathbf{of}\;\Conid{Pair}\;\Varid{a}\;\Varid{b}\to \Varid{a}{}\<[E]%
\ColumnHook
\end{hscode}\resethooks
But now multiplicity polymorphism infects all basic datatypes (such
as pairs), with knock-on consequences.  Moreover, \ensuremath{\mathbf{let}} is annotated so it seems
reasonable to annotate \ensuremath{\mathbf{case}} in the same way.

To put it another way, $\varid{case}_ω$ allows us to meaningfully inhabit
\ensuremath{∀\Varid{a}\;\Varid{b}. \Conid{Unrestricted}\;(\Varid{a},\Varid{b})\mathbin{⊸}(\Conid{Unrestricted}\;\Varid{a},\Conid{Unrestricted}\;\Varid{b})}, while linear logic
does not.

\paragraph{Subtyping}
Because the type $A⊸B$ only strengthens the contract of its elements
compared to $A→B$, one might expect the type $A⊸B$ to be a subtype of $A→B$.
But while \calc{} has \emph{polymorphism}, it does not have \emph{subtyping}.
For example, if
\begin{hscode}\SaveRestoreHook
\column{B}{@{}>{\hspre}l<{\hspost}@{}}%
\column{3}{@{}>{\hspre}l<{\hspost}@{}}%
\column{E}{@{}>{\hspre}l<{\hspost}@{}}%
\>[3]{}\Varid{f}\mathbin{::}\Conid{Int}\mathbin{⊸}\Conid{Int}{}\<[E]%
\\
\>[3]{}\Varid{g}\mathbin{::}(\Conid{Int}\to \Conid{Int})\to \Conid{Bool}{}\<[E]%
\ColumnHook
\end{hscode}\resethooks
then the call \ensuremath{(\Varid{g}\;\Varid{f})} is ill-typed, even though \ensuremath{\Varid{f}} provides more
guarantees than \ensuremath{\Varid{g}} requires.
On the other hand, \ensuremath{\Varid{g}} might well be multiplicity-polymorphic, with type
\ensuremath{∀\Varid{p}. (\Conid{Int}\to _{\Varid{p}}\;\Conid{Int})\to \Conid{Bool}}; in which case \ensuremath{(\Varid{g}\;\Varid{f})} is, indeed, typeable.

The lack of subtyping is a deliberate choice in our design: it is well
known that Hindley-Milner-style type inference does not mesh well with
subtyping (see, for example, the extensive exposition by
\citet{pottier_subtyping_1998}, but also \citet{dolan_mlsub_2017} for
a counterpoint).

However, while \ensuremath{(\Varid{g}\;\Varid{f})} is ill-typed in \calc{}, it is accepted in \HaskeLL{}.
The reason is that the $η$-expansion \ensuremath{\Varid{g}\;(\lambda \Varid{x}\to \Varid{f}\;\Varid{x})} \emph{is}
typeable, and \HaskeLL{} perform this expansions during type
inference. Such an $η$-expansion is not completely
semantics-preserving as \ensuremath{\lambda \Varid{x}\to \Varid{g}\;\Varid{x}} is always a well-defined value,
even if \ensuremath{\Varid{g}} loops: this difference can be observed with Haskell's
\ensuremath{\Varid{seq}} operator. Nevertheless, such $η$-expansions are already standard
practice in \textsc{ghc}: a similar situation arises when using rank-2
types. For example, the core language of \textsc{ghc} does not accept
\ensuremath{\Varid{g}\;\Varid{f}} for
\begin{hscode}\SaveRestoreHook
\column{B}{@{}>{\hspre}l<{\hspost}@{}}%
\column{E}{@{}>{\hspre}l<{\hspost}@{}}%
\>[B]{}\Varid{g}\mathbin{::}(∀\Varid{a}. (\Conid{Eq}\;\Varid{a},\Conid{Show}\;\Varid{a})\Rightarrow \Varid{a}\to \Conid{Int})\to \Conid{Char}{}\<[E]%
\\
\>[B]{}\Varid{f}\mathbin{::}∀\Varid{a}. (\Conid{Show}\;\Varid{a},\Conid{Eq}\;\Varid{a})\Rightarrow \Varid{a}\to \Conid{Int}{}\<[E]%
\ColumnHook
\end{hscode}\resethooks
The surface language, again, accepts \ensuremath{\Varid{g}\;\Varid{f}}, and elaborates it into \ensuremath{\Varid{g}\;(\lambda \Varid{a}\;(\Varid{d1}\mathbin{::}\Conid{Eq}\;\Varid{a})\;(\Varid{d2}\mathbin{::}\Conid{Show}\;\Varid{a})\to \Varid{f}\;\Varid{a}\;\Varid{d2}\;\Varid{d1})}. We simply extend this
mechanism to linear types.

\paragraph{Polymorphism \& multiplicities} Consider the definition: ``\ensuremath{\Varid{id}\;\Varid{x}\mathrel{=}\Varid{x}}''.
Our typing rules would validate both \ensuremath{\Varid{id}\mathbin{::}\Conid{Int}\mathbin{→}\Conid{Int}} and \ensuremath{\Varid{id}\mathbin{::}\Conid{Int}\mathbin{⊸}\Conid{Int}}.
So, since we think of multiplicities ranging over $\{1,ω\}$, surely we should
also have \ensuremath{\Varid{id}\mathbin{::}∀\Varid{p}. \Conid{Int}\mathbin{→}_{\Varid{p}}\;\Conid{Int}}?  But as it stands, our rules do
not accept it. To do so we would need $x :_p Int ⊢ x : Int$.  Looking
at the (var) rule in \fref{fig:typing}, we can prove that premise by case analysis,
trying $p=1$ and $p=ω$.
But if we had a domain of multiplicities which includes
$0$ (see \fref{sec:extending-multiplicities}), we would not be able to prove $x :_p Int ⊢ x : Int$, and rightly
so because it is not the case that \ensuremath{\Varid{id}\mathbin{::}\Conid{Int}\mathbin{→}_{\mathrm{0}}\;\Conid{Int}}.

More generally, we could type more programs if we added more laws
relating variables in \fref{def:equiv-multiplicity}, such as $p+q=ω$,
but this would prevent potential extensions to the set of
multiplicities. For now, we accept the more conservative rules of
\fref{def:equiv-multiplicity}.

\paragraph{Divergence.}  Consider this definition\footnote{
\ensuremath{(\Varid{repeat}\;\Varid{x})} returns the infinite list \ensuremath{[\mskip1.5mu \Varid{x},\Varid{x},\mathbin{...}\mskip1.5mu]}.  The function
\ensuremath{(\plus )} appends two lists, and has type \ensuremath{[\mskip1.5mu \Varid{a}\mskip1.5mu]\mathbin{⊸}[\mskip1.5mu \Varid{a}\mskip1.5mu]\mathbin{⊸}[\mskip1.5mu \Varid{a}\mskip1.5mu]}; we have not
given a formal typing rule for recursive definitions, but its form is
entirely standard.}:
\begin{hscode}\SaveRestoreHook
\column{B}{@{}>{\hspre}l<{\hspost}@{}}%
\column{3}{@{}>{\hspre}l<{\hspost}@{}}%
\column{E}{@{}>{\hspre}l<{\hspost}@{}}%
\>[3]{}\Varid{f}\mathbin{::}[\mskip1.5mu \Conid{Int}\mskip1.5mu]\mathbin{⊸}[\mskip1.5mu \Conid{Int}\mskip1.5mu]{}\<[E]%
\\
\>[3]{}\Varid{f}\;\Varid{xs}\mathrel{=}\Varid{repeat}\;\mathrm{1}\plus \Varid{xs}{}\<[E]%
\ColumnHook
\end{hscode}\resethooks
But wait!  Does \ensuremath{\Varid{f}} \emph{really} consume its argument \ensuremath{\Varid{xs}} exactly
once? After all, \ensuremath{(\Varid{repeat}\;\mathrm{1})} is infinite so \ensuremath{\Varid{f}} will never evaluate \ensuremath{\Varid{xs}}
at all!

In corner cases like this we look to metatheory.  Yes, the typing rules give
the specified types for \ensuremath{(\plus )} and \ensuremath{\Varid{f}}.
Yes, the operational claims guaranteed by the metatheory remain valid.
Intuitively you may imagine it like this: linearity claims that \emph{if} you were consume
the result of \ensuremath{\Varid{f}} completely, exactly once, then you would consume its argument once; but
since the result of \ensuremath{\Varid{f}} is infinite we cannot consume it completely exactly once, so
the claim holds vacuously.

\section{Implementing \HaskeLL}
\label{sec:implementation}
\label{sec:impl}

We implement \HaskeLL{} on top of the leading Haskell compiler,
\textsc{ghc}, version 8.2\footnote{https://github.com/tweag/ghc/tree/linear-types}.
The implementation modifies type inference and
type-checking in the compiler. Neither the intermediate language~\cite{sulzmann_fc_2007}
nor the run-time system are affected.
Our implementation of multiplicity polymorphism is incomplete, but the current
prototype is sufficient for the examples and case studies presented in
in this paper (see \fref{sec:evaluation}).

In order to implement the linear arrow, we added a multiplicity
annotation to function arrows as in \calc{}.
The constructor for arrow types is
constructed and destructed frequently in \textsc{ghc}'s type checker, and this
accounts for most of the modifications to existing code.

As suggested in \fref{sec:typing-contexts}, the multiplicities are an
output of the type inference algorithm. In order to infer the
multiplicities coming out of a \ensuremath{\mathbf{case}} expression we need a way to
aggregate the multiplicities coming out of the individual branches. To
this effect, we compute, for every variable, the join of its
multiplicity in each branch.

Implementing \HaskeLL{}
affects 1,152 lines of \textsc{ghc} (in subsystems of the compiler
that together amount to more than 100k lines of code), including 444
net extra lines. These figures support our claim that \HaskeLL{} is
easy to integrate into an existing implementation: despite
\textsc{ghc} being 25 years old, we implement a first version of
\HaskeLL{} with reasonable effort.

\section{Evaluation and Case Studies}
\label{sec:evaluation}
\label{sec:applications}

While many linear type systems have been proposed,
a {\em retrofitted} linear type system for a mature language like Haskell offers
the opportunity to implement non-trivial applications mixing linear and
non-linear code, I/O, etc., and observe how linear code interacts with existing
libraries and the optimiser of a sophisticated compiler.

Our first method for evaluating the implementation is to simply compile a large
existing code base together with the following changes: (1) all
(non-\textsc{gadt}) data constructors are
linear by default, as implied by the new type system; and (2) we update standard
list functions to have linear types (\ensuremath{\plus }, \ensuremath{\Varid{concat}}, \ensuremath{\Varid{uncons}}).
Under these conditions, we verified that the base \textsc{ghc} libraries and the nofib
benchmark suites compile successfully: 195K lines of Haskell, providing
preliminary evidence of backwards compatibility.

In the remainder of section, we describe case-studies implementated with the modified
\ghc  of \fref{sec:implementation}.
In \fref{sec:industry}, we propose further applications for \HaskeLL, which we
have not yet implemented, but which motivate this work.

\subsection{Computing directly with serialised data}
\label{sec:cursors}

While \fref{sec:freezing-arrays} covered simple mutable arrays, we now
turn to a related but more complicated application: operating directly on binary,
serialised representations of algebraic datatypes
(like \citet{vollmer_gibbon_2017} do).
The motivation is that programs are increasingly decoupled into separate (cloud)
services that communicate via serialised data in text or binary formats, carried
by remote procedure calls.
The standard approach is to deserialise data into an in-heap, pointer-based representation,
process it, and then serialise the result for transmission.
This process is inefficient, but nevertheless tolerated, because the alternative
--- computing directly with serialised data --- is far too difficult to program.
Nevertheless, the potential performance gain of working directly with serialised
data has motivated small steps in this direction:
libraries like ``Cap'N Proto''~\footnote{{\url{https://capnproto.org/}}} enable unifying
in-memory and on-the-wire formats for simple product types (protobufs).

Here is an unusual case where {advanced types can yield {\em performance}} by
making it practical to code in a previously infeasible style: accessing
serialised data at a fine grain without copying it.

\begin{wrapfigure}[8]{r}[34pt]{8.5cm}
\vspace{-4mm}
\begin{hscode}\SaveRestoreHook
\column{B}{@{}>{\hspre}l<{\hspost}@{}}%
\column{9}{@{}>{\hspre}l<{\hspost}@{}}%
\column{14}{@{}>{\hspre}l<{\hspost}@{}}%
\column{E}{@{}>{\hspre}l<{\hspost}@{}}%
\>[B]{}\mathbf{data}\;\Conid{Tree}\mathrel{=}\Conid{Leaf}\;\Conid{Int}\mid \Conid{Branch}\;\Conid{Tree}\;\Conid{Tree}{}\<[E]%
\\
\>[B]{}\Varid{pack}{}\<[9]%
\>[9]{}\mathbin{::}\Conid{Tree}\mathbin{⊸}\Conid{Packed}\mathop{{\kern 1pt}'}[\mskip1.5mu \Conid{Tree}\mskip1.5mu]{}\<[E]%
\\
\>[B]{}\Varid{unpack}{}\<[9]%
\>[9]{}\mathbin{::}\Conid{Packed}\mathop{{\kern 1pt}'}[\mskip1.5mu \Conid{Tree}\mskip1.5mu]\mathbin{⊸}\Conid{Tree}{}\<[E]%
\\
\>[B]{}\Varid{caseTree}\mathbin{::}{}\<[14]%
\>[14]{}\Conid{Packed}\;(\Conid{Tree}\mathop{{\kern 1pt}'\!\!:}\Varid{r})\to _{\Varid{p}}{}\<[E]%
\\
\>[14]{}(\Conid{Packed}\;(\Conid{Int}\mathop{{\kern 1pt}'\!\!:}\Varid{r})\to _{\Varid{p}}\;\Varid{a})\to {}\<[E]%
\\
\>[14]{}(\Conid{Packed}\;(\Conid{Tree}\mathop{{\kern 1pt}'\!\!:}\Conid{Tree}\mathop{{\kern 1pt}'\!\!:}\Varid{r})\to _{\Varid{p}}\;\Varid{a})\to \Varid{a}{}\<[E]%
\ColumnHook
\end{hscode}\resethooks
\end{wrapfigure}
The interface on the right gives an example of type-safe, {\em read-only}
access to serialised data for a particular datatype\footnote{This
  interface uses type-level lists as can be found in Haskell's
  \textsf{DataKind} extension}.
A \ensuremath{\Conid{Packed}} value is a pointer to raw bits (a bytestring), indexed by the
types of the values contained within.  We define a {\em type-safe} serialisation
layer as one which {\em reads} byte-ranges only at the type and size they were
originally {\em written}.  This is a small extension of the memory safety we
already expect of Haskell's heap --- extended to include the contents of
bytestrings containing serialised data\footnote{The additional safety ensured
  here is lower-stakes than typical memory-safety, as, even it is violated, the
  serialised values do not contain pointers and cannot segfault the program
  reading them.}.
To preserve this type safety, the \ensuremath{\Conid{Packed}} type {\em must} be abstract.
Consequently, a client of the module defining \ensuremath{\Conid{Tree}} need not be privy to the
memory layout of its serialisation.

If we cannot muck about with the bits inside a \ensuremath{\Conid{Packed}} directly, then we can
still retrieve data with \ensuremath{\Varid{unpack}}, \ie, the traditional, {\em copying},
approach to deserialisation.  Better still is to read the data {\em without}
copying.  We can manage this feat with \ensuremath{\Varid{caseTree}}, which is analogous to the
expression ``\ensuremath{\mathbf{case}\;\Varid{e}\;\mathbf{of}\;\{\mskip1.5mu \Conid{Leaf}\mathbin{...};\Conid{Branch}\mathbin{...}\mskip1.5mu\}}''.
Lacking built-in syntax, \ensuremath{(\Varid{caseTree}\;\Varid{p}\;\Varid{k1}\;\Varid{k2})} takes two continuations
corresponding to the two branches of the case expression.
Unlike the case expression, \ensuremath{\Varid{caseTree}} operates on the packed byte stream, reads
a tag byte, advances the pointer past it, and returns a type-safe pointer to the
fields (\eg \ensuremath{\Conid{Packed}\mathop{{\kern 1pt}'}[\mskip1.5mu \Conid{Int}\mskip1.5mu]} in the case of a leaf).

It is precisely to access multiple, consecutive fields that \ensuremath{\Conid{Packed}} is indexed
by a {\em list} of types as its phantom type parameter.  Individual atomic
values (\ensuremath{\Conid{Int}}, \ensuremath{\Conid{Double}}, etc) can be read one at a time with a lower-level
\ensuremath{\Varid{read}} primitive, which can efficiently read out scalars and store them in
registers:
\begin{hscode}\SaveRestoreHook
\column{B}{@{}>{\hspre}l<{\hspost}@{}}%
\column{3}{@{}>{\hspre}l<{\hspost}@{}}%
\column{E}{@{}>{\hspre}l<{\hspost}@{}}%
\>[3]{}\Varid{read}\mathbin{::}\Conid{Storable}\;\Varid{a}\Rightarrow \Conid{Packed}\;(\Varid{a}\mathop{{\kern 1pt}'\!\!:}\Varid{r})\mathbin{⊸}(\Varid{a},\Conid{Packed}\;\Varid{r}){}\<[E]%
\ColumnHook
\end{hscode}\resethooks

\ifx\longversion\undefined{
  \pagebreak 
}
\else{}
\fi

\begin{wrapfigure}[8]{r}[0pt]{7.0cm} 
\ifx\longversion\undefined
{\vspace{-2mm}}
\else
{\vspace{-5mm}}
\fi
\begin{hscode}\SaveRestoreHook
\column{B}{@{}>{\hspre}l<{\hspost}@{}}%
\column{3}{@{}>{\hspre}l<{\hspost}@{}}%
\column{13}{@{}>{\hspre}l<{\hspost}@{}}%
\column{18}{@{}>{\hspre}l<{\hspost}@{}}%
\column{22}{@{}>{\hspre}l<{\hspost}@{}}%
\column{27}{@{}>{\hspre}l<{\hspost}@{}}%
\column{E}{@{}>{\hspre}l<{\hspost}@{}}%
\>[B]{}\Varid{sumLeaves}\mathbin{::}\Conid{Packed}\mathop{{\kern 1pt}'}[\mskip1.5mu \Conid{Tree}\mskip1.5mu]\to \Conid{Int}{}\<[E]%
\\
\>[B]{}\Varid{sumLeaves}\;\Varid{p}\mathrel{=}\Varid{fst}\;(\Varid{go}\;\Varid{p}){}\<[E]%
\\
\>[B]{}\hsindent{3}{}\<[3]%
\>[3]{}\mathbf{where}\;\Varid{go}\;{}\<[13]%
\>[13]{}\Varid{p}\mathrel{=}{}\<[18]%
\>[18]{}\Varid{caseTree}\;\Varid{p}{}\<[E]%
\\
\>[13]{}\Varid{read}\mbox{\onelinecomment  Leaf case}{}\<[E]%
\\
\>[13]{}(\lambda \Varid{p2}\to {}\<[22]%
\>[22]{}\mathbf{let}\;{}\<[27]%
\>[27]{}(\Varid{n},\Varid{p3})\mathrel{=}\Varid{go}\;\Varid{p2}{}\<[E]%
\\
\>[27]{}(\Varid{m},\Varid{p4})\mathrel{=}\Varid{go}\;\Varid{p3}{}\<[E]%
\\
\>[22]{}\mathbf{in}\;(\Varid{n}\mathbin{+}\Varid{m},\Varid{p4})){}\<[E]%
\ColumnHook
\end{hscode}\resethooks
\end{wrapfigure}
Putting it together, we can write a function that consumes serialised data, such
as \ensuremath{\Varid{sumLeaves}}, shown on the right.
Indeed, we can even use \ensuremath{\Varid{caseTree}} to implement \ensuremath{\Varid{unpack}}, turning it into safe
``client code'' -- sitting outside the module that defines \ensuremath{\Conid{Tree}} and the
trusted code establishing its memory representation.

In this read-only example, linearity was not essential, only phantom types.
Next we consider an \textsc{API} for writing \ensuremath{\Conid{Packed}\mathop{{\kern 1pt}'}[\mskip1.5mu \Conid{Tree}\mskip1.5mu]} values bit by bit, where linearity is key.
In particular, can we also implement \ensuremath{\Varid{pack}} using a public interface?

\subsubsection{Writing serialised data}

To create a serialised data constructor, we must write a tag, followed by the
fields.
A {\em linear} write pointer can ensure all fields are initialised, in order.
We use a type ``\ensuremath{\Conid{Needs}}'' for write pointers, parameterised by (1) a
list of remaining things to be written, and (2) the type of the final value
which will be initialised once those writes are performed.  For example, after
we write the tag of a \ensuremath{\Conid{Leaf}} we are left with:
``\ensuremath{\Conid{Needs}\mathop{{\kern 1pt}'}[\mskip1.5mu \Conid{Int}\mskip1.5mu]\;\Conid{Tree}}'' ---
an {\em obligation} to write the \ensuremath{\Conid{Int}} field, and a {\em promise} to receive
a \ensuremath{\Conid{Tree}} value at the end (albeit a packed one).

To write an individal number, we provide a primitive that shaves one
element off the type-level list of obligations (a counterpart to \ensuremath{\Varid{read}}, above):
As with mutable arrays, this \ensuremath{\Varid{write}} operates in-place on the buffer, in spite
being a pure function.
\begin{hscode}\SaveRestoreHook
\column{B}{@{}>{\hspre}l<{\hspost}@{}}%
\column{3}{@{}>{\hspre}l<{\hspost}@{}}%
\column{E}{@{}>{\hspre}l<{\hspost}@{}}%
\>[3]{}\Varid{write}\mathbin{::}\Conid{Storable}\;\Varid{a}\Rightarrow \Varid{a}\mathbin{⊸}\Conid{Needs}\;(\Varid{a}\mathop{{\kern 1pt}'\!\!:}\Varid{r})\;\Varid{t}\mathbin{⊸}\Conid{Needs}\;\Varid{r}\;\Varid{t}{}\<[E]%
\ColumnHook
\end{hscode}\resethooks
When the list of outstanding writes is empty, we can retrive a readable packed buffer.
Just as when we froze arrays (\fref{sec:freezing-arrays}), the immutable value
is {\em unrestricted}, and can be used multiple times:
\begin{hscode}\SaveRestoreHook
\column{B}{@{}>{\hspre}l<{\hspost}@{}}%
\column{3}{@{}>{\hspre}l<{\hspost}@{}}%
\column{E}{@{}>{\hspre}l<{\hspost}@{}}%
\>[3]{}\Varid{finish}\mathbin{::}\Conid{Needs}\mathop{{\kern 1pt}'}[\mskip1.5mu \mskip1.5mu]\;\Varid{t}\mathbin{⊸}\Conid{Unrestricted}\;(\Conid{Packed}\mathop{{\kern 1pt}'}[\mskip1.5mu \Varid{t}\mskip1.5mu]){}\<[E]%
\ColumnHook
\end{hscode}\resethooks
Finalizing written values with \ensuremath{\Varid{finish}} works hand in hand with allocating new
buffers in which to write data (similar to \ensuremath{\Varid{newMArray}} from \fref{sec:freezing-arrays}):
\begin{hscode}\SaveRestoreHook
\column{B}{@{}>{\hspre}l<{\hspost}@{}}%
\column{3}{@{}>{\hspre}l<{\hspost}@{}}%
\column{E}{@{}>{\hspre}l<{\hspost}@{}}%
\>[3]{}\Varid{newBuffer}\mathbin{::}(\Conid{Needs}\mathop{{\kern 1pt}'}[\mskip1.5mu \Varid{a}\mskip1.5mu]\;\Varid{a}\mathbin{⊸}\Conid{Unrestricted}\;\Varid{b})\mathbin{⊸}\Varid{b}{}\<[E]%
\ColumnHook
\end{hscode}\resethooks
We also need to explicitly let go of linear input buffers we've exhausted.
\begin{hscode}\SaveRestoreHook
\column{B}{@{}>{\hspre}l<{\hspost}@{}}%
\column{3}{@{}>{\hspre}l<{\hspost}@{}}%
\column{E}{@{}>{\hspre}l<{\hspost}@{}}%
\>[3]{}\Varid{done}\mathbin{::}\Conid{Packed}\mathop{{\kern 1pt}'}[\mskip1.5mu \mskip1.5mu]\mathbin{⊸}(){}\<[E]%
\ColumnHook
\end{hscode}\resethooks

The primitives \ensuremath{\Varid{write}}, \ensuremath{\Varid{read}}, \ensuremath{\Varid{newBuffer}}, \ensuremath{\Varid{done}}, and \ensuremath{\Varid{finish}} are {\em general}
operations for serialised data, whereas \ensuremath{\Varid{caseTree}} is datatype-specific.
Further, the module that defines \ensuremath{\Conid{Tree}} exports a datatype-specific way to 
{\em write} each serialised data constructor:
%
\begin{hscode}\SaveRestoreHook
\column{B}{@{}>{\hspre}l<{\hspost}@{}}%
\column{14}{@{}>{\hspre}l<{\hspost}@{}}%
\column{E}{@{}>{\hspre}l<{\hspost}@{}}%
\>[B]{}\Varid{startLeaf}{}\<[14]%
\>[14]{}\mathbin{::}\Conid{Needs}\;(\Conid{Tree}\mathop{{\kern 1pt}'\!\!:}\Varid{r})\;\Varid{t}\mathbin{⊸}\Conid{Needs}\;(\Conid{Int}\mathop{{\kern 1pt}'\!\!:}\Varid{r})\;\Varid{t}{}\<[E]%
\\
\>[B]{}\Varid{startBranch}{}\<[14]%
\>[14]{}\mathbin{::}\Conid{Needs}\;(\Conid{Tree}\mathop{{\kern 1pt}'\!\!:}\Varid{r})\;\Varid{t}\mathbin{⊸}\Conid{Needs}\;(\Conid{Tree}\mathop{{\kern 1pt}'\!\!:}\Conid{Tree}\mathop{{\kern 1pt}'\!\!:}\Varid{r})\;\Varid{t}{}\<[E]%
\ColumnHook
\end{hscode}\resethooks
Operationally, \ensuremath{\Varid{start}\mathbin{*}} functions write only the tag, hiding the exact
tag-encoding from the client, and leaving field-writes as future obligations.
With these building blocks, we can move \ensuremath{\Varid{pack}} and \ensuremath{\Varid{unpack}} outside of the
private code that defines \ensuremath{\Conid{Tree}}s, which has this minimal interface:
\begin{hscode}\SaveRestoreHook
\column{B}{@{}>{\hspre}l<{\hspost}@{}}%
\column{E}{@{}>{\hspre}l<{\hspost}@{}}%
\>[B]{}\mathbf{module}\;\Conid{TreePrivate}\;(\Conid{Tree}\;(\mathinner{\ldotp\ldotp}),\Varid{caseTree},\Varid{startLeaf},\Varid{startBranch}){}\<[E]%
\\
\>[B]{}\mathbf{module}\;\Conid{Data.Packed}\;(\Conid{Packed},\Conid{Needs},\Varid{read},\Varid{write},\Varid{newBuffer},\Varid{finish},\Varid{done}){}\<[E]%
\ColumnHook
\end{hscode}\resethooks
On top of the safe interface, we can of course define higher-level construction
routines, such as for writing a complete \ensuremath{\Conid{Leaf}}:
\begin{hscode}\SaveRestoreHook
\column{B}{@{}>{\hspre}l<{\hspost}@{}}%
\column{3}{@{}>{\hspre}l<{\hspost}@{}}%
\column{E}{@{}>{\hspre}l<{\hspost}@{}}%
\>[3]{}\Varid{writeLeaf}\;\Varid{n}\mathrel{=}\Varid{write}\;\Varid{n}\mathbin{∘}\Varid{startLeaf}{}\<[E]%
\ColumnHook
\end{hscode}\resethooks
Now we can allocate and initialize a complete tree --- equivalent to \ensuremath{\Conid{Branch}\;(\Conid{Leaf}\;\mathrm{3})\;(\Conid{Leaf}\;\mathrm{4})}, but without ever creating the non-serialised values --- as
follows:
\begin{hscode}\SaveRestoreHook
\column{B}{@{}>{\hspre}l<{\hspost}@{}}%
\column{E}{@{}>{\hspre}l<{\hspost}@{}}%
\>[B]{}\Varid{newBuffer}\;(\Varid{finish}\mathbin{∘}\Varid{writeLeaf}\;\mathrm{4}\mathbin{∘}\Varid{writeLeaf}\;\mathrm{3}\mathbin{∘}\Varid{startBranch})\mathbin{::}\Conid{Packed}\mathop{{\kern 1pt}'}[\mskip1.5mu \Conid{Tree}\mskip1.5mu]{}\<[E]%
\ColumnHook
\end{hscode}\resethooks

Finally, we have what we need to build a map function that logically operates
on the leaves of a tree, but reads serialised input and writes serialised
output.
Indeed, in our current \HaskeLL{} implementation ``\ensuremath{\Varid{mapLeaves}\;(\mathbin{+}\mathrm{1})\;\Varid{tree}}'' touches {\em
  only} packed buffers --- it performs zero Haskell heap allocation!
We will return to this map example and benchmark it in \fref{sec:cursor-benchmark}.
With the safe interface to serialised data, functions like \ensuremath{\Varid{sumLeaves}} and
\ensuremath{\Varid{mapLeaves}} are not burdensome to program.  The code for \ensuremath{\Varid{mapLeaves}} is shown
below.

\begin{hscode}\SaveRestoreHook
\column{B}{@{}>{\hspre}l<{\hspost}@{}}%
\column{3}{@{}>{\hspre}l<{\hspost}@{}}%
\column{5}{@{}>{\hspre}l<{\hspost}@{}}%
\column{24}{@{}>{\hspre}l<{\hspost}@{}}%
\column{34}{@{}>{\hspre}l<{\hspost}@{}}%
\column{56}{@{}>{\hspre}l<{\hspost}@{}}%
\column{E}{@{}>{\hspre}l<{\hspost}@{}}%
\>[B]{}\Varid{mapLeaves}\mathbin{::}(\Conid{Int}\to \Conid{Int})\to \Conid{Packed}\mathop{{\kern 1pt}'}[\mskip1.5mu \Conid{Tree}\mskip1.5mu]\mathbin{⊸}\Conid{Packed}\mathop{{\kern 1pt}'}[\mskip1.5mu \Conid{Tree}\mskip1.5mu]{}\<[E]%
\\
\>[B]{}\Varid{mapLeaves}\;\Varid{fn}\;\Varid{pt}\mathrel{=}\Varid{newBuffer}\;(\Varid{extract}\mathbin{∘}\Varid{go}\;\Varid{pt}){}\<[E]%
\\
\>[B]{}\hsindent{3}{}\<[3]%
\>[3]{}\mathbf{where}{}\<[E]%
\\
\>[3]{}\hsindent{2}{}\<[5]%
\>[5]{}\Varid{extract}\;(\Varid{inp},\Varid{outp})\mathrel{=}\mathbf{case}\;\Varid{done}\;\Varid{inp}\;\mathbf{of}\;()\to \Varid{finish}\;\Varid{outp}{}\<[E]%
\\
\>[3]{}\hsindent{2}{}\<[5]%
\>[5]{}\Varid{go}\mathbin{::}\Conid{Packed}\;(\Conid{Tree}\mathop{{\kern 1pt}'\!\!:}\Varid{r})\mathbin{⊸}\Conid{Needs}\;(\Conid{Tree}\mathop{{\kern 1pt}'\!\!:}\Varid{r})\;\Varid{t}\mathbin{⊸}(\Conid{Packed}\;\Varid{r},\Conid{Needs}\;\Varid{r}\;\Varid{t}){}\<[E]%
\\
\>[3]{}\hsindent{2}{}\<[5]%
\>[5]{}\Varid{go}\;\Varid{p}\mathrel{=}\Varid{caseTree}\;\Varid{p}\;{}\<[24]%
\>[24]{}(\lambda \Varid{p}\;\Varid{o}\to {}\<[34]%
\>[34]{}\mathbf{let}\;(\Varid{x},\Varid{p'})\mathrel{=}\Varid{read}\;\Varid{p}\;{}\<[56]%
\>[56]{}\mathbf{in}\;(\Varid{p'},\Varid{writeLeaf}\;(\Varid{fn}\;\Varid{x})\;\Varid{o}))\;{}\<[E]%
\\
\>[24]{}(\lambda \Varid{p}\;\Varid{o}\to {}\<[34]%
\>[34]{}\mathbf{let}\;(\Varid{p'},\Varid{o'})\mathrel{=}\Varid{go}\;\Varid{p}\;(\Varid{writeBranch}\;\Varid{o})\;\mathbf{in}\;\Varid{go}\;\Varid{p'}\;\Varid{o'}){}\<[E]%
\ColumnHook
\end{hscode}\resethooks

\subsubsection{A version without linear types}\label{sec:st-cursors}

How would we build the same thing in Haskell without linear types?  It may appear
that the ST monad is a suitable choice:

\begin{hscode}\SaveRestoreHook
\column{B}{@{}>{\hspre}l<{\hspost}@{}}%
\column{E}{@{}>{\hspre}l<{\hspost}@{}}%
\>[B]{}\Varid{writeST}\mathbin{::}\Conid{Storable}\;\Varid{a}\Rightarrow \Varid{a}\to \Conid{Needs'}\;\Varid{s}\;(\Varid{a}\mathop{{\kern 1pt}'\!\!:}\Varid{r})\;\Varid{t}\to \Conid{ST}\;\Varid{s}\;(\Conid{Needs'}\;\Varid{s}\;\Varid{r}\;\Varid{t}){}\<[E]%
\ColumnHook
\end{hscode}\resethooks
Here we use the same typestate associated with a \ensuremath{\Conid{Needs}} pointer, while also
associating its mutable state with the \ensuremath{\Conid{ST}} session indexed by \ensuremath{\Varid{s}}.
Unfortunately, not only do we have the same trouble with freezing in the absence
of linearity (\ensuremath{\Varid{unsafeFreeze}}, \fref{sec:freezing-arrays}), we also
have an {\em additional} problem not present with arrays:
namely, a non-linear use of a \ensuremath{\Conid{Needs}} pointer can ruin our type-safe
deserialisation guarantee!
For example, we can write a \ensuremath{\Conid{Leaf}} and a \ensuremath{\Conid{Branch}} to the same pointer in an
interleaved fashion.  Both will place a tag at byte 0; but the leaf will place
an integer in bytes 1-9, while the branch will place another tag at byte 1.
We can receive a corrupted 8-byte integer, clobbered by a tag from an
interleaved ``alternate future''.

Fixing this problem would require switching to an indexed monad with additional
type-indices that model the typestate of all accessible pointers, which would in
turn need to have static, type-level identifiers.  That is, it would require
{\em encoding} linearity after all, but in a way which would become very
cumbersome as soon as several buffers are involved.

\subsubsection{Benchmarking compiler optimisations}
\label{sec:cursor-benchmark}

\begin{figure}
  \begin{minipage}{\textwidth}
  \includegraphics[width=0.49 \textwidth]{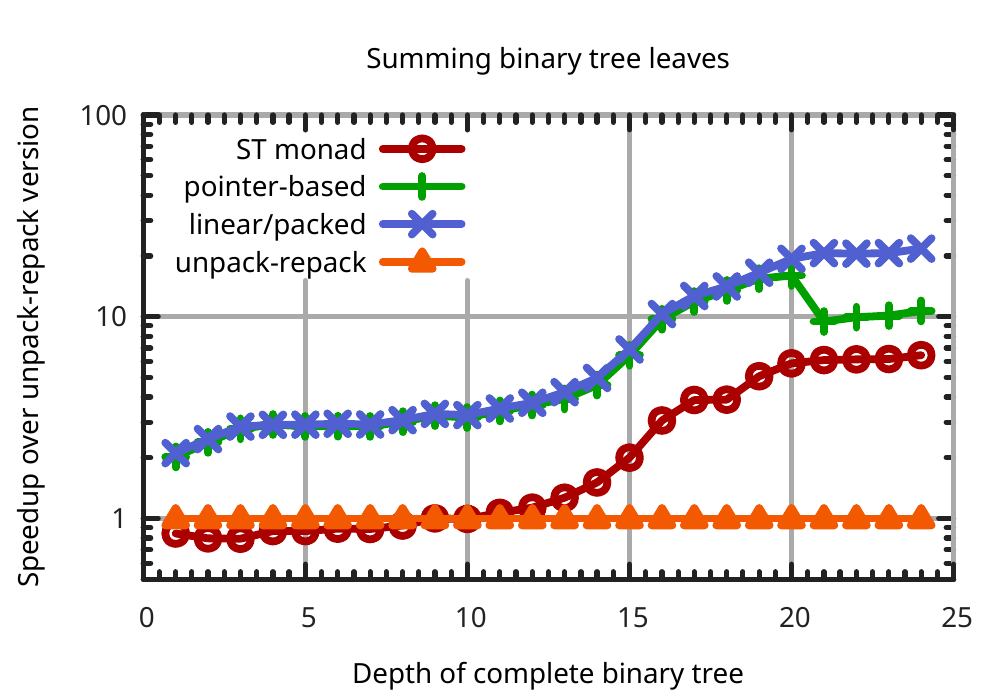}
  \includegraphics[width=0.49 \textwidth]{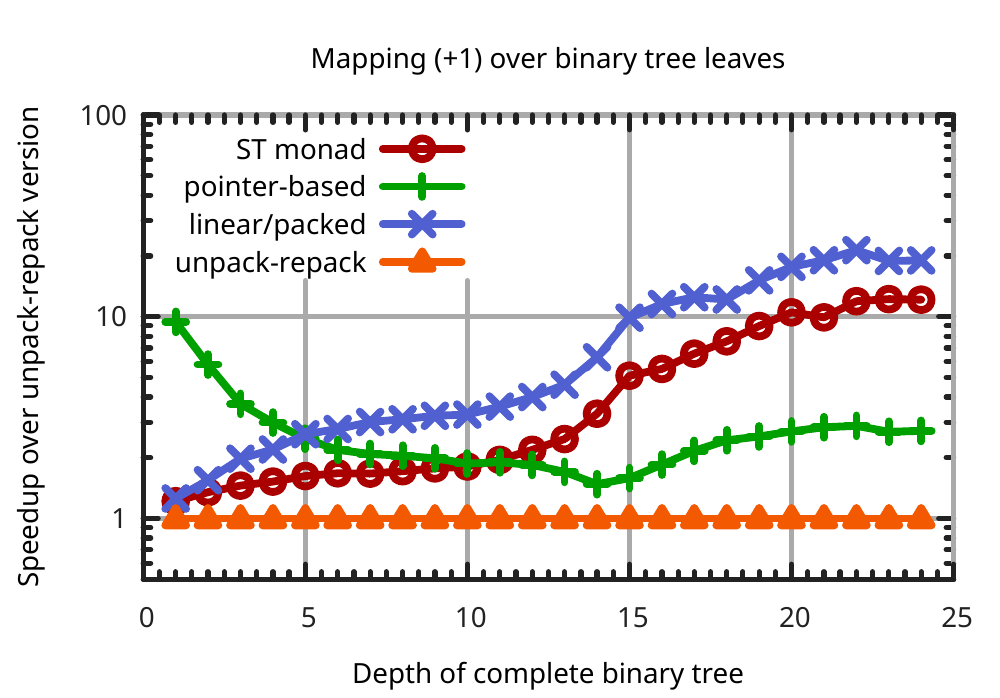}
  \end{minipage}
\caption{Speedup of operating directly on serialised data, either using
  linear-types or the ST monad, as compared to fully unpacking, processing,
  and repacking the data.  For reference, a ``pointer-based'' version is also
  included, which doesn't operate on serialised data at all, but instead normal
  heap objects --- it represents the hypothetical performance of ``unpack-repack''
  if (de)serialisation were instantaneous.}
\label{fig:pack-bench}
\end{figure}

Finally, as shown in \fref{fig:pack-bench}, there are some unexpected
performance consequences from using a linear versus a monadic, ST style in
\ghc.
Achieving allocation-free loops in \ghc{} is always a challenge --- 
tuple types and numeric types are lazy and ``boxed'' as heap objects by default.
As we saw in the \ensuremath{\Varid{sumLeaves}} and \ensuremath{\Varid{mapLeaves}} examples, each recursive call returned a tuple
of a result and a new pointer.  In a monadic formulation, an expression of type
\ensuremath{\Varid{m}\;\Varid{a}}, for \ensuremath{\Conid{Monad}\;\Varid{m}}, implies that the ``result'' type \ensuremath{\Varid{a}}, of kind \ensuremath{\mathbin{*}}, must be a
{\em lifted} type.  Nevertheless, in some situations, for some monads, the
optimiser is able to deforest data constructors returned by monadic actions.
In the particular case of \ensuremath{\Varid{fold}} and \ensuremath{\Varid{map}} operations over serialised trees,
unfortunately, we are currently unable to eliminate all allocation from
\ensuremath{\Conid{ST}}-based implementations of the algorithms.

For the linearly-typed code, however, we have more options.  \ghc{} has the ability to
directly express unboxed values such as a tuple \ensuremath{(\mathbin{\#}\Conid{Int}\mathbin{\#},\Conid{Double}\mathbin{\#}\mathbin{\#})}, which
fills two registers and inhabits an unboxed kind distinct from \ensuremath{\mathbin{*}}.
In fact, the type of a combinator like \ensuremath{\Varid{caseTree}} is a good fit for the recent
``levity polymorphism'' addition to
\ghc{}~\cite{eisenberg_levity_2017}.  
Thus we permit the branches of the \ensuremath{\mathbf{case}} to return unlifted, unboxed types, and
give \ensuremath{\Varid{caseTree}} a more general type:
{
\def\mathindent{0pt}
\begin{hscode}\SaveRestoreHook
\column{B}{@{}>{\hspre}l<{\hspost}@{}}%
\column{3}{@{}>{\hspre}l<{\hspost}@{}}%
\column{E}{@{}>{\hspre}l<{\hspost}@{}}%
\>[B]{}\Varid{caseTree}\mathbin{::}∀(\Varid{rep}\mathbin{::}\Conid{RuntimeRep})\;(\Varid{res}\mathbin{::}\Conid{TYPE}\;\Varid{rep})\;\Varid{b}. {}\<[E]%
\\
\>[B]{}\hsindent{3}{}\<[3]%
\>[3]{}\Conid{Packed}\;(\Conid{Tree}\mathop{{\kern 1pt}'\!\!:}\Varid{b})\to _{\Varid{p}}\;(\Conid{Packed}\;(\Conid{Int}\mathop{{\kern 1pt}'\!\!:}\Varid{b})\to _{\Varid{p}}\;\Varid{res})\to (\Conid{Packed}\;(\Conid{Tree}\mathop{{\kern 1pt}'\!\!:}\Conid{Tree}\mathop{{\kern 1pt}'\!\!:}\Varid{b})\to _{\Varid{p}}\;\Varid{res})\to \Varid{res}{}\<[E]%
\ColumnHook
\end{hscode}\resethooks
}
This works because we do not need to {\em call} a function with \ensuremath{\Varid{res}} as
argument (and thus unknown calling conventions) only return it.
Using this approach, we were able to ensure by construction that the
``linear/packed'' implementations in \fref{fig:pack-bench} were completely
non-allocating, rather than depending on the optimiser.
This results in better performance for the linear, compared to monadic version
of the serialised-data transformations.

The basic premise of \fref{fig:pack-bench} is that a machine in the network
receives, processes, and transmits serialized data (trees).
We consider two simple benchmarks: \ensuremath{\Varid{sumLeaves}} and \ensuremath{\Varid{mapTree}\;(\mathbin{+}\mathrm{1})}.
The baseline is the traditional approach: deserialise, transform, and reserialise,
the ``unpack-repack'' line in the plots.  Compared to this baseline,
{\em processing the data directly in its serialised form results in speedups of over
20$\times$ on large trees}.
What linear types makes safe, is also efficient.

The experiment was conducted on a Xeon E5-2699 CPU (2.30GHz, 64GB memory) using
our modified version of \ghc{} 8.2 (\fref{sec:impl}).  Each
data point was measured by performing many trials and taking a linear regression
of iteration count against time\footnote{using the criterion library~\cite{osullivan_criterion_2013}}.
This process allows for accurate measurements of both small and large times.  The
baseline unpack-repack tree-summing times vary from 25ns to 1.9 seconds at
depths 1 and 24 respectively.  Likewise, the baseline mapping times vary from
215ns to 2.93 seconds.
We use a simple contiguous implementation of buffers for
serialisation\footnote{A full, practical implementation should include growable
  or doubling buffers.}.
At depth 20, one copy of the tree takes around 10MB, and towards the right half
of the plot we see tree size exceeding cache size.

\subsection{Sockets with type-level state}
\label{sec:sockets}

The \textsc{bsd} socket \textsc{api} is a standard, if not \emph{the}
standard, through which computers connect over networks. It involves a
series of actions which must be performed in order: on the
server-side, a freshly created socket must be \emph{bound} to an
address, then start \emph{listening} incoming traffic, then
\emph{accept} connection requests; said connection is returned as a
new socket, this new socket can now \emph{receive} traffic. One reason
for having that many steps is that the precise sequence of
actions is protocol-dependent. For \textsc{tcp} traffic you would do as
described, but for \textsc{udp}, which does not need connections, you
would not accept a connection but receive messages directly.

The \texttt{socket} library for Haskell, exposes precisely this sequence of
actions.
Programming with it is exactly as clumsy as socket libraries for other
languages: after each action, the state of the socket changes, as do the
permissible actions, but these states are invisible in the types.
Better is to track the state of sockets in the type, akin to a typestate
analysis~\cite{strom_typestate_1983}.
In the \ensuremath{\Conid{File}} \textsc{api} of \fref{sec:io-protocols}, we made files
safer to use at the cost of having to thread a file handle
explicitely: each function consumes a file handle and returns a fresh
one. We can make this cost into an opportunity: we have the option of
returning a handle \emph{with a different type} from that of the
handle we consumed! So by adjoining a phantom type to sockets to track
their state, we can effectively encode the proper sequencing of socket
actions.

As an illustration, we implemented wrapper around the \textsc{api} of the
\texttt{socket} library. For concision, this wrapper is
specialised for the case of \textsc{tcp}.

\begin{hscode}\SaveRestoreHook
\column{B}{@{}>{\hspre}l<{\hspost}@{}}%
\column{10}{@{}>{\hspre}l<{\hspost}@{}}%
\column{12}{@{}>{\hspre}l<{\hspost}@{}}%
\column{13}{@{}>{\hspre}l<{\hspost}@{}}%
\column{E}{@{}>{\hspre}l<{\hspost}@{}}%
\>[B]{}\mathbf{data}\;\Conid{State}\mathrel{=}\Conid{Unbound}\mid \Conid{Bound}\mid \Conid{Listening}\mid \Conid{Connected}{}\<[E]%
\\
\>[B]{}\mathbf{data}\;\Conid{Socket}\;(\Varid{s}\mathbin{::}\Conid{State}){}\<[E]%
\\
\>[B]{}\mathbf{data}\;\Conid{SocketAddress}{}\<[E]%
\\[\blanklineskip]%
\>[B]{}\Varid{socket}\mathbin{::}{}\<[12]%
\>[12]{}\varid{IO}_{\varid{L}}\;\mathrm{1}\;(\Conid{Socket}\;\Conid{Unbound}){}\<[E]%
\\
\>[B]{}\Varid{bind}\mathbin{::}{}\<[10]%
\>[10]{}\Conid{Socket}\;\Conid{Unbound}\mathbin{⊸}\Conid{SocketAddress}\to \varid{IO}_{\varid{L}}\;\mathrm{1}\;(\Conid{Socket}\;\Conid{Bound}){}\<[E]%
\\
\>[B]{}\Varid{listen}\mathbin{::}\Conid{Socket}\;\Conid{Bound}\mathbin{⊸}\varid{IO}_{\varid{L}}\;\mathrm{1}\;(\Conid{Socket}\;\Conid{Listening}){}\<[E]%
\\
\>[B]{}\Varid{accept}\mathbin{::}{}\<[12]%
\>[12]{}\Conid{Socket}\;\Conid{Listening}\mathbin{⊸}\varid{IO}_{\varid{L}}\;\mathrm{1}\;(\Conid{Socket}\;\Conid{Listening},\Conid{Socket}\;\Conid{Connected}){}\<[E]%
\\
\>[B]{}\Varid{connect}\mathbin{::}{}\<[13]%
\>[13]{}\Conid{Socket}\;\Conid{Unbound}\mathbin{⊸}\Conid{SocketAddress}\to \varid{IO}_{\varid{L}}\;\mathrm{1}\;(\Conid{Socket}\;\Conid{Connected}){}\<[E]%
\\
\>[B]{}\Varid{send}\mathbin{::}\Conid{Socket}\;\Conid{Connected}\mathbin{⊸}\Conid{ByteString}\to \varid{IO}_{\varid{L}}\;\mathrm{1}\;(\Conid{Socket}\;\Conid{Connected},\Conid{Unrestricted}\;\Conid{Int}){}\<[E]%
\\
\>[B]{}\Varid{receive}\mathbin{::}\Conid{Socket}\;\Conid{Connected}\mathbin{⊸}\varid{IO}_{\varid{L}}\;\mathrm{1}\;(\Conid{Socket}\;\Conid{Connected},\Conid{Unrestricted}\;\Conid{ByteString}){}\<[E]%
\\
\>[B]{}\Varid{close}\mathbin{::}∀\Varid{s}. \Conid{Socket}\;\Varid{s}\to \varid{IO}_{\varid{L}}\;\omega\;(){}\<[E]%
\ColumnHook
\end{hscode}\resethooks

This linear socket \textsc{api} is very similar to that of files: we use the
\ensuremath{\varid{IO}_{\varid{L}}} monad in order to enforce linear use of sockets. The difference
is the argument to \ensuremath{\Conid{Socket}}, which represents the current state of the
socket and is used to limit the functions which apply to a socket
at a given time.

\paragraph{Implementing the linear socket \textsc{api}}
Our socket \textsc{api} has been tested by writing a small
echo-server. The \textsc{api} is implemented as a wrapper around the
\texttt{socket} library. Each function wrapped takes half-a-dozen
lines of code, of type annotation and coercions between \ensuremath{\Conid{IO}} and
\ensuremath{\varid{IO}_{\varid{L}}}\footnote{Since our implementation of \HaskeLL{} does not yet
  have multiplicity-polymorphism, we had to fake it with type families
  and \textsc{gadts}}. There is no computational behaviour besides
error recovery.

It would have been too restrictive to limit the typestate to enforce
the usage protocol of \textsc{tcp}. We do not intend for a new set of
wrapper functions to be implemented for each protocol. Instead the
wrappers are implemented with a generic type-state evolving according
to the rules of a deterministic automaton. Each protocol can define
it's own automaton, which is represented as a set of instances of
a type class.

\subsection{Pure bindings to impure \textsc{api}s}
\label{sec:spritekit}

In Haskell SpriteKit, \citet{chakravarty_spritekit_2017} have a different kind
of problem. They build a pure interface for graphics, in the same style
as the Elm programming language~\cite{czaplicki_elm_2012}, but implement it in terms
of an existing imperative graphical interface engine.

Basically, the pure interface takes an update function \ensuremath{\Varid{u}\mathbin{:}\Conid{Scene}\to \Conid{Scene}} which is
tasked with returning the next state that the screen will display.
The scene is first converted to a
pure tree where each node keeps, along with the pure data, a pointer
to its imperative counterpart when it applies, or \ensuremath{\Conid{Nothing}} for new
nodes.
\begin{hscode}\SaveRestoreHook
\column{B}{@{}>{\hspre}l<{\hspost}@{}}%
\column{E}{@{}>{\hspre}l<{\hspost}@{}}%
\>[B]{}\mathbf{data}\;\Conid{Node}\mathrel{=}\Conid{Node}\;\{\mskip1.5mu \Varid{payload}\mathbin{::}\Conid{Int},\Varid{ref}\mathbin{::}\Conid{Maybe}\;(\Conid{IORef}\;\Conid{ImperativeNode}),\Varid{children}\mathbin{::}[\mskip1.5mu \Conid{Node}\mskip1.5mu]\mskip1.5mu\}{}\<[E]%
\ColumnHook
\end{hscode}\resethooks

On each frame, SpriteKit applies \ensuremath{\Varid{u}} to the current scene, and checks
if a node \ensuremath{\Varid{n}} was updated. If it was, it applies the update directly
onto \ensuremath{\Varid{ref}\;\Varid{n}} or creates a new imperative node.

Things can go wrong though: if the update function {\em duplicates}
any proxy node, one gets the situation where two nodes \ensuremath{\Varid{n}} and
\ensuremath{\Varid{n'}} can point to the same imperative source \ensuremath{\Varid{ref}\;\Varid{n}\mathrel{=}\Varid{ref}\;\Varid{n'}}, but
have different payloads. In this situation the \ensuremath{\Conid{Scene}} has become
inconsistent and the behaviour of SpriteKit is unpredictable.

In the {\sc api} of \citet{chakravarty_spritekit_2017}, the burden of checking
non-duplication is on the programmer.  Using linear types, we can
switch that burden to the compiler: we change the update function to
type \ensuremath{\Conid{Scene}\mathbin{⊸}\Conid{Scene}}, and the \ensuremath{\Varid{ref}} field is made linear too.  Thanks
to linearity, no reference can be duplicated: if a node is copied, the
programmer must choose which one will correspond to the old imperative
counterpart and which will be new.

We implemented such an \textsc{api} in our implementation of
\HaskeLL{}. The library-side code does not use any linear code, the
\ensuremath{\Conid{Node}}s are actually used unrestrictedly. Linearity is only imposed on
the user of the interface, in order to enforced the above restriction.

\section{Related work}
\label{sec:related}
\subsection{Linearity via arrows vs. linearity via kinds}
\label{sec:lin-arrow}

There are two possible choices to indicate the distinction between
linear and unrestricted objects.  Our choice is to use the arrow
type. That is, we have both a linear arrow to introduce linear objects
in the environment, and an unrestricted arrow to introduce
unrestricted objects. This choice is featured in the work of
\citet{mcbride_rig_2016} and \citet{ghica_bounded_2014} and is
ultimately inspired by Girard's presentation of linear logic, which
features only linear arrows, and where the unrestricted arrow $A → B$
is encoded as $!A ⊸ B$.

Another popular choice
\cite{wadler_linear_1990,mazurak_lightweight_2010,morris_best_2016,tov_practical_2011}
is to separate types into two kinds: a linear kind and an unrestricted
kind. Values with a type whose kind is linear are linear, and the
others are unrestricted. (Thus in particular such systems feature
``linear arrows'', but they have a completely different interpretation
from ours.) While this does not match linear logic (there is no such
thing as a linear proposition), it is attractive on the surface
because, intuitively, some types are inherently linear (file handles,
updateable arrays, etc.) and some types are inherently unrestricted
(\ensuremath{\Conid{Int}}, \ensuremath{\Conid{Bool}}, etc.).  However, after scratching the surface we have
discovered that ``linearity via arrows'' has an edge over ``linearity
via kinds''.

\paragraph{Better code reuse}  When retrofitting linear types
  in an existing language, it is important to share as much code as
  possible between linear and non-linear code. In a system with
  linearity on arrows, the subsumption relation (linear arrows subsume
  unrestricted arrows) and the scaling of context in the application
  rule mean that much linear code can be used as-is from unrestricted
  code, and be properly promoted. Indeed, assuming lists as defined in
  \fref{sec:compatibility} and:
  \begin{hscode}\SaveRestoreHook
\column{B}{@{}>{\hspre}l<{\hspost}@{}}%
\column{5}{@{}>{\hspre}l<{\hspost}@{}}%
\column{31}{@{}>{\hspre}l<{\hspost}@{}}%
\column{E}{@{}>{\hspre}l<{\hspost}@{}}%
\>[5]{}(\plus )\mathbin{::}[\mskip1.5mu \Varid{a}\mskip1.5mu]\mathbin{⊸}[\mskip1.5mu \Varid{a}\mskip1.5mu]\mathbin{⊸}[\mskip1.5mu \Varid{a}\mskip1.5mu]{}\<[31]%
\>[31]{}\mbox{\onelinecomment  Append two lists}{}\<[E]%
\\
\>[5]{}\Varid{cycle}\mathbin{::}[\mskip1.5mu \Varid{a}\mskip1.5mu]\mathbin{→}[\mskip1.5mu \Varid{a}\mskip1.5mu]{}\<[31]%
\>[31]{}\mbox{\onelinecomment  Repeat a list, infinitely}{}\<[E]%
\ColumnHook
\end{hscode}\resethooks
  The following definition type-checks, even though \ensuremath{\plus } is applied to
  unrestricted values and used in an unrestricted context.
  \begin{hscode}\SaveRestoreHook
\column{B}{@{}>{\hspre}l<{\hspost}@{}}%
\column{5}{@{}>{\hspre}l<{\hspost}@{}}%
\column{E}{@{}>{\hspre}l<{\hspost}@{}}%
\>[5]{}\Varid{f}\mathbin{::}[\mskip1.5mu \Varid{a}\mskip1.5mu]\mathbin{→}[\mskip1.5mu \Varid{a}\mskip1.5mu]\mathbin{→}[\mskip1.5mu \Varid{a}\mskip1.5mu]{}\<[E]%
\\
\>[5]{}\Varid{f}\;\Varid{xs}\;\Varid{ys}\mathrel{=}\Varid{cycle}\;(\Varid{xs}\plus \Varid{ys}){}\<[E]%
\ColumnHook
\end{hscode}\resethooks
  In contrast, in a two-kind system, a function must declare the
  \emph{exact} linearity of its return value. Consequently, to make a
  function promotable from linear to unrestriced, its declaration must
  use polymorphism over kinds.

  As seen in \fref{sec:programming-intro}, in \HaskeLL{} reuse of linear
  code extends to datatypes: the usual parametric datatypes (lists,
  pairs, etc.) work both with linear and unrestricted values. On the
  contrary, if linearity depends on the kind, then if a linear value
  is contained in a type, the container type must be linear
  too. (Indeed, an unrestricted container could be discarded or
  duplicated, and its contents with it.) Consequently, sharing data
  types also requires polymorphism.  For example, in a two-kinds
  system, the \ensuremath{\Conid{List}} type may look like so, if one assumes a that
  \ensuremath{\Conid{Type}\;\mathrm{1}} is the kind of linear types and \ensuremath{\Conid{Type}\;\omega} is the kind of
  unrestricted types.
  \begin{hscode}\SaveRestoreHook
\column{B}{@{}>{\hspre}l<{\hspost}@{}}%
\column{5}{@{}>{\hspre}l<{\hspost}@{}}%
\column{E}{@{}>{\hspre}l<{\hspost}@{}}%
\>[5]{}\mathbf{data}\;\Conid{List}\;(\Varid{p}\mathbin{::}\Conid{Multiplicity})\;(\Varid{a}\mathbin{::}\Conid{Type}\;\Varid{p})\mathbin{::}\Conid{Type}\;\Varid{p}\mathrel{=}[\mskip1.5mu \mskip1.5mu]\mid \Varid{a}\mathbin{:}(\Conid{List}\;\Varid{p}\;\Varid{m}\;\Varid{a}){}\<[E]%
\ColumnHook
\end{hscode}\resethooks
  The above declaration ensures that the linearity of the list
  inherits the linearity of the contents. A linearity-polymorphic
  \ensuremath{(\plus )} function could have the definition:
  \begin{hscode}\SaveRestoreHook
\column{B}{@{}>{\hspre}l<{\hspost}@{}}%
\column{5}{@{}>{\hspre}l<{\hspost}@{}}%
\column{E}{@{}>{\hspre}l<{\hspost}@{}}%
\>[5]{}(\plus )\mathbin{::}\Conid{List}\;\Varid{p}\;\Varid{a}\to \Conid{List}\;\Varid{p}\;\Varid{a}\to \Conid{List}\;\Varid{p}\;\Varid{a}{}\<[E]%
\\
\>[5]{}[\mskip1.5mu \mskip1.5mu]\plus \Varid{xs}\mathrel{=}\Varid{xs}{}\<[E]%
\\
\>[5]{}(\Varid{x}\mathbin{:}\Varid{xs})\plus \Varid{ys}\mathrel{=}\Varid{x}\mathbin{:}(\Varid{xs}\plus \Varid{ys}){}\<[E]%
\ColumnHook
\end{hscode}\resethooks
  Compared to our append function, the type of the above requires
  multiplicity polymorphism \ensuremath{\Varid{p}}. Additionnally, the above function
  cannot (and should not) mix linear and unrestricted
  lists. Indeed because multiplicity is attached to the type of
  elements it must be the same for both arguments and the returned
  value.

  Note that, in the above, we parameterize over multiplicities instead
  of parameterizing over kinds directly, as is customary in the
  literature. We do so because it fits better \ghc{}, whose kinds are
  already parameterized over a so-called
  levity~\cite{eisenberg_levity_2017}.

  \paragraph{\textsc{Ats}} The \textsc{ats} language has a unique take
  on linear types \cite{zhu_linear-views_2005}, which can be
  classified as linearity via kinds and does not have
  polymorphism.\textsc{Ats} has a notion of \emph{stateful views}: views are
  linear values without run-time representation and are meant to
  track the state of pointers (\emph{e.g.} whether they are
  initialised). Pointers themselves remain unrestricted values: only
  views are linear.

  \paragraph{Dependent types}
  Linearity on the arrow meshes better with dependent types (see
  \fref{sec:extending-multiplicities}).  Indeed, consider a typical
  predicate over files (\ensuremath{\Conid{P}\mathbin{:}\Conid{File}\mathbin{→}\mathbin{∗}}). It may need to mention its
  argument several times to relate several possible sequences of
  operations on the file. While this is not a problem in our system,
  the function \ensuremath{\Conid{P}} is not expressible if \ensuremath{\Conid{File}} is intrinsically
  linear. Leaving the door open to dependent types is crucial to us,
  as this is currently explored as a possible extension to
  \ghc{}~\cite{weirich_dependent-haskell_2017}.

\paragraph{Linear values}
Yet, an advantage of ``linearity via kinds'' is the possibility to directly
declare the linearity of values returned by a function -- not just that
of the argument of a function. In contrast, in our system if one wants
to indicate that a returned value is linear, we have to
use a double-negation trick. That is, given $f : A → (B ⊸ !r) ⊸ r$,
then $B$ can be used a single time in the (single) continuation, and
effectively $f$ ``returns'' a single $B$. One can obviously declare a
type for linear values \ensuremath{\Conid{Linear}\;\Varid{a}\mathrel{=}(\Varid{a}\mathbin{⊸}\mathbin{!}\Varid{r})\mathbin{⊸}\Varid{r}} and chain
\ensuremath{\Conid{Linear}}-returning functions with appropriate combinators.  In fact,
as explained in \fref{sec:linear-io}, the cost of the double negation
almost entirely vanishes in the presence of an ambient monad.

\subsection{Other variants of ``linearity on the arrow''}
\label{sec:related-type-systems}

The \calc{} type system is heavily inspired from the work of
\citet{ghica_bounded_2014} and \citet{mcbride_rig_2016}. Both of them
present a type system where arrows are annotated with the multiplicty
of the the argument that they require, and where the multiplicities
form a semi-ring.

In contrast with \calc, \citeauthor{mcbride_rig_2016} uses a
multiplicity-annotated type judgement $Γ ⊢_ρ t : A$, where $ρ$
represents the multiplicity of $t$. So, in
\citeauthor{mcbride_rig_2016}'s system, when an unrestricted value is
required, instead of computing $ωΓ$, it is enough to check that
$ρ=ω$. The problem is that this check is arguably too coarse, and
results in the judgement $⊢_ω λx. (x,x) : A ⊸ (A,A)$ being derivable.
This derivation is not desirable: it implies that there cannot be
reusable definitions of linear functions. In terms of linear logic~\cite{girard_linear_1987},
\citeauthor{mcbride_rig_2016} makes the natural function of type $!(A⊸B) ⟹ !A⊸!B$
into an isomorphism.
In that respect, our system is closer to
\citeauthor{ghica_bounded_2014}'s.

The essential differences between our system and that of
\citeauthor{ghica_bounded_2014} is that we support
multiplicity-polymorphism and datatypes. In particular our \ensuremath{\mathbf{case}} rule
is novel.

The literature on so-called
coeffects~\cite{petricek_coeffects_2013,brunel_coeffect_core_2014}
uses type systems similar to \citeauthor{ghica_bounded_2014}, but
with a linear arrow and multiplicities carried by the exponential
modality instead. \Citet{brunel_coeffect_core_2014}, in particular,
develops a Krivine-style realisability model for such a calculus. We are not
aware of an account of Krivine realisability for lazy languages, hence
this work is not directly applicable to \calc.

\subsection{Uniqueness and ownership typing}
\label{sec:uniqueness}

The literature
contains many proposals for uniqueness (or ownership) types (in contrast with
linear types).
Prominent representative languages with uniqueness types include
Clean~\cite{barendsen_uniqueness_1993} and
Rust~\cite{matsakis_rust_2014}. \HaskeLL, on the other hand, is
designed around linear types based on linear
logic~\cite{girard_linear_1987}.

Idris \cite{brady_idris_2013} features uniqueness types, which have
been used, in particular, to enforce communication
protocols~\cite{brady_uniqueness_2017}. Uniqueness types, in Idris,
are being replaced by linear types based on
\textsc{qtt}~\cite{atkey_qtt_2017}, a variant of
\citet{mcbride_rig_2016}.

Linear types and uniqueness types are, at their core, dual: whereas a linear type is
a contract that a function uses its argument exactly once
even if the call's context can share a linear argument as many times as it
pleases, a uniqueness type ensures that the argument of a function is
not used anywhere else in the expression's context even if the callee
can work with the argument as it pleases.
Seen as a system of constraints, uniqueness typing is a {\em non-aliasing
analysis} while linear typing provides a {\em cardinality analysis}. The
former aims at in-place updates and related optimisations, the latter
at inlining and fusion. Rust and Clean largely explore the
consequences of uniqueness on in-place update; an in-depth exploration
of linear types in relation with fusion can be found
in~\citet{bernardy_composable_2015};
see also the discussion in \fref{sec:fusion}.

Because of this weak duality, we could have
retrofitted uniqueness types to Haskell. But several points
guided our choice of designing \HaskeLL{} around linear
logic instead: (a) functional languages have more use
for fusion than in-place update (if the fact that \textsc{ghc} has a cardinality
analysis but no non-aliasing analysis is any indication); (b) there is a
wealth of literature detailing the applications of linear
logic — see \fref{sec:applications}; (c) and decisively, linear type systems are
conceptually simpler than uniqueness type systems, giving a
clearer path to implementation in \textsc{ghc}.

\paragraph{Rust \& Borrowing}

In \HaskeLL{} we need to thread linear variables throughout the
program.  Even
though this burden could be alleviated using syntactic sugar, Rust uses
instead a type-system feature for this purpose: \emph{borrowing}.
Borrowed values differ from owned values in that they can
be used in an unrestricted fashion, albeit in a \emph{delimited
  scope}.

Borrowing does not come without a cost, however: if a
function \texttt{f} borrows a value \texttt{v} of type \texttt{T},
then the caller of the function \emph{must} retain \texttt{v} alive
until \texttt{f} has returned; the consequence is that Rust cannot, in
general, perform tail-call elimination, crucial to the operation
behaviour of many functional programs, as some resources must be
released \emph{after} \texttt{f} has returned.

The reason that Rust programs depend so much on borrowing is that
unique values are the default. \HaskeLL{} aims to hit a different
point in the design space where regular non-linear expressions are the
norm, yet gracefully scaling up investing extra effort to enforce
linearity invariants is possible.
Nevertheless, we discuss in \fref{sec:extending-multiplicities} how to
extend \HaskeLL{} with borrowing.

\subsection{Linearity via monads}

\citet{launchbury_st_1995} taught us
a conceptually simple approach to lifetimes: the \ensuremath{\Conid{ST}} monad. It has
a phantom type parameter \ensuremath{\Varid{s}} (the \emph{region}) that is instantiated once at the
beginning of the computation by a \ensuremath{\Varid{runST}} function of type:
\begin{hscode}\SaveRestoreHook
\column{B}{@{}>{\hspre}l<{\hspost}@{}}%
\column{3}{@{}>{\hspre}l<{\hspost}@{}}%
\column{E}{@{}>{\hspre}l<{\hspost}@{}}%
\>[3]{}\Varid{runST}\mathbin{::}∀\Varid{a}. (∀\Varid{s}. \Conid{ST}\;\Varid{s}\;\Varid{a})\to \Varid{a}{}\<[E]%
\ColumnHook
\end{hscode}\resethooks
This way, resources that are allocated during the computation, such
as mutable cell references, cannot escape the dynamic scope of the call
to \ensuremath{\Varid{runST}} because they are themselves tagged with the same phantom
type parameter.

\paragraph{Region-types}

With region-types such as \ensuremath{\Conid{ST}}, we cannot express typestates, but this
is sufficient to offer a safe \textsc{api} for freezing array or
ensuring that files are eventually closed.
This simplicity (one only needs rank-2 polymorphism)
comes at a cost: we've already mentionned in
\fref{sec:freezing-arrays} that it forces operations to be more
sequentialised than need be, but more importantly, it does not
support \emph{prima facie} the interaction of nested regions.

\citet{kiselyov_regions_2008} show that it is possible to promote
resources in parent regions to resources in a subregion. But this is
an explicit and monadic operation, forcing an unnatural imperative
style of programming where order of evaluation is explicit.
The HaskellR project~\cite{boespflug_project_2014} uses monadic
regions in the style of \citeauthor{kiselyov_regions_2008} to safely
synchronise values shared between two different garbage collectors for
two different languages. \Citeauthor{boespflug_project_2014} report
that custom monads make writing code at an interactive prompt
difficult, compromises code reuse, forces otherwise pure functions to
be written monadically and rules out useful syntactic facilities like
view patterns.

In contrast, with linear types, values in two regions hence can safely
be mixed: elements can be moved from one data structure (or heap) to
another, linearly, with responsibility for deallocation transferred
along.

\paragraph{Idris's dependent indexed monad}
To go beyond simple regions, Idris \cite{brady_idris_2013} introduces  a generic way to add typestate on top of a monad, the
\ensuremath{\Conid{ST}} indexed monad transformer\footnote{See \eg
  \url{http://docs.idris-lang.org/en/latest/st/index.html}. Where you
  will also discover that \ensuremath{\Conid{ST}} is actually defined in terms of a more
  primitive \ensuremath{\Conid{STrans}}}. The basic
idea is that everything which must be single-threaded~--~and that we
would track with linearity~--~become part of the state of the
monad. For instance, coming back to the sockets of \fref{sec:sockets},
the type of \ensuremath{\Varid{bind}} would be as follows:
\begin{hscode}\SaveRestoreHook
\column{B}{@{}>{\hspre}l<{\hspost}@{}}%
\column{3}{@{}>{\hspre}l<{\hspost}@{}}%
\column{E}{@{}>{\hspre}l<{\hspost}@{}}%
\>[3]{}\Varid{bind}\mathbin{::}(\Varid{sock}\mathbin{::}\Conid{Var})\to \Conid{SocketAddress}\to \Conid{ST}\;\Conid{IO}\;()\;[\mskip1.5mu \Varid{sock}\mathbin{:::}\Conid{Socket}\;\Conid{Unbound}:↦\Conid{Socket}\;\Conid{Bound}\mskip1.5mu]{}\<[E]%
\ColumnHook
\end{hscode}\resethooks
Where \ensuremath{\Varid{sock}} is a reference into the monads's state, and \ensuremath{\Conid{Socket}\;\Conid{Unbound}} is
the type of \ensuremath{\Varid{sock}} before \ensuremath{\Varid{bind}}, and \ensuremath{\Conid{Socket}\;\Conid{Bound}}, the type of
\ensuremath{\Varid{sock}} after \ensuremath{\Varid{bind}}.

Idris uses its dependent types to associate a state to the value of
its first argument. Dependent types are put to even greater use for
error management where the state of the socket depends on whether
\ensuremath{\Varid{bind}} succeeded or not:
\begin{hscode}\SaveRestoreHook
\column{B}{@{}>{\hspre}l<{\hspost}@{}}%
\column{3}{@{}>{\hspre}l<{\hspost}@{}}%
\column{12}{@{}>{\hspre}l<{\hspost}@{}}%
\column{E}{@{}>{\hspre}l<{\hspost}@{}}%
\>[3]{}\mbox{\onelinecomment  In Idris, \ensuremath{\Varid{bind}} uses a type-level function (\ensuremath{\Varid{or}}) to handle errors}{}\<[E]%
\\
\>[3]{}\Varid{bind}\mathbin{::}{}\<[12]%
\>[12]{}(\Varid{sock}\mathbin{::}\Conid{Var})\to \Conid{SocketAddress}\to {}\<[E]%
\\
\>[12]{}\Conid{ST}\;\Conid{IO}\;(\Conid{Either}\;()\;())\;[\mskip1.5mu \Varid{sock}\mathbin{:::}\Conid{Socket}\;\Conid{Unbound}:↦(\Conid{Socket}\;\Conid{Bound}\mathbin{`\Varid{or}`}\Conid{Socket}\;\Conid{Unbound})\mskip1.5mu]{}\<[E]%
\\
\>[3]{}\mbox{\onelinecomment  In \HaskeLL{}, by contrast, the typestate is part of the return type}{}\<[E]%
\\
\>[3]{}\Varid{bind}\mathbin{::}\Conid{Socket}\;\Conid{Unbound}\mathbin{⊸}\Conid{SocketAddress}\to \Conid{Either}\;(\Conid{Socket}\;\Conid{Bound})\;(\Conid{Socket}\;\Conid{Unbound}){}\<[E]%
\ColumnHook
\end{hscode}\resethooks
The support for dependent types in \textsc{ghc} is not as
comprehensive as Idris's. But it is conceivable to implement such an
indexed monad transformer in Haskell. However, this is not an easy task,
and we can anticipate that the error messages would be hard to stomach.

\section{Future work}

\subsection{Controlling program optimisations}
\label{sec:fusion}
Inlining is a cornerstone of program optimisation, exposing
opportunities for many program transformations. Yet not every function can
be inlined without negative effects on performance: inlining a
function with more than one use sites of the argument may result in
duplicating a computation. For example one should avoid the following
reduction: \ensuremath{(\lambda \Varid{x}\to \Varid{x}\plus \Varid{x})\;\Varid{expensive}\mathbin{⟶}\Varid{expensive}\plus \Varid{expensive}}.

Many compilers can discover safe inlining opportunities by analysing
source code and determine how many times functions use their
arguments.  (In \textsc{ghc} it is called the cardinality
analysis~\cite{sergey_cardinality_2014}). A limitation of such an
analysis is that it is necessarily heuristic (the problem is
undecidable for Haskell). Because inlining is crucial to efficiency, programmers
find themselves in the uncomfortable position of relying on a
heuristic to obtain efficient programs. Consequently, a small, seemingly
innocuous change can prevent a critical inlining opportunity and have
rippling catastrophic effects throughout the program.
Such unpredictable behaviour justifies the folklore that high-level languages should be
abandoned to gain precise control over program efficiency.

A remedy is to use the multiplicity annotations of \calc{} as
cardinality \emph{declarations}. Formalising and implementing the
integration of multiplicities in the cardinality analysis is
left as future work.

\subsection{Extending multiplicities}
\label{sec:extending-multiplicities}

For the sake of this article, we use only multiplicities
$1$ and $ω$.  But in fact \calc{} can readily be extended to
more, following \citet{ghica_bounded_2014} and
\citet{mcbride_rig_2016}. The general setting for \calc{} is an
ordered-semiring of multiplicities (with a join operation for type
inference).  In particular, in order to support dependent types,
\citeauthor{mcbride_rig_2016} needs a $0$ multiplicity.   We may also want to add a
multiplicity for affine arguments (\ie  arguments which can be
used \emph{at most once}).

The typing rules are mostly unchanged with the \emph{caveat} that
$\mathsf{case}_π$ must exclude $π=0$ (in particular we see that we
cannot substitute multiplicity variables by $0$). The variable rule becomes:
$$
\inferrule{ x :_1 A \leqslant Γ }{Γ ⊢ x : A}
$$
Where the order on contexts is the point-wise extension of the order
on multiplicities.

In \fref{sec:uniqueness}, we have considered the notion of
\emph{borrowing}: delimiting life-time without restricting to linear
usage. This seems to be a useful pattern, and we believe it can be
encoded as an additional multiplicity as follows: let $β$ be an
additional multiplicity with the following characteristics:
\begin{itemize}
\item $1 < β < ω$
\item $β+β = 1+β = 0+β = 1+1 = β$
\end{itemize}
That is, $β$ supports contraction and weakening but is smaller than
$ω$. We can then introduce a value with an explicit lifetime with the
following pattern
\begin{hscode}\SaveRestoreHook
\column{B}{@{}>{\hspre}l<{\hspost}@{}}%
\column{3}{@{}>{\hspre}l<{\hspost}@{}}%
\column{E}{@{}>{\hspre}l<{\hspost}@{}}%
\>[3]{}\Varid{borrow}\mathbin{::}(\Conid{T}\to _{\Varid{β}}\;\Conid{Unrestricted}\;\Varid{a})\to _{\Varid{β}}\;\Conid{Unrestricted}\;\Varid{a}{}\<[E]%
\ColumnHook
\end{hscode}\resethooks
The \ensuremath{\Varid{borrow}} function makes the life-time manifest in the structure of
the program. In particular, it is clear that calls within the argument
of \ensuremath{\Varid{borrow}} are not tail: a shortcoming of borrowing that we mentioned
in \fref{sec:uniqueness}.

\subsection{Future industrial applications}
\label{sec:industry}
Our own work in an industrial context triggered our efforts to add
linear types to \textsc{ghc}. We were originally motivated by
precisely typed protocols for complex interactions and by taming
\textsc{gc} latencies in distributed systems. But we have since
noticed other potential applications of linearity in a variety of
other industrial projects.

\begin{description}

\item[Streaming I/O] Program inputs and outputs are frequently much
  larger than the available \textsc{ram} on any single node. Rather than
  building complex pipelines with brittle explicit loops copying data
  piecemeal to spare our precious \textsc{ram}, one approach is to compose
  combinators that transform, split and merge data wholemeal but in
  a streaming fashion. These combinators manipulate first-class {\em
    streams} and guarantee bounded memory usage, as in the below
  infinitely running echo service:
  \begin{hscode}\SaveRestoreHook
\column{B}{@{}>{\hspre}l<{\hspost}@{}}%
\column{5}{@{}>{\hspre}l<{\hspost}@{}}%
\column{E}{@{}>{\hspre}l<{\hspost}@{}}%
\>[5]{}\Varid{receive}\mathbin{::}\Conid{Socket}\to \Conid{IOStream}\;\Conid{Msg}{}\<[E]%
\\
\>[5]{}\Varid{send}\mathbin{::}\Conid{Socket}\to \Conid{IOStream}\;\Conid{Msg}\to \Conid{IO}\;(){}\<[E]%
\\[\blanklineskip]%
\>[5]{}\Varid{echo}\;\Varid{isock}\;\Varid{osock}\mathrel{=}\Varid{send}\;\Varid{osock}\;(\Varid{receive}\;\Varid{isock}){}\<[E]%
\ColumnHook
\end{hscode}\resethooks
  However, reifying sequences of \ensuremath{\Conid{IO}} actions (socket reads) in this
  way runs the risk that effects might be duplicated inadvertently. In
  the above example, we wouldn't want to inadvertently hand over the
  receive stream to multiple consumers, or the abstraction of
  wholemeal I/O programming would be broken (like in \citet[Section
    2.2]{lippmeier_parallel_2016}), because neither consumer would
  ultimately see the same values from the stream. If say one consumer
  reads in the stream first, the second consumer would see an
  empty stream --- not what the first consumer saw.
  We have seen this very error several times in industrial projects,
  where the symptoms are bugs whose root cause are painful to track
  down. A linear type discipline would prevent such bugs.

\item[Programming foreign heaps] Complex projects with large teams
  invariably involve a mix of programming languages. Reusing legacy
  code is often much cheaper than reimplementing it. A key to
  successful interoperation between languages is performance. If all
  code lives in the same address space, then data need not be copied
  as it flows from function to function implemented in multiple
  programming languages. Trouble is, language A needs to tell language
  B what objects in language A's heap still have live references in
  the call stack of language B to avoid too eager garbage collection.

  For instance, users of \texttt{inline-java} call the \textsc{jvm}
  from Haskell via the \textsc{jni}. The \textsc{jvm} implicitly
  creates so-called \emph{local references} any time we request a Java
  object from the \textsc{jvm}. The references count as \textsc{gc}
  roots that prevent eager garbage collection. For performance, local
  references have a restricted scope: they are purely thread-local and
  never survive the call frame in which they were created. Both
  restrictions to their use can be enforced with linear types.

\item[Remote direct memory access]
  \Fref{sec:cursors} is an example of an \textsc{api} requiring
  destination-passing style. This style often appears
  in performance-sensitive contexts. One notable example from our
  industrial experience is \textsc{rdma} (Remote Direct Memory Access),
  which enables machines in high-performance clusters to copy data
  from the address space in one process to that of a remote process
  directly, bypassing the kernel and even the
  \textsc{cpu}, thereby avoiding any unneeded copy in the process.

  One could treat a remote memory location as a low-level resource, to
  be accessed using an imperative {\sc api}. Using linear types, one
  can instead treat it as a high-level value which can be written to
  directly (but exactly once). Using linear types the compiler can
  ensure that, as soon as the writing operation is complete, the
  destination computer is notified.
\end{description}

\section{Conclusion}

This article demonstrates how an existing lazy language, such
as Haskell, can be extended with linear types, without compromising
the language, in the sense that:
\begin{itemize}
\item existing programs are valid in the extended language
  \emph{without modification},
\item such programs retain the same operational semantics, and in particular
\item the performance of existing programs is not affected,
\item yet existing library functions can be reused to serve the
  objectives of resource-sensitive programs with simple changes to
  their types, and no code duplication.
\end{itemize}
In other words: regular Haskell comes first. Additionally, first-order
linearly typed functions and data structures are usable directly from
regular Haskell code. In such a setting their semantics is that of
the same code with linearity erased.

\HaskeLL{} was engineered as an unintrusive design, making it tractable
to integrate to an existing, mature compiler with a large ecosystem.
We have developed a prototype implementation extending
\textsc{ghc} with multiplicities. As we hoped, this
design integrates well in \textsc{ghc}.

Even though we change only \ghc's type system, we found that the
compiler and runtime already had the features necessary for unboxed,
off-heap, and in-place data structures.  That is, \ghc{} has the
low-level compiler primitives and \textsc{ffi} support to implement,
for example, mutable arrays, mutable cursors into serialised data, or
off-heap foreign data structures without garbage collection.  These
features could be used {\em before} this work, but their correct use
put some burden-of-proof on the programmers. Linearity unlocks these
capabilities for safe, compiler-checked use, within pure code.

\begin{acks}                            
  This work has received funding from the \grantsponsor{ERC}{European
    Commission}{https://erc.europa.eu/} through the \textsc{sage} project
  (Grant agreement
  no. \grantnum[http://www.sagestorage.eu]{ERC}{671500}), as well as
  from \grantsponsor{VR}{Swedish Research Council}{https://www.vr.se/}
  via the establishment of the Centre for Linguistic Theory and
  Studies in Probability (\textsc{clasp}) at the University of Gothenburg.
  We thank Manuel Chakravarty, Stephen Dolan and Peter Thiemann for their valuable
  feedback on early versions of this paper.
\end{acks}

\bibliography{../PaperTools/bibtex/jp,../local}{}

\ifx\longversion\undefined{}\else{
\clearpage
\appendix
\section{Semantics and soundness of \calc{}}
\label{appendix:dynamics}

In accordance with our stated goals in \fref{sec:introduction}, we are
interested in two key properties of our system: 1. that we can implement
a linear \textsc{api} with mutation under the hood, while exposing a
pure interface, and 2. that typestates are indeed enforced. This
appendix establishes these results in accordance to
\fref{sec:metatheory} where they were presented.

We introduce two dynamic semantics for \calc{}: a semantics with
mutation which models the implementation but blocks on incorrect
type-states, and a pure semantics, which we dub \emph{denotational} as
it represents the intended meaning of the program. We consider here array
primitives in the style of \fref{sec:freezing-arrays}, but we could extend to
any number of other examples such as the files of
\fref{sec:io-protocols}.

We prove the two semantics bisimilar, so that type-safety and progress
can be transported from the denotational semantics to the ordinary
semantics with mutation. The bisimilarity itself ensures that the
mutations are not observable and that the semantics is correct in exposing
a pure semantics. The progress result proves that typestates need not be
tracked dynamically.

\subsection{Preliminaries}
\label{sec:partial-derivations}

Our operational semantics are big-step semantics with laziness in the
style of \citet{launchbury_natural_1993}. In such semantics, sharing
is expressed by mutating the environment assigning value to
variables. Terms are transformed ahead of execution in order to
reflect the amount of sharing that a language like Haskell would
offer. In particular the arguments of an application are always
variables.

Following \citet{gunter_partial-big-step_1993}, however, we consider
not only standard big-step derivations but also partial derivations.
The reason to consider partial derivation is that they make it
possible to express properties such as \emph{progress}: all partial
derivations can be extended.

Given a number of rules defining $a⇓b$ with ordered premises (we will
use the ordering of premises shortly), we
define a \emph{total derivation} of $a⇓b$ as a tree in the standard
fashion. As usual $a⇓b$ holds if there is a total derivation for it.
A \emph{partial} derivation of $a⇓?$ (the question mark is part of the
syntax: the right-hand value is the result of the evaluation, it is
not yet known for a partial derivation!) is either:
\begin{itemize}
\item just a root labelled with $a⇓?$,
\item or an application of a rule matching $a$ where exactly one of
  the premises, $a'⇓?$ has a partial derivation, all the premises to
  the left of $a'⇓?$ have a total derivation, and the premises to the
  right of $a'⇓?$ are not known yet (since we would need to know the
  value $?$ to know what the root of the next premise is).
\end{itemize}

Remark that, by definition, in a partial derivation, there is exactly one
$b⇓?$ with no sub-derivation. Let us call $b$ the \emph{head} of the
partial derivation. And let us write $a⇓^*b$ for the relation which
holds when $b$ is the head of some partial derivation with root
$a⇓?$. We call $a⇓^*b$ the \emph{partial evaluation relation}, and, by
contrast, $a⇓b$ is sometimes referred to as the the \emph{complete
  evaluation relation}.

\subsection{Ordinary semantics}

Our semantics, which we often call \emph{ordinary} to constrast it
with the denotational semantics of \fref{sec:denotational}, follows
closely the semantics of \citet{launchbury_natural_1993}. The main
differences are that we keep the type annotation, and that we have
primitives for proper mutation.

Mixing mutation and laziness is not usual, as the unspecified
evaluation order of lazy languages would make mutation order
unpredicable, hence programs non-deterministic. It is our goal to show
that the linear typing discipline ensures that, despite the mutations,
the evaluation is pure.

Just like \citet{launchbury_natural_1993} does, we constrain the terms
to be explicit about sharing, before any evaluation takes
place. \Fref{fig:launchbury:syntax} shows the translation of abtrary
term to terms in the explicit-sharing form.

\begin{figure}
  \figuresection{Translation of typed terms}
  \begin{align*}
    (λ_π(x{:}A). t)^* &= λ_π(x{:}A). (t)^* \\
    x^*             &= x \\
    (t  x )^*     &= (t)^*  x \\
    (t  u )^*     &= \flet[π] y : A = (u)^* \fin (t)^*  y &
    \text{with $Γ⊢ t : A →_π B$}
  \end{align*}
  \begin{align*}
    c_k  t₁ … t_n   &= \flet x₁ :_{π_1} A_1 = (t₁)^*,…, x_n :_{π_n} A_n = (t_n)^*
                      \fin c_k x₁ … x_n & \text{with $c_k : A_1
                                          →_{π_1}…A_n →_{π_n}D$}
  \end{align*}
  \begin{align*}
    (\case[π] t {c_k  x₁ … x_{n_k} → u_k})^* &= \case[π] {(t)^*} {c_k  x₁ … x_{n_k} → (u_k)^*} \\
    (\flet[π] x_1: A_1= t₁  …  x_n : A_n = t_n \fin u)^* & = \flet[π] x₁:A_1 = (t₁)^*,…, x_n:A_n= (t_n)^* \fin (u)^*
  \end{align*}

  \caption{Syntax for the Launchbury-style semantics}
  \label{fig:launchbury:syntax}
\end{figure}

The evaluation relation is of the form $Γ : e ⇓ Δ : z$ where $e$ is an
expression, $z$ a value $Γ$ and $Δ$ are \emph{environments} with
bindings of the form $x :_ω A = e$ assigning the expression $e$ the the
variable $x$ of type $A$. Compared to the pure semantic of
\citeauthor{launchbury_natural_1993}, we have one additional kind of
values, $l$ for names of arrays. Array names are given semantics by
additional bindings in environments which we write, suggestively, $l
:_1 A = arr$. The $1$ is here to remind us that arrays cannot be used
arbitrarily, however, it does not mean they are always used in a
linear fashion: frozen arrays are not necessarily linear, but they
still appear as array names.

The details of the ordinary evaluation relation are given in
\fref{fig:dynamics}. Let us describe the noteworthy rules:
\begin{description}
\item[mutable cell] array names are values, hence are not
  reduced. In that they differ from variables.
\item[newMArray] allocates a fresh array of the given size (we write
  $i$ for an integer value). Note that
  the value of $a$ is not evaluated: an array in the environment is
  a concrete list of (not necessarily distinct) variables.
\item[writeArray] Mutates its array argument
\item[freezeArray] Mutates \emph{the type} of its argument to \ensuremath{\Conid{Array}}
  before wrapping it in $\varid{Unrestricted}$, so that we cannot call
  $\varid{write}$ on it anymore: $\varid{write}$ would block because
  the type of $l$ is not \ensuremath{\Conid{MArray}}. Of course, in an implementation
  this would not be checked because progress ensures that the case
  never arises.
\end{description}

\begin{figure}
  \begin{mathpar}
    \inferrule{ }{Γ : λp. t ⇓ Γ : λp. t}\text{m.abs}

    \inferrule{Γ : e ⇓ Δ : λp.e' \\ Δ : e'[π/q] ⇓ Θ : z} {Γ :
      e π ⇓ Θ : z} \text{m.app}

    \inferrule{ }{Γ : λ_π(x{:}A). e ⇓ Γ : λ_π(x{:}A). e}\text{abs}

    \inferrule{Γ : e ⇓ Δ : λ_π(y{:}A).e' \\ Δ : e'[x/y] ⇓ Θ : z} {Γ :
      e x ⇓ Θ : z} \text{application}

    \inferrule{Γ : e ⇓ Δ : z}{(Γ,x :_ω A = e) : x ⇓ (Δ;x :_ω A = z) :
      z}\text{variable}

    \inferrule{ }
    {(Γ,l :_1 A = arr) : l ⇓ (Γ, l :_1 A = arr) : l}\text{mutable cell}

    \inferrule{(Γ,x_1 :_ω A_1 = e_1,…,x_n :_ω A_n e_n) : e ⇓ Δ : z}
    {Γ : \flet[π] x₁ : A_1 = e₁ … x_n : A_n = e_n \fin e ⇓ Δ :
      z}\text{let}

    \inferrule{ }{Γ : c  x₁ … x_n ⇓ Γ : c  x₁ …
      x_n}\text{constructor}

    \inferrule{Γ: e ⇓ Δ : c_k  x₁ … x_n \\ Δ : e_k[x_i/y_i] ⇓ Θ : z}
    {Γ : \case[π] e {c_k  y₁ … y_n ↦ e_k } ⇓ Θ : z}\text{case}


    \inferrule
    {Γ:n ⇓ Δ:i \\ (Δ, l :_1 \varid{MArray}~a = [a,…,a]) : \flet[1] x = l \fin f~x ⇓ Θ : \varid{Unrestricted}~x}
    {Γ : \varid{newMArray}~n~a~f ⇓ Θ : \varid{Unrestricted}~x}\text{newMArray}

    \inferrule{Γ:n⇓Δ:i \\ Δ:arr ⇓ (Θ,l:_1 \varid{MArray}~a = [a_1,…,a_i,…,a_n]):l}
    {Γ : \varid{write}~arr~n~a ⇓ Θ,l :_1 \varid{MArray}~a =
      [a_1,…,a,…,a_n] : l}\text{write}

    \inferrule{Γ:arr ⇓ (Δ,l :_1 \varid{MArray}~a = [a_1,…,a_n]):l}
    { Γ : \varid{freeze}~arr ⇓ (Δ,l :_1 \varid{Array}~a = [a_1,…,a_n],
      x :_ω \varid{Array}~a = l) :
      \varid{Unrestricted}~x}\text{freeze}

    \inferrule
    {Γ : n ⇓ Δ : i \\ Δ:arr ⇓ (Θ,l :_1 \varid{Array}~a =
      [a_1,…,a_i,…,a_n]) : l \\ (Θ,l :_1 \varid{Array}~a =
      [a_1,…,a_i,…,a_n]) : a_i ⇓ Λ : z}
    {Γ : \varid{index}~arr~n ⇓ Λ : z}

  \end{mathpar}

  \caption{Ordinary dynamic semantics}
  \label{fig:dynamics}
\end{figure}

\subsection{Denotational semantics}
\label{sec:denotational}

\begin{figure}
  \begin{mathpar}
\inferrule{ }{Ξ ⊢ (Γ | λp. t ⇓ Γ | λp. t) :_ρ A, Σ}\text{m.abs}

\inferrule{Ξ ⊢ (Γ | e ⇓ Δ | λp.e') :_ρ A, Σ \\ Ξ ⊢ (Δ | e'[π/q] ⇓ Θ | z) :_ρ A, Σ} {(Γ :
  e π ⇓ Θ : z) :_ρ A, Σ} \text{m.app}

\inferrule{ }{Ξ ⊢ (Γ | λ_π(x{:}A). e ⇓ Γ | λ_π (x{:}A). e) :_ρ A→_π B, Σ}\text{abs}

\inferrule
    {Ξ  ⊢  (Γ|e      ⇓ Δ|λ(y:_π A).u):_ρ A →_π B, x:_{πρ} A, Σ \\
     Ξ  ⊢  (Δ|u[x/y] ⇓ Θ|z)   :_ρ       B,            Σ}
    {Ξ  ⊢  (Γ|e x ⇓ Θ|z) :_ρ B ,Σ}
{\text{app}}

\inferrule
  {Ξ, x:_ωB ⊢ (Γ|e ⇓ Δ|z) :_ρ A, Σ}
  {Ξ ⊢ (Γ,x :_ω B = e | x  ⇓ Δ, x :_ω B = z | z) :_ρ A, Σ}
{\text{shared variable}}

\inferrule
  {Ξ ⊢ (Γ|e ⇓ Δ|z) :_1 A, Σ}
  {Ξ ⊢ (Γ,x :_1 B = e| x  ⇓  Δ|z) :_1 A,  Σ}
{\text{linear variable}}

\inferrule
  {Ξ ⊢ (Γ,       x_1 :_{ρπ} A_1 = e_1 … x_n :_{ρπ} A_n = e_n  |  t ⇓ Δ|z) :_ρ C, Σ}
  {Ξ ⊢ (Γ|\flet[π] x_1 :  A_1 = e_1 … x_n : A_n = e_n \fin t ⇓ Δ|z) :_ρ C, Σ}
{\text{let}}

\inferrule
  { }
  {Ξ ⊢ (Γ | c x_1…x_n  ⇓ Γ | c x_1…x_n) :_ρ A, Σ}
{\text{constructor}}

\inferrule
  {Ξ,y_1:_{π_1qρ} A_1 … ,y_n:_{π_nqρ} A_n ⊢ (Γ|e ⇓ Δ|c_k x_1…x_n) :_{πρ} D, u_k:_ρ C, Σ \\
    Ξ ⊢ (Δ|u_k[x_i/y_i] ⇓ Θ|z) :_ρ C, Σ}
  {Ξ ⊢ (Γ|\case[π] e {c_k y_1…y_n ↦ u_k} ⇓ Θ|z) :_ρ C, Σ}
  {\text{case}}


\inferrule
{Ξ ⊢ (Γ|n ⇓ Δ|i), \varid{Int}, (arr:_ρ \varid{MArray}~a, Σ) \\ Ξ ⊢
  (Δ|\flet[1] x = [a,…,a] \fin f~x) ⇓ Θ|\varid{Unrestricted}~x) :_1 \varid{Unrestricted}~B, Σ}
{Ξ ⊢ (Γ|\varid{newMArray}~n~a~f ⇓ Θ|\varid{Unrestricted}~x) :_ρ \varid{Unrestricted}~B, Σ}\text{newMArray}

\inferrule
{Ξ ⊢ (Γ|n⇓Δ|i) :_ρ \varid{Int}, (arr:_ρ \varid{MArray}~a, Σ) \\ Ξ ⊢ (Δ|arr⇓Θ|[a_1,…,a_i,…,a_n]) :_ρ
  \varid{MArray}~a, Σ }
{Ξ ⊢ (Γ|\varid{write}~arr~n~a
  ⇓ Γ|[a_1,…,a,…,a_n]) :_ρ \varid{MArray}~a, Σ}\text{write}

\inferrule
{Ξ ⊢ (Γ|arr ⇓ Δ|[a_1,…,a_n]) :_ρ \varid{MArray}~a, Σ}
{Ξ ⊢ (Γ|\varid{freeze}~arr ⇓ Δ,x :_1 \varid{Array a} = [a_1,…,a_n]|\varid{Unrestricted}~x
  ) :_ρ \varid{Unrestricted} (\varid{Array}~a), Σ}\text{freeze}

\inferrule
{Ξ ⊢ (Γ | n ⇓ Δ | i) :_ρ \varid{Int},Σ \\ Ξ ⊢ (Δ|arr ⇓ Θ |
  [a_1,…,a_i,…,a_n])) :_ρ \varid{Array}~a,Σ \\ Ξ ⊢ (Θ | a_i ⇓ Λ | z)
:_ρ A, Σ}
{Ξ ⊢ (Γ | \varid{index}~arr~n ⇓ Λ | z) :_ρ a, Σ}

  \end{mathpar}
  \caption{Denotational dynamic semantics}
  \label{fig:typed-semop}
\end{figure}

The ordinary semantics, if a good model of what we are implementing in
practice, is not very convenient to reason about directly. First
mutation is a complication in itself, but it is also difficult to
recover types (or even multiplicity annotations) from a particular
evaluation state.

We address this difficulty by introducing a denotational semantics where the states
are annotated with types and linearity. The denotational semantics
also does not perform mutations: arrays are seen as ordinary values
which we modify by copy.

\begin{definition}[Annotated state]

  An annotated state is a tuple $Ξ ⊢ (Γ|t :_ρ A),Σ$ where
  \begin{itemize}
  \item $Ξ$ is a typing context
  \item $Γ$ is a \emph{typed environment}, \ie a collection of
    bindings of the form $x :_ρ A = e$
  \item $t$ is a term
  \item $ρ∈\{1,ω\}$ is a multiplicity
  \item $A$ is a type
  \item $Σ$ is a typed stack, \ie a list of triple $e:_ρ A$ of
    a term, a multiplicity and an annotation.
  \end{itemize}

  Terms are extended with array expressions: $[a_1, …, a_n]$ where the
  $a_i$ are variables.

\end{definition}

Let us introduce a notation which will be needed in the definition of
well-typed state.
\begin{definition}[Weighted pairs]
  We define a type of left-weighted pairs:

  $$\data a~{}_π\!⊗ a = ({}_π\!,) : a →_π b ⊸ a~{}_π\!⊗ b$$

  Let us remark that
  \begin{itemize}
  \item We have not introduced type parameters in datatypes, but it
    is straightforward to do so
  \item We annotate the data constructor with the multiplicity $π$,
    which is not mandated by the syntax. It will make things simpler.
  \end{itemize}

\end{definition}

Weighted pairs are used to internalise a notion of stack that keeps
track of multiplicities in the following definition, which defines
when annotated states are well-typed.

\begin{definition}[Well-typed state]
  We say that an annotated state is well-typed if the following
  typing judgement holds:
  $$
  Ξ ⊢ \flet Γ \fin (t,\termsOf{Σ}) : (A~{}_ρ\!⊗\multiplicatedTypes{Σ})‌
  $$
  Where $\flet Γ \fin e$ stands for the use of $Γ$ as a sequence of
  let-bindings with the appropriate multiplicity\footnote{We skip over
    the case of mutually recursive bindings in our presentation. But
    we can easily extend the formalism with then. Recursive bindings
    must be of multiplicity $ω$, and mutually recursive definition are
    part of a single $\flet$ block. When defining $\flet Γ$ we need to
    pull a mutually recursive block as a single $\flet$ block as
    well.}, $\termsOf{e_1 :_{ρ_1} A_1, … , e_n :_{ρ_n} A_n}$ for
  $(e_1~{}_{ρ_1}\!, (…, (e_n~{}_{ρ_n},())))$, and
  $\multiplicatedTypes{e_1 :_{ρ_1} A_1, … , e_n :_{ρ_n} A_n}$ for
  $A_1~{}_{ρ_1}\!⊗(…(A_n~{}_{ρ_n}\!⊗()))$.
\end{definition}

\begin{definition}[Denotational reduction relation]
  We define the denotational reduction relation, also written $⇓$, as a
  relation on annotated states. Because $Ξ$, $ρ$, $A$ and $Σ$ are
  always the same for related states, we abbreviate
  $$
  (Ξ ⊢ Γ|t :_ρ A,Σ) ⇓ (Ξ ⊢ Δ|z :_ρ A,Σ)
  $$
  as
  $$
  Ξ ⊢ (Γ|t ⇓ Δ|z) :_ρ A, Σ
  $$

  The denotational reduction relation is defined inductively by the
  rules of \fref{fig:typed-semop}.

  A few rules of notice:
  \begin{description}
  \item[linear variable] linear variables are removed from the
    environment when they are evaluated: they are no longer accessible
    (if the state is well-typed)
  \item[let] even if we are evaluating a $\flet_1$m we may have to
    introduce a non-linear binding in the environemnt: if the value we
    are currently computing will be used as the argument of a
    non-linear function, the newly introduced variables may be forced
    several times (or not at all). An example is
    $\flet[ω] x = (\flet[1] y = \varid{True}) \fin y \fin (x,x)$: if evaluating
    this example yielded the binding $y :_1 Bool = True$, then the
    intermediate state would be ill-typed. So for the sake of proofs,
    instead we add $y :_ω Bool = True$ to the environment
  \item[write] No mutation is performed in array write: we just return
    a new copy of the array.
  \end{description}

\end{definition}

The denotation semantics preserves the well-typedness of annotated
states throughout the evaluation\fref{thm:type-safety}. From then on, we
will only consider the evaluation of well-typed states.

{
\renewcommand{\thetheorem}{\ref{thm:type-safety}}
\begin{theorem}[Type preservation]
  If  $Ξ ⊢ (Γ|t ⇓ Δ|z) :_ρ A, Σ$, or $Ξ ⊢ (Γ|t ⇓^* Δ|z) :_ρ A, Σ$ then
  $$
  Ξ ⊢ (Γ|t :_ρ A),Σ \text{\quad{}implies\quad{}} Ξ ⊢ (Δ|z :_ρ A),Σ.
  $$
\end{theorem}
\addtocounter{theorem}{-1}
}
\begin{proof}
  By induction on the typed-reduction.

  The case of the linear variable rule is interesting, as it uses the
  fact that, by the constructor rule, $x :_1 B$ can only be used in
  the typing of the variable $x$, it is absent from the context when
  type-checking $Σ$. In particular note how the rule \emph{must}
  remove $x$ from the environment in order to preserve typing.
\end{proof}

{
\renewcommand{\thetheorem}{\ref{thm:progress-denotational}}
\begin{theorem}[Progress]
  Evaluation does not block. That is, for any partial derivation of
  $Ξ ⊢ (Γ'|e ⇓ ?) :_ρ A,Σ$, the derivation can be
  extended.
\end{theorem}
\addtocounter{theorem}{-1}
}
\begin{proof}
  The proof of progress, for the denotational semantics, is almost
  entirely standard. The only unusual rule is the linear variable
  rule, in which there are two things of notice:
  \begin{itemize}
  \item The linear variable rule blocks if $ρ = ω$
  \item The linear variable rule removes the variable $x$ from the
    environment.
  \end{itemize}
  Therefore it suffices to show that the former case never arises. And
  that whenever a variable is evaluated, then it appears in the
  environment.
  \begin{itemize}
  \item Notice that $Ξ⊢Γ,x:_1 B = e|x :_ω A,Σ$ is not a well-typed
    state because it reduces to $x:_1B = x:_{ωπ} B$ for some $π$,
    which never holds. By type preservation (\fref{thm:type-safety}),
    $Ξ⊢Γ,x:_1 B = e|x :_ω A,Σ$ cannot be the head of a partial
    derivation.
  \item Similarly $Ξ⊢Γ|x :_1 A,Σ$ where $x∉Γ$ is not well-typed, and
    hence cannot be the head of a partial derivation\footnote{Notice
      that it is an invariant of the denotational evaluation, that
      variables in $Ξ$ are not reachable from $e$. This is only true
      because let-bindings are not recursive. In the case that they
      are recursive, the shared variable rule make it possible to run
      into a situation where $x$ is evaluated and part of $Ξ$, in
      which case the reduction blocks: this models so-called
      \emph{black-holing} in which ill-founded recursive lazy
      definitions report an error rather than looping. This
      presentation follows \citet{launchbury_natural_1993}, and in
      presence of such recursion, progress must be extended to say
      that partial derivation can be either extended or is in a black
      hole. An alternative solution is to change the shared variable
      rule to loop instead of blocking in case of such ill-founded
      recursion.}.
  \end{itemize}
\end{proof}

\subsection{Bisimilarity and all that}
\label{sec:bisimilarity}

The crux of our metatheory is that the two semantics are
bisimilar. Bisimilarity allows to tranport properties from the
denational semantics, on which it is easy to reason, and the ordinary
semantics which is close to the implementation. It also makes it
possible to prove observational equality results. Our first definition
is the relation between the states of the ordinary evaluation and
those of the denotational evaluation which witnesses the bisimulation.

\begin{definition}[Denotation assignment]
  A well-typed state is said to be a denotation assignment for an ordinary
  state, written $\ta{Γ:e}{Ξ ⊢ Γ' | e' :_ρ A , Σ}$, if
  $e[Γ_1]=e' ∧ Γ' = Γ''[Γ_1] ∧ Γ'' \leqslant Γ_ω$. Where
  \begin{itemize}
  \item $Γ_ω$ is the restriction of $Γ$ (a context of the ordinary
    semantics) to the variable bindings $x :_ω A = u$
  \item $Γ_1$ is the restriction of $Γ$ to the array bindings $l :_1 A
    = [a_1, …, a_n]$, seen as a substitution.
  \end{itemize}
  That is, $Γ'$ is allowed to strengthen some $ω$ bindings to be
  linear, and to drop unnecessary bindings. Array pointers are
  substituted with their value (since we have array pointers in the
  ordinary semantics but only array values in the denotational
  semantics). The substitution is subject to

  The substitution must abide by the following restrictions in order
  to preserve the invariant that \ensuremath{\Conid{MArray}} pointers are not shared:
  \begin{itemize}
  \item An \ensuremath{\Conid{MArray}} pointer in $Γ_1$ is substituted either exactly
    in one place in $Γ''$ or exactly in one place in $e$.
  \item If an \ensuremath{\Conid{MArray}} pointer is substituted in $Γ''$ then it is
    substituded in a linear binding $x :_1 A = u$
  \item If an \ensuremath{\Conid{MArray}} pointer is substituted in $e$ the $ρ=1$
  \item If an \ensuremath{\Conid{MArray}} pointer is substituted in the body $u$ as of
    $let_p x = u in v$ (sub)expression, the $p=1$
  \end{itemize}
\end{definition}

\begin{lemma}[Safety]\label{lem:actual_type_safety}
  The denotation assignment relation defines a simulation of the
  ordinary evaluation by the denotational evaluation, both in the
  complete and partial case.

  That is:
  \begin{itemize}
  \item for all $\ta{Γ:e}{Ξ ⊢ (Γ'|e) :_ρ A,Σ}$ such that $Γ:e⇓Δ:z$,
    there exists a well-typed state $Ξ ⊢ (Δ'|z) :_ρ A,Σ$ such that
    $Ξ ⊢ (Γ|t ⇓ Δ|z) :_ρ A, Σ$ and $\ta{Δ:z}{Ξ ⊢ (Δ'|z) :_ρ A,Σ}$.
  \item for all $\ta{Γ:e}{Ξ ⊢ (Γ'|e) :_ρ A,Σ}$ such that $Γ:e⇓^*Δ:z$,
    there exists a well-typed state $Ξ ⊢ (Δ'|z) :_ρ A,Σ$ such that
    $Ξ ⊢ (Γ|t ⇓^* Δ|z) :_ρ A, Σ$ and $\ta{Δ:z}{Ξ ⊢ (Δ'|z) :_ρ A,Σ}$.
  \end{itemize}
\end{lemma}
\begin{proof}
  Both simulations are proved by a similar induction on the derivation
  of $Γ:e⇓Δ:z$ (resp. $Γ:e⇓Δ:z$):
  \begin{itemize}
  \item The properties of the substitution of \ensuremath{\Conid{MArray}} in the
    definition of denotation assignments are crafted to make the
    \emph{variable} and \emph{let} rules carry through
  \item The other rules are straightforward
  \end{itemize}
\end{proof}

\begin{lemma}[Liveness]\label{lem:liveness}
  The refinement relation defines a simulation of the strengthened
  reduction by the ordinary reduction, both in the complete and
  partial case.

  That is:
  \begin{itemize}
  \item for all $\ta{Γ:e}{Ξ ⊢ (Γ'|e) :_ρ A,Σ}$ such that
    $Ξ ⊢ (Γ'|e' ⇓ Δ'|z') :_ρ A,Σ$, there exists a state $Δ:z$ such
    that $Γ:e'⇓Δ:z'$ and $\ta{Δ:z}{Ξ ⊢ (Δ'|z') :_ρ A,Σ}$.
  \item for all $\ta{Γ:e}{Ξ ⊢ (Γ'|e) :_ρ A,Σ}$ such that
    $Ξ ⊢ (Γ'|e ⇓^* Δ'|t) :_ρ A,Σ$, there exists a state $Δ:t'$ such
    that $Γ:e⇓^*Δ:z$ and $\ta{Δ:t}{Ξ ⊢ (Δ'|t') :_ρ A,Σ}$.
  \end{itemize}
\end{lemma}
\begin{proof}
  Both are proved by a straightforward induction over the derivation of
  $Ξ ⊢ (Γ'|e ⇓ Δ'|z) :_ρ A,Σ$ (resp. $Ξ ⊢ (Γ'|e ⇓ Δ'|z) :_ρ A,Σ$).
\end{proof}

Equipped with this bisimulation, we are ready to prove the soundness
properties of the ordinary semantics. We say that a state $Γ:e$ is
well-typed if there exists an annotated state $Ξ⊢Γ'|e':_ρ A,Σ$, such
that $\ta{Γ:e}{Ξ⊢Γ'|e':_ρ A,Σ}$.

{
\renewcommand{\thetheorem}{\ref{thm:type-preservation}}
\begin{theorem}[Type preservation]
  For any well-typed $Γ:e$, if $Γ:e⇓Δ:t$ or $Γ:e⇓^*Δ:t$, then $Δ:t$ is
  well-typed.
\end{theorem}
\addtocounter{theorem}{-1}
}
\begin{proof}
  This is precisely the same statement as \fref{lem:actual_type_safety}
\end{proof}
{
\renewcommand{\thetheorem}{\ref{thm:progress}}
\begin{theorem}[Progress]
  Evaluation does not block. That is, for any partial derivation of
  $Γ:e⇓?$, for $Γ:e$ well-typed, the derivation can be
  extended.

  In particular, typestates need not be checked dynamically.
\end{theorem}
\addtocounter{theorem}{-1}
}
\begin{proof}
  By liveness (\fref{lem:liveness}) it is sufficient to prove the case
  of the denotational semantics. Which follows from \fref{thm:progress-denotational}
\end{proof}

Observational equivalence, which means, for \calc{}, that an
implementation in terms of in-place mutation is indistinguishable from
a pure implementation, is phrased in terms of the \ensuremath{\Conid{Bool}}: any
distinction which we can make between two evaulations can be
extended so that one evaluates to \ensuremath{\Conid{False}} and the other to \ensuremath{\Conid{True}}.
{
\renewcommand{\thetheorem}{\ref{thm:obs-equiv}}
\begin{theorem}[Observational equivalence]
  The ordinary semantics, with in-place mutation is observationally equivalent
  to the pure denotational semantics.

  That is, for all $\ta{⋅:e}{⊢ (⋅|e) :_ρ \varid{Bool},⋅}$, if $⋅:e ⇓ Δ:z$ and
  $⋅ ⊢ (⋅|e⇓ Δ|z')  :_ρ \varid{Bool}, ⋅ $, then $z=z'$
\end{theorem}
\addtocounter{theorem}{-1}
}
\begin{proof}
  Because the semantics are deterministic, this is a direct
  consequence of bisimilarity.
\end{proof}
}
\fi

\end{document}